\documentclass[acmsmall,screen]{acmart}

\usepackage[noend]{algpseudocode}
\usepackage[prologue,usenames,dvipsnames,svgnames]{xcolor}
\usepackage[subrefformat=parens]{subcaption}

\usepackage{algorithm}
\usepackage{amsfonts}
\usepackage{amsmath}
\usepackage{amssymb}
\usepackage{amsthm}
\usepackage{bibentry}
\usepackage{booktabs}
\usepackage{enumitem}
\usepackage{framed}
\usepackage{graphicx}
\usepackage{graphicx}
\usepackage{lipsum}
\usepackage{listings}
\usepackage{mathtools}
\usepackage{multirow}
\usepackage{proof}
\usepackage{siunitx}
\usepackage{soul}
\usepackage{stmaryrd}
\usepackage{tabularx}
\usepackage{textcomp}
\usepackage{tikz}
\usepackage{wrapfig}

\usepackage{kbordermatrix}
\usepackage{xr}

\lstdefinestyle{venturescript}{
  basicstyle=\ttfamily\scriptsize,
  columns=fullflexible,
  keepspaces=true,
  upquote=true,
  alsoletter={\.,\%,\#, \?, \/, \$},
  morekeywords=[1]{assume, if, else, define, observe, for_each, infer, cond,
    return, mem, mapM, run, mapv, sample, sample_all, \$,
    pair, simplex, categorical, minimal_subproblem, predict, arange, quote,
    plot, dict, mh, sample, sample_all, resimulate, resample, one,},
  keywordstyle=[1]\bfseries\textcolor{Brown},
  morekeywords=[2]{\#structure, \#hypers, /?structure, /?hypers},
  keywordstyle=[2]\bfseries\textcolor{OliveGreen},
  morekeywords=[3]{},
  keywordstyle=[3]\bfseries\textcolor{OliveGreen},
  showstringspaces=False,
  stringstyle=\ttfamily\color{NavyBlue},
  morestring=[b]{"},
  morestring=[b]{'},
  morecomment=[l]{//},
  commentstyle=\bfseries\color{black}\ttfamily,
}

\lstdefinestyle{dsledits}{
    language={[LaTeX]TeX},
    basicstyle=\ttfamily\footnotesize,
    escapeinside={(*@}{@*)},
    numbers=none,
    numberstyle=\normalsize,
    stepnumber=1,
    numbersep=3pt,
    breaklines=true,
    framesep=3pt,
    showstringspaces=false,
    keywordstyle=\itshape\color{blue},
    stringstyle=\color{maroon},
    commentstyle=\color{black},
    rulecolor=\color{black},
    xleftmargin=0pt,
    xrightmargin=0pt,
    aboveskip=\medskipamount,
    belowskip=\medskipamount,
    backgroundcolor=\color{white},
}

\algnewcommand{\LineComment}[1]{\State \(\triangleright\) #1}

\usetikzlibrary{calc}
\usetikzlibrary{decorations.pathreplacing}
\usetikzlibrary{positioning}
\usetikzlibrary{shadings}
\usetikzlibrary{shadows}
\usetikzlibrary{shapes.arrows}
\usetikzlibrary{shapes}

\newcommand{\abs}[1]{\lvert#1\rvert}
\newcommand{\set}[1]{\left\lbrace#1\right\rbrace}

\newcommand{\ttpl}{\texttt{(}}
\newcommand{\ttpr}{\texttt{)}}

\newcommand{\Denotation}[2][]{#1\left\llbracket#2\right\rrbracket}
\newcommand{\Prior}{\mathsf{Prior}}
\newcommand{\Likelihood}{\mathsf{Lik}}
\newcommand{\Posterior}{\mathsf{Post}}
\newcommand{\ApproxPosterior}{\mathsf{ApproxPost}}
\newcommand{\Subexpr}{\mathsf{Subexpr}}
\newcommand{\Sever}{\mathsf{Sever}}
\newcommand{\GenDist}{\mathsf{Expand}}
\newcommand{\HasProperty}{\mathsf{HasProperty}}
\newcommand{\Venture}{\mathsf{Venture}}

\newcommand{\SeverTo}{\xrightarrow[\mathrm{sever}]{}}
\newcommand{\TransOper}{\mathcal{T}}
\newcommand{\TransOperX}[1]{\mathcal{T}(X, #1)}
\newcommand{\TransOperXa}[2]{\mathcal{T}(X, #1; #2)}

\newcommand{\aroot}{a_{\mathsf{root}}}
\newcommand{\Ehole}{E_{\mathsf{hole}}}
\newcommand{\Esev}{E_{\mathsf{sev}}}
\newcommand{\Esub}{E_{\mathsf{sub}}}
\newcommand{\paccept}{p_{\mathsf{accept}}}

\newcommand{\post}[1]{\Denotation[\Posterior]{#1}(X)}
\newcommand{\prior}[1]{\Denotation[\Prior]{#1}}
\newcommand{\lik}[1]{\Denotation[\Likelihood]{#1}(X)}
\newcommand{\subexpr}[2]{\Denotation[\Subexpr_{#2}]{#1}}
\newcommand{\sever}[2]{\Denotation[\Sever_{#2}]{#1}}
\newcommand{\gendist}[2]{\Denotation[\GenDist]{#1}(#2)}

\newcommand{\Covar}{\mathsf{Cov}}
\newcommand{\VKernel}{\mathsf{VentureCov}}
\newcommand{\VProg}{\mathsf{VentureProg}}

\newcommand{\DomReals}{\mathsf{Numeric}}

\newcommand{\DomKernel}{\mathsf{Kernel}}
\newcommand{\DomParameters}{\mathsf{Parameters}}

\newcommand{\DomPartition}{\mathsf{Partition}}
\newcommand{\DomBlock}{\mathsf{Block}}
\newcommand{\DomCluster}{\mathsf{Cluster}}
\newcommand{\DomVariable}{\mathsf{Variable}}
\newcommand{\DomDist}{\mathsf{Dist}}

\newcommand{\Lang}{\mathcal{L}}

\newcommand{\Reals}{\mathbb{R}}

\newcommand{\DataSpace}{\mathcal{X}}
\newcommand{\tik}{T_{ik}}
\newcommand{\rik}{R_{ik}}
\newcommand{\hik}{h_{ik}}

\theoremstyle{acmplain}
\newtheorem{theorem}{Theorem}[section]

\newtheorem{lemma}[theorem]{Lemma}

\theoremstyle{acmdefinition}
\newtheorem{condition}[theorem]{Condition}
\newtheorem{objective}[theorem]{Objective}
\newtheorem{definition}[theorem]{Definition}

\algdef{SE}[DOWHILE]{Do}{doWhile}{\algorithmicdo}[1]{\algorithmicwhile\ #1}%

\allowdisplaybreaks

\setcopyright{rightsretained}
\acmPrice{}
\acmDOI{10.1145/3290350}
\acmYear{2019}
\copyrightyear{2019}
\acmJournal{PACMPL}
\acmVolume{3}
\acmNumber{POPL}
\acmArticle{37}
\acmMonth{1}

\setcitestyle{nosort}

\begin{document}

\title{Bayesian Synthesis of Probabilistic Programs for Automatic Data Modeling}
%





\author{Feras A.~Saad}
\orcid{nnnn-nnnn-nnnn-nnnn}
\affiliation{
  \department{Computer Science \& Artificial Intelligence Laboratory}
  \institution{Massachusetts Institute of Technology}
  \city{Cambridge}
  \state{MA}
  \postcode{02139}
  \country{USA}
}
\email{fsaad@mit.edu}

\author{Marco F.~Cusumano-Towner}
\orcid{nnnn-nnnn-nnnn-nnnn}
\affiliation{
  \department{Computer Science \& Artificial Intelligence Laboratory}
  \institution{Massachusetts Institute of Technology}
  \city{Cambridge}
  \state{MA}
  \postcode{02139}
  \country{USA}
}
\email{marcoct@mit.edu}

\author{Ulrich Schaechtle}
\orcid{nnnn-nnnn-nnnn-nnnn}
\affiliation{
  \department{Department of Brain \& Cognitive Sciences}
  \institution{Massachusetts Institute of Technology}
  \city{Cambridge}
  \state{MA}
  \postcode{02139}
  \country{USA}
}
\email{ulli@mit.edu}

\author{Martin C.~Rinard}
\orcid{nnnn-nnnn-nnnn-nnnn}
\affiliation{
  \department{Department of Electrical Engineering \& Computer Science}
  \institution{Massachusetts Institute of Technology}
  \city{Cambridge}
  \state{MA}
  \postcode{02139}
  \country{USA}
}
\email{rinard@csail.mit.edu}

\author{Vikash K.~Mansinghka}
\orcid{nnnn-nnnn-nnnn-nnnn}
\affiliation{
  \department{Department of Brain \& Cognitive Sciences}
  \institution{Massachusetts Institute of Technology}
  \city{Cambridge}
  \state{MA}
  \postcode{02139}
  \country{USA}
}
\email{vkm@mit.edu}

\renewcommand{\shortauthors}{Saad, Cusumano-Towner, Schaechtle, Rinard, and Mansinghka}

\authorsaddresses{}


\begin{abstract}
We present new techniques for automatically constructing probabilistic programs
for data analysis, interpretation, and prediction. These techniques work with
probabilistic domain-specific data modeling languages that capture key
properties of a broad class of data generating processes, using Bayesian
inference to synthesize probabilistic programs in these modeling languages given
observed data.
We provide a precise formulation of Bayesian synthesis for automatic data
modeling that identifies sufficient conditions for the resulting synthesis
procedure to be sound. We also derive a general class of synthesis algorithms
for domain-specific languages specified by probabilistic context-free grammars
and establish the soundness of our approach for these languages.
We apply the techniques to automatically synthesize probabilistic programs for
time series data and multivariate tabular data. We show how to analyze the
structure of the synthesized programs to compute, for key qualitative properties
of interest, the probability that the underlying data generating process
exhibits each of these properties. Second, we translate probabilistic programs
in the domain-specific language into probabilistic programs in Venture, a
general-purpose probabilistic programming system. The translated Venture
programs are then executed to obtain predictions of new time series data and new
multivariate data records.
Experimental results show that our techniques can accurately infer qualitative
structure in multiple real-world data sets and outperform standard data analysis
methods in forecasting and predicting new data.
\end{abstract}

\ccsdesc[300]{Mathematics of computing~Probability and statistics}
\ccsdesc[300]{Mathematics of computing~Bayesian computation}
\ccsdesc[300]{Mathematics of computing~Markov-chain Monte Carlo methods}
\ccsdesc[300]{Mathematics of computing~Time series analysis}
\ccsdesc[300]{Mathematics of computing~Multivariate statistics}
\ccsdesc[300]{Software and its engineering~Domain specific languages}
\ccsdesc[300]{Theory of computation~Probabilistic computation}
\ccsdesc[300]{Theory of computation~Bayesian analysis}
\ccsdesc[300]{Theory of computation~Data modeling}
\ccsdesc[300]{Theory of computation~Denotational semantics}
\ccsdesc[300]{Computing methodologies~Machine learning}

\keywords{probabilistic programming, Bayesian inference, synthesis, model discovery}

\maketitle


\section{Introduction}

Data analysis is an important and longstanding activity in many areas of natural
science, social science, and engineering~\citep{tukey1977,gelman2007}. Within
this field, probabilistic approaches that enable users to more accurately
analyze, interpret, and predict underlying phenomena behind their data are
rising in importance and prominence~\citep{murphy2012}.

A primary goal of modern data modeling techniques is to obtain an artifact
that explains the data. With current practice, this artifact takes the form of
either a set of parameters for a fixed model
\citep{nie1975,spiegelhalter1996,plummer2003,cody2005,hyndman2008,rasmussen2010,seabold2010,pedergosa2011,james2013}
or a set of parameters for a fixed probabilistic program structure
\citep{pfeffer2001,milch2007,goodman08,mccallum2009,wood14,goodman2014,carpenter2015,pfeffer2016,salvatier2016,tran2017,ge2018}.
With this approach, users manually iterate over multiple increasingly refined models
before obtaining satisfactory results. Drawbacks of this approach include the
need for users to manually select the model or program structure, the need for
significant modeling expertise, limited modeling capacity, and the potential for
missing important aspects of the data if users do not explore a wide enough
range of model or program structures.

In contrast to this current practice, we model the data with an ensemble of
probabilistic programs sampled from a joint space of program structures and
parameters. This approach eliminates the need for the user to select a specific
model structure, extends the range of data that can be easily and productively
analyzed, and (as our experimental results show) delivers artifacts that more
accurately model the observed data. A challenge is that this approach requires
substantially more sophisticated and effective probabilistic modeling and
inference techniques.

We meet this challenge by combining techniques from machine learning, program
synthesis, and programming language design. From machine learning we import the
Bayesian modeling framework and the Bayesian inference algorithms required to
express and sample from rich probabilistic structures. From program synthesis we
import the concept of representing knowledge as programs and searching program
spaces to derive solutions to technical problems. From programming language
design we import the idea of using domain-specific languages to precisely
capture families of computations. The entire framework rests on the foundation
of a precise formal semantics which enables a clear formulation of the synthesis
problem along with proofs that precisely characterize the soundness guarantees
that our inference algorithms deliver.

\subsection{Automated Data Modeling via Bayesian Synthesis of Probabilistic
Programs in Domain-Specific Languages}

We use Bayesian inference to synthesize ensembles of probabilistic programs
sampled from domain-specific languages given observed data. Each language is
designed to capture key properties of a broad class of data generating
processes. Probabilistic programs in the domain-specific language (DSL) provide
concise representations of probabilistic models that summarize the qualitative
and quantitative structure in the underlying data generating process. The
synthesized ensembles are then used for data analysis, interpretation, and
prediction.

We precisely formalize the problem of Bayesian synthesis of probabilistic
programs for automatic data modeling. To each expression in the domain-specific
language we assign two denotational semantics: one that corresponds to the prior
probability distribution over expressions in the DSL and another that
corresponds to the probability distribution that each DSL expression assigns to
new data. We provide sufficient conditions on these semantics needed for
Bayesian synthesis to be well-defined.
We outline a template for a broad class of synthesis algorithms and prove that
algorithms that conform to this template and satisfy certain preconditions are
sound, i.e., they converge asymptotically to the Bayesian posterior distribution
on programs given the data.

Our approach provides automation that is lacking in both statistical programming
environments and in probabilistic programming languages. Statistical programming
environments such as SAS~\citep{cody2005}, SPSS~\citep{nie1975} and
BUGS~\citep{spiegelhalter1996} require users to first choose a model family and
then to write code that estimates model parameters. Probabilistic programming
languages such as Stan~\citep{carpenter2015}, Figaro~\citep{pfeffer2016}, and
Edward~\citep{tran2017} provide constructs that make it easier to estimate
model parameters given data, but still require users to write probabilistic code
that explicitly specifies the underlying model structure. Our approach, in
contrast, automates model selection within the domain-specific languages.
Instead of delivering a single model or probabilistic program, it delivers
ensembles of synthesized probabilistic programs that, together, more accurately
model the observed data.

\subsection{Inferring Qualitative Structure by Processing Synthesized Programs}

In this paper we work with probabilistic programs for model families in two
domains: (i) analysis of univariate time series data via Gaussian
processes~\citep{rasmussen2006}; and (ii) analysis of multivariate tabular data
using nonparametric mixture models~\citep{mansinghka2016}.
The synthesized programs in the domain-specific language provide a compact model
of the data that make qualitative properties apparent in the surface syntax of
the program. We exploit this fact to develop simple program processing routines
that automatically extract these properties and present them to users, along
with a characterization of the uncertainty with which these properties are
actually present in the observed data.

For our analysis of time series data, we focus on the presence or absence of
basic temporal features such as linear trends, periodicity, and change points in
the time series. For our analysis of multivariate tabular data, we focus on
detecting the presence or absence of predictive relationships, which may be
characterized by nonlinear or multi-modal patterns.

\subsection{Predicting New Data by Translating Synthesized Programs to Venture}

In addition to capturing qualitative properties, we obtain executable versions
of the synthesized probabilistic programs specified in the DSL by translating
them into probabilistic programs specified in Venture \citep{mansinghka2014}.
This translation step produces executable Venture programs that define a
probability distribution over new, hypothetical data from a generative process
learned from the observed data and that can deliver predictions for new data.
Using Venture as the underlying prediction platform allows us to leverage the
full expressiveness of Venture's general-purpose probabilistic inference
machinery to obtain accurate predictions, as opposed to writing custom
prediction engines on a per-DSL basis.

For time series data analyzed with Gaussian processes, we use the Venture
programs to forecast the time series into the future. For multivariate tabular
data analyzed with mixture models, we use the programs to simulate
new data with similar characteristics to the observed data.

\subsection{Experimental Results}

We deploy the method to synthesize probabilistic programs and perform data
analysis tasks on several real-world datasets that contain rich probabilistic
structure.

For time series applications, we process the synthesized Gaussian process
programs to infer the probable existence or absence of temporal structures in
multiple econometric time series with widely-varying patterns. The synthesized
programs accurately report the existence of structures when they truly exist in
the data and the absence of structures when they do not. We also demonstrate
improved forecasting accuracy as compared to multiple widely-used statistical
and baselines, indicating that they capture patterns in the underlying data
generating processes.

For applications to multivariate tabular data, we show that synthesized programs report the
existence of predictive relationships with nonlinear and multi-modal characteristics
that are missed by standard pairwise correlation metrics. We also show that the
predictive distribution over new data from the probabilistic programs faithfully
represents the observed data as compared to simulations from generalized linear
models and produce probability density estimates that are orders of magnitude
more accurate than standard kernel density estimation techniques.

\subsection{Contributions}

This paper makes the following contributions:

\begin{itemize}[leftmargin=*]

\item {\bf Bayesian synthesis of probabilistic programs.} It introduces and
precisely formalizes the problem of Bayesian synthesis of probabilistic programs
for automatic data modeling in domain-specific data modeling languages.
It also provides sufficient conditions for Bayesian
synthesis to be well-defined. These conditions leverage the fact that
expressions in the DSL correspond to data models and are given two
denotational semantics: one that corresponds to the prior probability
distribution over expressions in the DSL and another that corresponds to the
probability distribution that each DSL expression assigns to new data.
It also defines a template for a broad class of synthesis algorithms based on
Markov chains and proves that algorithms which conform to this template and
satisfy certain preconditions are sound, i.e., they converge to
the Bayesian posterior distribution on programs given the data.

\item {\bf Languages defined by probabilistic context-free grammars.} It defines
a class of domain-specific data modeling languages defined by probabilistic
context-free grammars and identifies when these languages satisfy the sufficient
conditions for sound Bayesian synthesis. We also provide an algorithm for
Bayesian synthesis for this class of DSLs that conforms to the above template
and prove that it satisfies the soundness preconditions.

\item {\bf Example domain-specific languages for modeling time series and
multivariate data.} It introduces two domain-specific languages, each based on a
state-of-the-art model discovery technique from the literature on Bayesian
statistics and machine learning. These languages are suitable for modeling broad
classes of real-world datasets.

\item {\bf Applications to inferring qualitative structure and making
quantitative predictions.} It shows how to use the synthesized programs to (i)
infer the probability that qualitative structure of interest is present in the
data, by processing collections of programs produced by Bayesian synthesis, and
(ii) make quantitative predictions for new data.

\item {\bf Empirical results in multiple real-world domains.} It presents
empirical results for both inferring qualitative structure and predicting new
data. The results show that the qualitative structures match the real-world
data, for both domain specific languages, and can be more accurate than those
identified by standard techniques from statistics. The results also show that
the quantitative predictions are more accurate than multiple widely-used
baselines.

\end{itemize}


\section{Example}
\label{sec:example}


\begin{figure}[!htbp]
\centering

\begin{subfigure}[m]{.575\linewidth}
\begin{subfigure}{\linewidth}
\begin{lstlisting}[style=venturescript,frame=single]
// ** PRIOR OVER DSL SOURCE CODE **
assume get_hyper ~ mem((node) ~> {
  -log_logistic(log_odds_uniform() #hypers:node)
});
assume choose_primitive =  mem((node) ~> {
  base_kernel = uniform_discrete(0, 5) #structure:node;
  cond(
    (base_kernel == 0)(["WN",  get_hyper(pair("WN", node))]),
    (base_kernel == 1)(["C",   get_hyper(pair("C", node))]),
    (base_kernel == 2)(["LIN", get_hyper(pair("LIN", node))]),
    (base_kernel == 3)(["SE",  get_hyper(pair("SE", node))]),
    (base_kernel == 4)(["PER", get_hyper(pair("PER_l", node)),
                               get_hyper(pair("PER_t", node))]))
});
assume choose_operator = mem((node) ~> {
  operator_symbol ~ categorical(
    simplex(0.45, 0.45, 0.1), ["+", "*", "CP"])
    #structure:pair("operator", node);
  if (operator_symbol == "CP") {
    [operator_symbol, hyperprior(pair("CP", node)), .1]
  } else { operator_symbol }
});
assume generate_random_dsl_code = mem((node) ~> {
  cov = if (flip(.3) #structure:pair("branch", node)) {
    operator ~ choose_operator(node);
    [operator,
      generate_random_dsl_code(2 * node),
      generate_random_dsl_code((2 * node + 1))]
  } else { choose_primitive(node) };
  ["+", cov, ["WN", 0.01]]
});
assume dsl_source ~ generate_random_dsl_code(node:1);
\end{lstlisting}
\end{subfigure}

\begin{subfigure}[t]{\linewidth}
\begin{lstlisting}[style=venturescript,frame=single]
// ** TRANSLATING DSL CODE INTO VENTURE **
assume ppl_source ~ generate_venturescript_code(dsl_source);
assume gp_executable = venture_eval(ppl_source);
\end{lstlisting}
\end{subfigure}

\begin{subfigure}[t]{\linewidth}
\begin{lstlisting}[style=venturescript,frame=single]
// ** DATA OBSERVATION PROGRAM **
define xs = get_data_xs("./data.csv");
define ys = get_data_ys("./data.csv");
observe gp_executable(${xs}) = ys;
\end{lstlisting}
\end{subfigure}

\begin{subfigure}[t]{\linewidth}
\begin{lstlisting}[style=venturescript,frame=single]
// ** BAYESIAN SYNTHESIS PROGRAM **
resample(60);
for_each(arange(T), (_) -> {
  resimulate([|structure|], one, steps:100);
  resimulate([|hypers|], one, steps:100)})
\end{lstlisting}
\end{subfigure}
\begin{subfigure}[t]{\linewidth}
\begin{lstlisting}[style=venturescript,frame=single]
// ** PROCESSING SYNTHESIZED DSL CODE **
define count_kernels = (dsl, kernel) -> {
  if contains(["*", "+", "CP"], dsl[0]) {
    count_kernels(dsl[1], kernel) + count_kernels(dsl[2], kernel)
  } else { if (kernel == dsl[0]) {1} else {0} }
};
define count_operators = (dsl, operator) -> {
  if contains(["*", "+", "CP"], dsl[0]) {
    (if (operator == dsl[0]) {1} else {0})
      + count_operators(dsl[1], operator)
      + count_operators(dsl[2], operator)
  } else { 0 }
};
\end{lstlisting}
\end{subfigure}

\begin{subfigure}[t]{\linewidth}
\begin{lstlisting}[style=venturescript,frame=single]
// ** SAMPLING FROM VENTURE EXECUTABLE FOR PREDICTION **
define xs_test = get_data_xs("./data.test.csv");
define ys_test = get_data_ys("./data.test.csv");
define ys_pred = sample_all(gp_executable(${xs_test}));
\end{lstlisting}
\end{subfigure}
\end{subfigure}\hfill%
\begin{subfigure}[m]{.4\linewidth}

\subcaption{Observed Dataset}
\label{subfig:gp-tutorial-data}
\begin{lstlisting}[style=venturescript]
>> plot(xs, ys, "Year", "Passenger Volume")
\end{lstlisting}
\vspace{-.2cm}
\includegraphics[width=.95\linewidth]{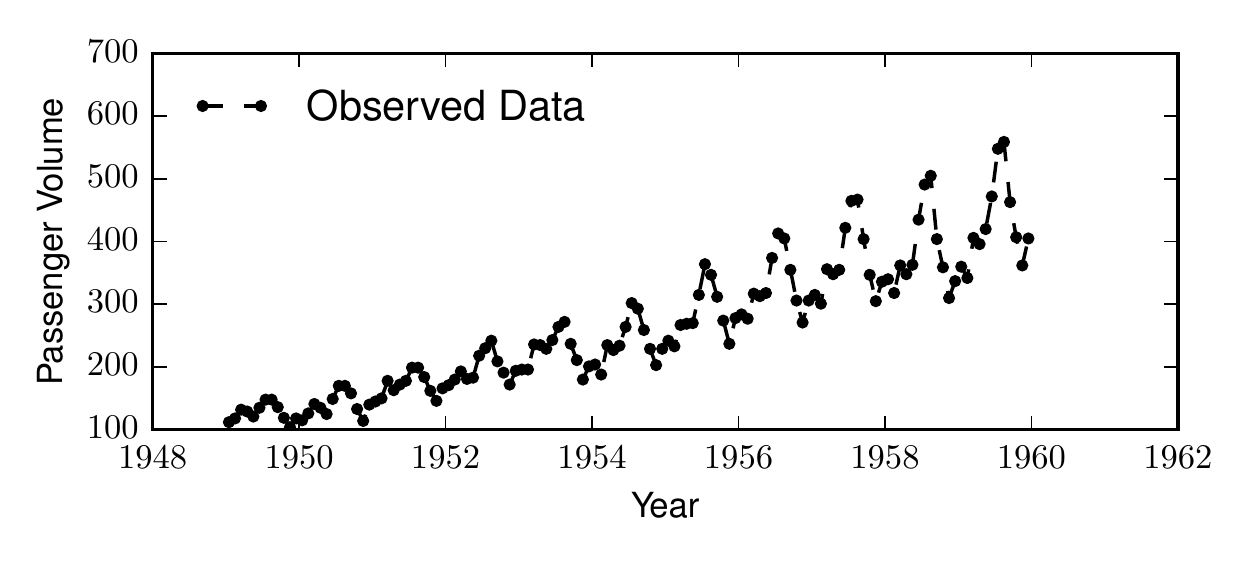}
\hrule

\subcaption{Synthesized DSL and PPL Programs}
\label{subfig:gp-tutorial-posterior}
\begin{lstlisting}[style=venturescript]
>> sample_all(dsl_source)[0]
['+',
  ['*',
    ['+', ['WN',49.5], ['C',250.9]],
    ['+', ['PER', 13.2, 8.6],
      ['+',
        ['LIN', 1.2],
        ['LIN', 4.9]]]],
  ['WN', 0.1]]
>> sample_all(ppl_source)[0]
make_gp(gp_mean_const(0.),
  ((x1,x2) -> {((x1,x2) -> {((x1,x2) ->
  {((x1,x2) ->
  if (x1==x2) {49.5} else {0})(x1,x2)
    + ((x1,x2) -> {250.9})(x1,x2)})(x1,x2)
    * ((x1,x2) -> {((x1,x2) -> {
    -2/174.2400*sin(2*pi/8.6*
      abs(x1-x2))**2})(x1,x2)
    + ((x1,x2) -> {((x1,x2) ->
    {(x1-1.2)*(x2-1.2)})(x1,x2)
    + ((x1,x2) ->
      {(x1-4.9)*(x2-4.9)})(x1,x2)})
    (x1,x2)})(x1,x2)})(x1,x2)
    + ((x1,x2) -> if (x1==x2) {0.1}
    else {0})(x1,x2)}))
\end{lstlisting}
\hrule

\subcaption{Processing Synthesized DSL Program}
\label{subfig:gp-tutorial-processing}
\begin{lstlisting}[style=venturescript]
// estimate probability that data has
// change point or periodic structure
>> mean(mapv((prog) -> {
    count_kernels(prog, "PER")
    + count_operators(prog, "CP") > 0
  }, sample_all(dsl_source)))
0.97
\end{lstlisting}
\hrule

\subcaption{Predictions from PPL Program}
\label{subfig:gp-tutorial-dynamic}
\begin{lstlisting}[style=venturescript]
>> plot(xs, ys, "Year", "Passenger Volume",
    xs_test, ys_test, ys_pred)
\end{lstlisting}
\vspace{-.2cm}
\includegraphics[width=.95\linewidth]{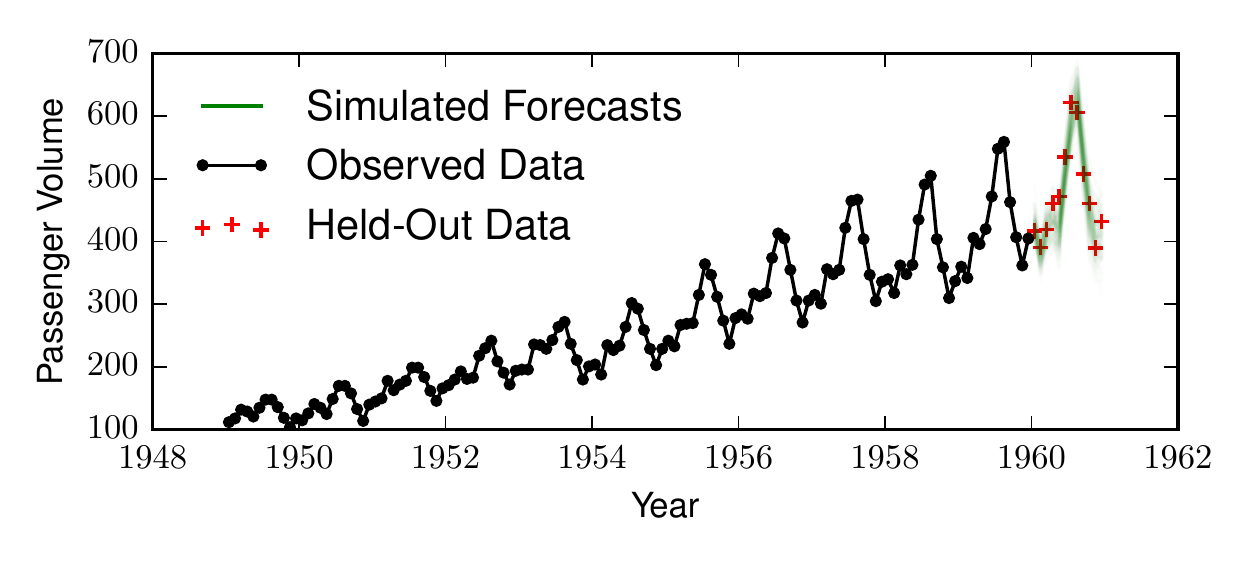}
\end{subfigure}
\captionsetup{skip=0pt}
\caption{Overview of Bayesian synthesis and execution of probabilistic programs
in the Gaussian process DSL.}
\label{fig:gp-tutorial}
\end{figure}

We next present an example that illustrates our synthesis technique and
applications to data interpretation and prediction.
Figure~\ref{subfig:gp-tutorial-data} presents a time series that plots
world-wide airline passenger volume from 1948 to 1960~\citep{box1976}. The
data contains the year and passenger volume data points. As is apparent from the
time series plot, the data contain several temporal patterns. Passenger volume
linearly increases over time, with periodic fluctuations correlated with the
time of year --- volume peaks in summer and falls in winter, with a small peak
around winter vacation.

\subsection{Gaussian Process Models}
\label{subsec:example-gp-introduction}

Gaussian processes (GPs) define a family of nonparametric regression models that
are widely used to model a variety of data~\citep{rasmussen2006}. It is possible
to use GPs to discover complex temporal patterns in univariate time
series~\citep{duvenaud2013}.
We first briefly review the Gaussian process, which places a prior distribution
over functions $f: \mathcal{X} \to \mathbb{R}$. In our airline passenger
example, $x \in \mathcal{X}$ is a time point and $f(x) \in \mathbb{R}$ is passenger
volume at time $x$. The GP prior can express both simple parametric forms, such
as polynomial functions, as well as more complex relationships dictated by
periodicity, smoothness, and time-varying functionals. Following notation of
\citet{rasmussen2006}, we formalize a GP
\mbox{$f \sim \mathsf{GP}(m, k)$} with mean function $m: \mathcal{X} \to
\mathcal{Y}$ and covariance function $k: \mathcal{X} \times \mathcal{X} \to
\mathbb{R}$ as follows: $f$ is a collection of random variables $\set{f(x): x
\in \mathcal{X}}$. Give time points $\set{x_1,\dots,x_n}$ the vector
of random variables $[f(x_1),\dots,f(x_n)]$ is
jointly Gaussian with mean vector $[m(x_1),\dots,m(x_n)]$ and covariance matrix
$[k(x_i,x_j)]_{1\le{i,j}\le{n}}$.

The prior mean is typically set to zero (as it can be absorbed by the
covariance). The functional form of the covariance $k$ defines essential
features of the unknown function $f$, allowing the GP to (i) fit
structural patterns in observed data, and (ii) make time series forecasts.
GP covariance functions (also called kernels) can be created by
composing a set of simple base kernels through sum, product, and change point
operators \citep[Section 4.2]{rasmussen2006}. We therefore define the following
domain-specific language for expressing the covariance function of a specific
GP:
\begin{align*}
K \in \DomKernel &\Coloneqq
    \ttpl \texttt{C}\; v \ttpr
    \mid \ttpl \texttt{WN}\; v \ttpr
    \mid \ttpl \texttt{SE}\; v \ttpr
    \mid \ttpl \texttt{LIN}\; v \ttpr
    \mid \ttpl \texttt{PER}\; v_1\; v_2 \ttpr && \mathrm{[BaseKernels]}\\
    &
    \mid \ttpl \texttt{+}\; K_1\; K_2 \ttpr
    \mid \ttpl \texttt{$\times$}\; K_1\; K_1 \ttpr
    \mid \ttpl \texttt{CP}\; v\; K_1\; K_2 \ttpr && \mathrm{[CompositeKernels]}
\end{align*}
The base kernels are constant (\texttt{C}), white noise (\texttt{WN}), squared
exponential (\texttt{SE}), linear (\texttt{LIN}), and periodic (\texttt{PER}).
Each base kernel has one or more numeric parameters $v$. For example,
\texttt{LIN} has an x-intercept and \texttt{PER} has a length scale and period.
The composition operators are sum ($+$), product ($\times$), and change point
($\texttt{CP}$, which smoothly transitions between two kernels at some
x-location).
Given a covariance function specified in this language, the predictive
distribution over airline passenger volumes $[f(x_1),\dots,f(x_n)]$ at time
points $(x_1,\dots,x_n)$ is a standard multivariate normal.

\subsection{Synthesized Gaussian Process Model Programs}
\label{subsec:example-gp-synthesized}

The first code box in Figure~\ref{fig:gp-tutorial} (lines 148-171) presents
Venture code that defines a prior distribution over programs in our
domain-specific Gaussian process  language. This code implements a probabilistic
context-free grammar that samples a specific program from the DSL. In the code
box, we have two representations of the program: (i) a program in the
domain-specific language (\texttt{dsl\_source}, line 169); and (ii) the
translation of this program into executable code in the Venture probabilistic
programming language (\texttt{ppl\_source}, line 171). When evaluated using
\texttt{venture\_eval}, the translated Venture code generates a stochastic
procedure that implements the Gaussian process model (\texttt{gp\_executable},
line 172).
Given the observed time points (\texttt{xs}, line 173 of the second code box in
Figure~\ref{fig:gp-tutorial}), the stochastic procedure \texttt{gp\_executable}
defines the probability of observing the corresponding passenger volume data
(\texttt{ys}, line 174) according to the GP.
This probability imposes a posterior distribution on the random choices made by
\texttt{generate\_random\_dsl\_code} (line 163) which sampled the
\texttt{dsl\_source}.
The next step is to perform Bayesian inference over these random choices to
obtain a collection of programs (approximately) sampled from this posterior
distribution. Programs sampled from the posterior tend to deliver a good fit
between the observed data and the Gaussian process. In our example we configure
Venture to sample 60 programs from the (approximately inferred) posterior (line
176).
The third code box in Figure~\ref{fig:gp-tutorial} actually performs the
inference, specifically by using a Venture custom inference strategy that
alternates between performing a Metropolis-Hastings step over the program
structure (base kernels and operators) and a Metropolis-Hastings step over the
program parameters (lines 177-179).


\subsection{Inferring Qualitative Structure by Processing Synthesized Programs}
\label{subsec:example-gp-processing}

Figure~\ref{subfig:gp-tutorial-posterior} presents one of the programs sampled
from the posterior. The syntactic structure of this synthesized program reflects
temporal structure present in the airline passenger volume data. Specifically,
this program models the data as the sum of linear and periodic components (base
kernels) along with white noise. Of course, this is only one of many programs
(in our example, 60 programs) sampled from the posterior. To analyze the
structure that these programs expose in the data, we check for
the existence of different base kernels and operators in each program (fourth
code box in Figure~\ref{fig:gp-tutorial} and
Figure~\ref{subfig:gp-tutorial-processing}). By averging across
the ensemble of synthesized program we
estimate the probable presence of each structure in the data. In our example,
all of the programs have white noise, 95\% exhibit a linear trend, 95\% exhibit
periodicity, and only 22\% contain a change point. These results indicate the
likely presence of both linear and periodic structure.

\subsection{Predicting New Data via Probabilistic Inference}
\label{subsec:example-gp-dynamic}

In addition to supporting structure discovery by processing the synthesized
domain-specific program (\texttt{dsl\_source}), it is also possible to
sample new data directly from the Venture executable (\texttt{gp\_executable})
obtained from the translated program (\texttt{ppl\_source}). In our example, this
capability makes it possible to forecast future passenger volumes. The final
code box in Figure~\ref{fig:gp-tutorial} presents one such forecast, which we
obtain by sampling data from the collection of 60 synthesized stochastic
procedures (the green line in Figure~\ref{subfig:gp-tutorial-dynamic} overlays
all the sampled predictions). The forecasts closely match the actual held-out
data, which indicates that the synthesized programs effectively capture
important aspects of the underlying temporal structure in the actual data.


\section{Bayesian Synthesis in Domain-Specific Data Modeling Languages}
\label{sec:bayesian-synthesis}


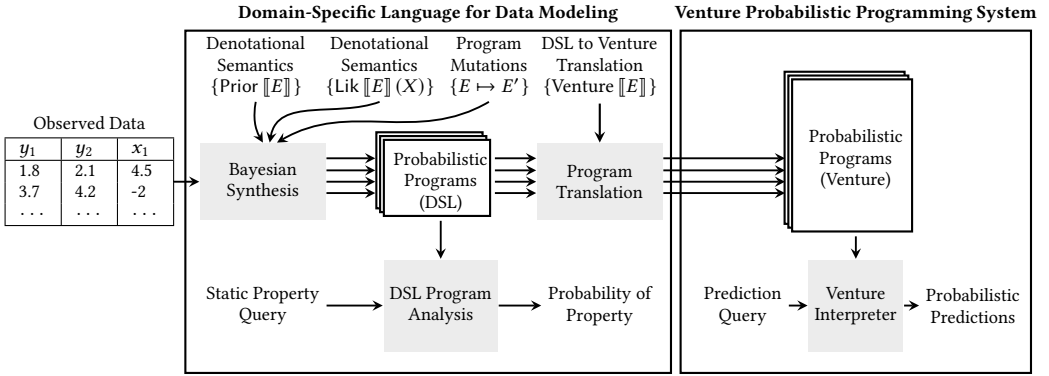
\begin{figure}[!t]

\scriptsize

\begin{tikzpicture}[thick]

\node[
  rectangle,
  draw=none,
  minimum height=1cm,
  align=center,
  text width = 1.5cm,
  fill=lightgray!30!white,
] (program-synthesis) {Bayesian\\Synthesis};

\node[
    rectangle,
    inner sep = 0pt,
    left = 0.20 of program-synthesis,
    text width=2.5cm,
    align=center,
    label={above:{Observed Data}},
] (data-table) {
    \begin{tabular}{|l|l|l|}\hline
    $y_1$ & $y_2$ & $x_1$ \\ \hline
    1.8 & 2.1 & 4.5 \\
    3.7 & 4.2 & -2 \\
    $\cdots$ & $\cdots$ & $\cdots$ \\ \hline
    \end{tabular}
};

\node[
  rectangle,
  draw=none,
  minimum height=1.2cm,
  align=center,
  above=1 of program-synthesis.north west,
  anchor=west,
] (denotation-prior) {Denotational\\Semantics\\$\set{\Denotation[\Prior]{E}}$};

\node[
  rectangle,
  draw=none,
  inner sep = 0pt,
  minimum height=1.2cm,
  align=center,
  right = .15 of denotation-prior.east,
  anchor=west,
] (denotation-likelihood) {Denotational\\Semantics\\$\set{\Denotation[\Likelihood]{E}(X)}$};

\node[
  rectangle,
  draw=none,
  inner sep = 0pt,
  minimum height=1.2cm,
  align=center,
  right=.15 of denotation-likelihood.east,
  anchor=west,
] (program-mutations) {Program\\Mutations\\$\set{E \mapsto E'}$};

\node[
  rectangle,
  align=center,
  minimum height = 1cm,
  text width = 1.25cm,
  right= .80 of program-synthesis.east,
  anchor=west,
  copy shadow={draw, fill=white, shadow xshift=-1.5mm, shadow yshift=1.5mm},
  copy shadow={draw, fill=white, shadow xshift=-1.0mm, shadow yshift=1.0mm},
  copy shadow={draw, fill=white, shadow xshift=-0.5mm, shadow yshift=0.5mm},
] (dsl-programs) {Probabilistic Programs\\(DSL)};

\node[
  rectangle,
  draw=none,
  minimum height=1cm,
  text width = 1.5cm,
  align=center,
  right= .55 of dsl-programs,
  fill=lightgray!30!white,
] (program-translation) {Program\\Translation};

\node[
  rectangle,
  align=center,
] (denotation-transformation)
  at (denotation-likelihood -| program-translation)
  {DSL to Venture\\Translation\\$\set{\Denotation[\mathrm{Venture}]{E}}$};

\node[
  rectangle,
  draw=none,
  minimum height=1.2cm,
  align=center,
  below= .5 of dsl-programs,
  fill=lightgray!30!white,
] (dsl-program-analysis) {DSL Program\\Analysis};

\node[
  rectangle,
  draw=none,
  minimum height=1.2cm,
  align=center,
] (dsl-property-query) at (program-synthesis |- dsl-program-analysis)
{Static Property\\Query};

\node[
  rectangle,
  draw=none,
  minimum height=1.2cm,
  align=center,
] (dsl-property-results) at (dsl-program-analysis -| program-translation)
{Probability of\\Property};

\node[
  rectangle,
  align=center,
  inner sep = 5pt,
  minimum height = 2cm,
  minimum width = 1.25cm,
  right=1.75 of program-translation.east,
  anchor=200,
  copy shadow={draw, fill=white, shadow xshift=-1.5mm, shadow yshift=1.5mm},
  copy shadow={draw, fill=white, shadow xshift=-1.0mm, shadow yshift=1.0mm},
  copy shadow={draw, fill=white, shadow xshift=-0.5mm, shadow yshift=0.5mm},
] (venture-programs) {Probabilistic\\Programs\\(Venture)};

\node[
  rectangle,
  draw=none,
  minimum height=1.2cm,
  align=center,
  fill=lightgray!30!white,
] (venture-interpreter)
  at (dsl-program-analysis -| venture-programs)
  {Venture\\Interpreter};

\node[
  rectangle,
  draw=none,
  minimum height=1.2cm,
  align=center,
  left= .25 of venture-interpreter
] (venture-predictive-query) {Prediction\\Query};

\node[
  rectangle,
  draw=none,
  minimum height=1.2cm,
  align=center,
  right= .20 of venture-interpreter
] (venture-result) {Probabilistic\\Predictions};

\node[
    rectangle,
    draw=black,
    above = 2 of program-synthesis.west,
    minimum height=4.5cm,
    text width= 6.3cm,
    anchor = north west,
    align=center,
    xshift=-0.2cm,
    label={\bfseries Domain-Specific Language for Data Modeling},
] (synthesis-system) {};

\node[
    rectangle,
    draw=black,
    right = 0.10 of synthesis-system.east,
    minimum height=4.5cm,
    text width = 4.5cm,
    anchor = west,
    align=center,
    label={\bfseries Venture Probabilistic Programming System},
] (venture-system) {};

\draw[-stealth, line width=0.25mm] ([xshift=-.15cm]data-table.east) to (program-synthesis.west);
\draw[-stealth, line width=0.25mm] ([yshift=.15cm]denotation-prior.south) [out=260,in=120] to (program-synthesis.90);
\draw[-stealth, line width=0.25mm] ([yshift=.15cm]denotation-likelihood.south) [out=200,in=80] to (program-synthesis.80);
\draw[-stealth, line width=0.25mm] ([yshift=.15cm]program-mutations.south) [out=210,in=30] to (program-synthesis.70);

\draw[-stealth, line width=0.25mm] (dsl-programs) -- (dsl-program-analysis);
\draw[-stealth, line width=0.25mm] (dsl-property-query) -- (dsl-program-analysis);
\draw[-stealth, line width=0.25mm] (dsl-program-analysis) -- (dsl-property-results);

\coordinate (ps-c0) at (program-synthesis.20 -| dsl-programs.west);
\coordinate (ps-c1) at (program-synthesis.10 -| dsl-programs.west);
\coordinate (ps-c2) at (program-synthesis.west -| dsl-programs.west);
\coordinate (ps-c3) at (program-synthesis.-10 -| dsl-programs.west);
\draw[-stealth, line width=0.25mm] (program-synthesis.20) -- ([xshift=-0.15cm]ps-c0);
\draw[-stealth, line width=0.25mm] (program-synthesis.10) -- ([xshift=-0.15cm]ps-c1);
\draw[-stealth, line width=0.25mm] (program-synthesis.east) -- ([xshift=-0.15cm]ps-c2);
\draw[-stealth, line width=0.25mm] (program-synthesis.-10) -- ([xshift=-0.15cm]ps-c3);

\coordinate (ds-c0) at (ps-c0 -| dsl-programs.east);
\coordinate (ds-c1) at (ps-c1 -| dsl-programs.east);
\coordinate (ds-c2) at (ps-c2 -| dsl-programs.east);
\coordinate (ds-c3) at (ps-c3 -| dsl-programs.east);
\draw[-stealth, line width=0.25mm] (ds-c0) -- (ds-c0 -| program-translation.west);
\draw[-stealth, line width=0.25mm] (ds-c1) -- (ds-c1 -| program-translation.west);
\draw[-stealth, line width=0.25mm] (ds-c2) -- (ds-c2 -| program-translation.west);
\draw[-stealth, line width=0.25mm] (ds-c3) -- (ds-c3 -| program-translation.west);

\draw[-stealth] (denotation-transformation) -- (program-translation);

\coordinate (ps-c0) at (ds-c0 -| program-translation.east);
\coordinate (ps-c1) at (ds-c1 -| program-translation.east);
\coordinate (ps-c2) at (ds-c2 -| program-translation.east);
\coordinate (ps-c3) at (ds-c3 -| program-translation.east);
\coordinate (vs-c0) at  (ps-c0 -| venture-interpreter.west);
\coordinate (vs-c1) at (ps-c1 -| venture-interpreter.west);
\coordinate (vs-c2) at (ps-c2 -| venture-interpreter.west);
\coordinate (vs-c3) at (ps-c3 -| venture-interpreter.west);
\draw[-stealth, line width=0.25mm] (ps-c0) -- ([xshift=-0.30cm]vs-c0);
\draw[-stealth, line width=0.25mm] (ps-c1) -- ([xshift=-0.30cm]vs-c1);
\draw[-stealth, line width=0.25mm] (ps-c2) -- ([xshift=-0.30cm]vs-c2);
\draw[-stealth, line width=0.25mm] (ps-c3) -- ([xshift=-0.30cm]vs-c3);

\draw[-stealth, line width=0.25mm] (venture-programs) -- (venture-interpreter);
\draw[-stealth, line width=0.25mm] (venture-predictive-query) -- (venture-interpreter);
\draw[-stealth, line width=0.25mm] (venture-interpreter) -- (venture-result);

\end{tikzpicture}

\smallskip
\caption{Components of Bayesian synthesis of probabilistic programs for
automatic data modeling.}
\label{fig:synthesis-diagram}
\end{figure}

We next present a general framework for Bayesian synthesis in
probabilistic domain-specific languages. We formalize the probabilistic DSL
$\Lang$ as a countable set of strings, where an expression $E \in \Lang$
represents the structure and parameters of a family of statistical models
specified by the DSL (such as the space of all Gaussian process models from
Section~\ref{sec:example}).
We associate $\Lang$ with a pair of denotational semantics $\Prior$ and
$\Likelihood$ that describe the meaning of expressions $E$ in the language:

\begin{itemize}
\item The ``prior'' semantic function $\Prior : \Lang \to (0,1]$, where
$\Denotation[\Prior]{E}$ is a positive real number describing the probability
that a given random process generated $E$.

\item The ``likelihood'' semantic function $\Likelihood : \Lang \to (\DataSpace
\to \Reals_{\ge 0})$, where $\Denotation[\Likelihood]{E}$ is a probability
function over a data space $\DataSpace$, where each item $X \in \DataSpace$ has a relative
probability $\Denotation[\Likelihood]{E}(X)$.
\end{itemize}

These two semantic functions are required to satisfy three technical conditions.

\begin{condition}[Normalized Prior]\label{cond:prior-normalized}
$\Prior$ defines a valid probability distribution over $\Lang$:
\begin{align}
\sum_{E \in \Lang} \Denotation[\Prior]{E} = 1.
\end{align}

If $\Lang$ is finite then this condition can be verified directly by summing the
values. Otherwise if $\Lang$ is countably infinite then a more careful analysis
is required (see Section~\ref{subsec:pcfg-prior}).
\end{condition}

\begin{condition}[Normalized Likelihood]\label{cond:likelihood-normalized}
For each $E \in \Lang$, $\Denotation[\Likelihood]{E}$ is a probability
distribution (for countable data) or density (for continuous data) over
$\mathcal{X}$:
\begin{align}
\forall E \in \Lang. \begin{cases}
\sum_{X \in \DataSpace} \Denotation[\Likelihood]{E}(X) = 1
  & \textrm{(if $\DataSpace$ is a countable space)} \\
\int_{X \in \DataSpace} \Denotation[\Likelihood]{E}(X)\mu(dX) = 1
  & \textrm{(if $\DataSpace$ is a general space with base measure $\mu$)}.
\end{cases}
\end{align}
\end{condition}

\begin{condition}[Bounded Likelihood]\label{cond:likelihood-bounded}
$\Likelihood$ is bounded, and is non-zero for some expression $E$:
\begin{align}
\forall X \in \DataSpace.\;
  0 < c^{\max}_X \Coloneqq \sup
  \set{\Denotation[\Likelihood]{E}(X) \mid E \in \Lang} < \infty.
\end{align}
\end{condition}

For each $X \in \DataSpace$, define $c_X \Coloneqq \sum_{E \in \Lang}
\Denotation[\Likelihood]{E}(X) \cdot \Denotation[\Prior]{E}$ to be the marginal
probability of $X$ (which is finite by Conditions~\ref{cond:prior-normalized}
and \ref{cond:likelihood-bounded}). When the semantic functions satisfy the
above conditions, they induce a new semantic function $\Posterior : \Lang \to
\DataSpace \to \Reals_{\ge{0}}$, called the posterior distribution:
\begin{align}
\Denotation[\Posterior]{E}(X) &\Coloneqq
  (\Denotation[\Likelihood]{E}(X)\; \Denotation[\Prior]{E}) / c_X.
\end{align}

\begin{lemma}
Let $\Lang$ be a language whose denotational semantics $\Prior$ and
$\Likelihood$ satisfy Conditions~\ref{cond:prior-normalized},
\ref{cond:likelihood-normalized}, and \ref{cond:likelihood-bounded}. For each $X
\in \DataSpace$, the function $\lambda E.\Denotation[\Posterior]{E}(X)$ is a
probability distribution over $\Lang$.
\end{lemma}

\begin{proof}
Fix $X$. Then $\sum_{E \in \Lang}\Denotation[\Posterior]{E}(X) =
\sum_{E \in \Lang} (\Denotation[\Likelihood]{E}(X)\; \Denotation[\Prior]{E}) /
c_X = c_X / c_X = 1.$
\end{proof}

The objective of Bayesian synthesis can now be stated.

\begin{objective}[Bayesian Synthesis]\label{objective:bayesian-synthesis}
Let $\Lang$ be a language whose denotational semantics $\Prior$ and
$\Likelihood$ satisfy Conditions~\ref{cond:prior-normalized},
\ref{cond:likelihood-normalized}, and \ref{cond:likelihood-bounded}. Given a
dataset $X \in \DataSpace$, generate expressions $E$ with probability
$\Denotation[\Posterior]{E}(X)$.
\end{objective}

Figure~\ref{fig:synthesis-diagram} shows the main components of our approach.
Assuming we can achieve the objective of Bayesian synthesis, we now outline how
to use synthesized probabilistic programs in the DSL for (i) inferring
qualitative statistical structure, by using simple program analyses on the
program text; and (ii) predicting new data by translating them into executable
Venture probabilistic programs.

\subsection{Inferring Probabilities of Qualitative Properties in Synthesized DSL
Programs}
\label{subsec:bayesian-syntheis-qualitative}

After synthesizing an ensemble of $n$ DSL expressions $\set{E_1,\dots,E_n}$
according to the posterior probability distribution given $X$, we can
obtain insight into the learned models by processing the
synthesized DSL expressions.
In this paper, a typical query has the signature $\HasProperty: \Lang \to
\set{0,1}$ that checks whether a particular property holds in a given DSL
program. In other words, $\Denotation[\HasProperty]{E}$ is 1 if the $E$ reflects
the property, and 0 otherwise. We can use the synthesized programs to form an
unbiased estimate of the posterior probability of a property:
\begin{align}
\label{eq:prob-has-property}
\Pr\set{\Denotation[\HasProperty]{E} \mid X}
  \approx \frac{1}{n}\sum_{i=1}^{n} \Denotation[\HasProperty]{E_i}.
\end{align}

\subsection{Forming Predictions by Translating DSL Programs into Venture
Programs}
\label{subsec:bayesian-syntheis-predictions}

In our formalism, each DSL expression $E \in \Lang$ is a probabilistic program
that represents the structure and parameters of a family of statistical models
specified by the domain-specific language.
In this paper, we translate probabilistic programs $E$ in the DSL into new
probabilistic programs $\Denotation[\Venture]{E}$ specified in Venture
\citep{mansinghka2014}.
A main advantage of translating DSL programs into Venture programs is that we
can reuse general-purpose inference machinery in Venture \citep{mansinghka2018}
to obtain accurate predictions, as opposed to writing custom interpreters for
data prediction on a per-DSL basis.
Sections \ref{sec:dsl-time-series} and \ref{sec:dsl-crosscat} show two concrete
examples of how DSL programs can be translated into Venture programs and how to
use inference programming in Venture to obtain predictions on new data.

\subsection{Markov Chain Monte Carlo Algorithms for Bayesian Synthesis}
\label{subsec:bayesian-syntheis-mcmc}

We next describe Markov chain Monte Carlo (MCMC) techniques to
randomly sample expressions $E \in \Lang$ with probability that approximates
$\post{E}$ to achieve Objective~\ref{objective:bayesian-synthesis}.
In particular, we employ a general class of MCMC sampling algorithms
(Algorithm~\ref{alg:mcmc}) that first obtains an expression $E_0 \in \Lang$ such
that $\lik{E} > 0$ and then iteratively generates a sequence of expressions
$E_1, \ldots, E_n \in \Lang$.
The algorithm iteratively generates $E_i$ from $E_{i-1}$ using a DSL-specific
\emph{transition operator} $\TransOper$, which takes as input expression $E \in
\Lang$ and a data set $X$, and stochastically samples an expression $E'$ with
probability denoted $\TransOperX{E \to E'}$, where $\sum_{E' \in \Lang}
\TransOperX{E \to E'} = 1$ for all $E \in \Lang$.
The algorithm returns the final expression $E_n$.
Implementing Bayesian synthesis for a data-modeling DSL $\Lang$ requires
implementing three procedures:
\begin{enumerate}
\item $\textproc{generate-expression-from-prior}$, which generates expression
  $E$ with probability $\prior{E}$.
\item $\textproc{evaluate-likelihood}(X, E)$, which evaluates $\lik{E}$.
\item $\textproc{generate-new-expression}(X, E)$, which generates expression
$E'$ from expression $E$ with probability $\TransOperX{E \to E'}$.
\end{enumerate}

\begin{algorithm}[H]
\caption{Template of Markov chain Monte Carlo algorithm for Bayesian synthesis.}
\label{alg:mcmc}
\begin{algorithmic}[1]
  \Procedure{bayesian-synthesis}{$X$, $\TransOper$, $n$}
    \Do
      \State $E_0 \sim \textproc{generate-expression-from-prior}()$
      \doWhile{$\lik{E_0} = 0$}
      \For{$i=1\ldots n$}
          \State $E_i \sim \textproc{generate-new-expression}(X, E_{i-1}))$
      \EndFor
      \State \Return $E_n$
  \EndProcedure
\end{algorithmic}
\end{algorithm}

%
Using MCMC for synthesis allows us to rigorously characterize soundness conditions.
We now describe three conditions on the operator $\TransOper$ that are sufficient
to show that Algorithm~\ref{alg:mcmc} generates expressions from arbitrarily
close approximation to the Bayesian posterior distribution on expressions
$\post{E}$.
In Section~\ref{subsec:pcfg-mcmc} we show how to construct such an operator for
a class of DSLs generated by probabilistic context-free grammars and prove that
this operator satisfies the conditions.

\begin{condition}[Posterior invariance]\label{cond:invariance}
If an expression $E$ is sampled from the posterior distribution and a new
expression $E' \in \Lang$ is sampled with probability $\TransOperX{E \to E'}$,
then $E'$ is also a sample from the posterior distribution:
\begin{align*}
  \sum_{E \in \Lang} \post{E} \cdot \TransOperX{E \to E'} = \post{E'}.
\end{align*}
\end{condition}

\begin{condition}[Posterior irreducibility]\label{cond:irreducibility}
Every expression $E'$ with non-zero likelihood is reachable from every
expression $E \in \Lang$ in a finite number of steps.
That is, for all pairs of expressions $E \in \Lang$ and $E' \in \{\Lang :
\lik{E'} > 0\}$ there exists an integer $n \ge 1$ and a sequence of expressions
$E_1, E_2, \ldots, E_n$ where $E_1 = E$ and $E_n = E'$ such that
$\TransOperX{E_{i-1} \to E_i} > 0$ for all $i \in \{2, \ldots n\}$.
\end{condition}

\begin{condition}[Aperiodicity]\label{cond:aperiodicity}
There exists some expression $E \in \Lang$ such that the transition operator has
a non-zero probability of returning to the same expression,
i.e.~$\TransOperX{E \to E} > 0$.
\end{condition}

We now show that Algorithm~\ref{alg:mcmc} gives asymptotically correct results
provided these three conditions hold.
First, we show that it is possible to obtain an expression $E_0 \in \Lang$ where
$\lik{E_0} > 0$ using a finite number of invocations of
\textproc{generate-expression-from-prior}.

\begin{lemma}
The do-while loop of Algorithm~\ref{alg:mcmc} will terminate with probability 1,
and the expected number of iterations in the do-while loop is at most
$c_X^{\max} / c_X$.
\end{lemma}
\begin{proof}
The number of iterations of the loop is geometrically distributed with mean $p$,
given by:
\begin{align*}
  p &= \sum_{E \in \Lang} \prior{E} \cdot \mathbb{I}[\lik{E} > 0]
  = \frac{1}{c_X^{\max}}
    \sum_{E \in \Lang} \prior{E} \cdot c^{\max}_{X} \cdot
      \mathbb{I}[\lik{E} > 0]\\
  &\ge \frac{1}{c_X^{\max}} \sum_{E \in \Lang} \prior{E} \cdot \lik{E}
  = \frac{c_X}{c_X^{\max}}.
\end{align*}
Therefore, the expected number of iterations of the do-while loop is at most
$1/p = c_X^{\max} / c_X < \infty$.
\end{proof}

\noindent We denote the probability that Algorithm~\ref{alg:mcmc} returns
expression $E$, given that it started with expression $E_0$ by
$\Denotation[\ApproxPosterior{_{E_0}^1}]{E}(X)$, and define the $n$-step
probability inductively:
\begin{align*}
    \Denotation[\ApproxPosterior{_{E_0}^1}]{E}(X)
      & \Coloneqq \TransOperX{E_0 \to E}\\
    \Denotation[\ApproxPosterior{_{E_0}^n}]{E}(X)
      & \Coloneqq \sum_{E' \in \Lang} \TransOperX{E' \to E}
        \cdot \Denotation[\ApproxPosterior{_{E_0}^{n-1}}]{E'}(X) && (n > 1).
\end{align*}

We can now state a key convergence theorem, due to \citet{tierney1994}.

\begin{theorem}[Convergence of MCMC \citep{tierney1994}]
\label{thm:mcmc-convergence}
If Condition~\ref{cond:invariance} and Condition~\ref{cond:irreducibility} and
Condition~\ref{cond:aperiodicity} hold for some language $\Lang$, transition
operator $\TransOper$, and data $X \in \DataSpace$, then
\begin{align*}
  \forall E_0, E \in \Lang.\;
  \lik{E_0} > 0 \implies
    \lim_{n \to \infty}
    \Denotation[\ApproxPosterior{_{E_0}^n}]{E}(X) \to \post{E}.
\end{align*}
\end{theorem}

\section{Bayesian Synthesis for Domain-Specific Languages Defined By
  Probabilistic Context-Free Grammars}
\label{sec:pcfg}

Having described the general framework for Bayesian synthesis in general
domain-specific languages, we now focus on the class of domain-specific
languages which are generated by a probabilistic context-free grammar (PCFG)
over symbolic expressions (s-expressions).
We begin by describing the formal model of the grammar and the languages it can
produce. We then outline a default denotational semantics $\Prior$, derive a
general MCMC algorithm for all languages in this class, and prove that the
synthesis technique converges to the posterior distribution over expressions.

\subsection{Tagged Probabilistic Context-Free Grammars with Random Symbols}
\label{subsec:pcfg-definition}

Probabilistic context-free grammars are commonly used models for constructing
languages~\citep{jelinek1992}.
We describe a special type of PCFG, called a tagged PCFG with random symbols, by
extending the standard definition in two ways: (i) we require the grammar to
produce s-expressions containing a unique phrase tag for each production rule,
which guarantees that the grammar is unambiguous; and (ii) we allow each
non-terminal to specify a probability distribution over symbols in the alphabet.
The model is formally described below.

\begin{definition}[Tagged probabilistic context-free grammar with random symbols]
\label{def:tagged-pcfg}

A tagged probabilistic context-free grammar with random symbols is a tuple
$G=(\Sigma, N, R, T, P, Q, S)$ where

\begin{itemize}
\item $\Sigma$ is a finite set of terminal symbols.

\item $N \Coloneqq \set{N_1, \dots, N_m}$ is a finite set of non-terminal symbols.

\item $R \Coloneqq \set{R_{ik} \mid i=1,\dots,m; k=1,\dots,r_i}$ is a set of
production rules, where $\rik$ is the $k\textsuperscript{th}$ production
rule of non-terminal $N_i$.
Each production rule $\rik$ is a tuple of the form
\begin{align}
\rik &\Coloneqq (N_i, \tik, \tilde{N}_1\cdots\tilde{N}_{\hik}),
\label{eq:production-rule}
\end{align}
where $\hik \ge 0$ is the number of non-terminals on the right-hand side of $\rik$
and each $\tilde{N}_j \in N$ $(j=1,2,\dots,\hik)$ is a non-terminal symbol.
If $\hik = 0$ then $\tilde{N}_1\cdots\tilde{N}_{\hik} = \epsilon$ is the empty
string and we call $\rik$ a \textit{non-recursive} production rule. Otherwise it
is a \textit{recursive} rule.

\item $T \Coloneqq \set{T_{ik} \mid i=1,\dots,m; k=1,\dots,r_i}$
 is a set of phrase tag symbols, disjoint from $N$, where $\tik$ is a
unique symbol identifying the production rule $\rik$.

\item $P : T \to (0,1]$ is a map from phrase tag symbols to their probabilities,
where $P(\tik)$ is the probability that non-terminal $N_i$ selects its
$k\textsuperscript{th}$ production rule $\rik$. For each non-terminal $N_i$, the
probabilities over its production rules sum to unity
$\sum_{k=1}^{r_i}P(\tik) = 1$.

\item $Q : T \times \Sigma \to [0,1]$ is a map from phrase tags and terminal
symbols to probabilities, where $Q(\tik, s)$ is the probability that production
rule $\rik$ of non-terminal $N_i$ draws the terminal symbol $s \in
\Sigma$. For each tag $\tik$, the probabilities over symbols sum to unity
$\sum_{s \in \Sigma}Q(\tik, s) = 1$.

\item $S \in N$ is a designated start symbol.
\end{itemize}

We additionally assume that grammar $G$ is \textit{proper}: production rules
must be cycle-free, and there are no useless symbols in $\Sigma \cup N \cup T$
\citep{nijholt1980}.
\end{definition}

%
The production rules from Eq~\eqref{eq:production-rule} define the rewriting
rules that describe how to syntactically replace a non-terminal $N_i$ with a
tagged s-expression.
%
%
We now describe how an evaluator uses these production rules to generate tagged
s-expressions according to the probabilities $P$ and $Q$.
The big-step sampling semantics for this evaluator are shown below, where the
notation $N_i \Downarrow^p_G E$ means that starting from non-terminal $N_i$, the
evaluator yielded expression $E$ with $p$ being the total probability of all the
phrase tags and terminal symbols in the generated s-expression.
\begin{align*}
\textrm{\footnotesize[Sampling: Non-Recursive Production Rule]} &&
\textrm{\footnotesize[Sampling: Recursive Production Rule]} &&
\\
\infer
  {N_i \Downarrow^{P(\tik)Q(\tik,s)\phantom{\prod_{z=1}^{\hik}{p_z}}}_G
     \hspace{-2.5em} (\tik\; s)}
  {(N_i, \tik, \epsilon) \in R,\, s \in \Sigma,\, Q(\tik,s) > 0}
&&
\infer
  {N_i \Downarrow^{P(\tik)\prod_{z=1}^{\hik}{p_z}}_G (\tik\; E_1\; \dots\; E_{\hik})}
  {(N_i, \tik, \tilde{N}_1 \cdots \tilde{N}_{\hik}) \in R,\;
    \tilde{N}_1 \Downarrow^{p_1}_G E_1,\; \dots,\;
    \tilde{N}_{\hik} \Downarrow^{p_{\hik}}_{G} E_j}
\end{align*}

In words, when the evaluator encounters a non-terminal symbol $N_i$, it chooses
a production rule $\rik$ with probability $P(\tik)$.
If the selected production rule is non-recursive, the evaluator randomly samples
a symbol $s \in \Sigma$ with probability $Q(\tik, s)$, and returns an
s-expression starting with $\tik$ followed by the random symbol.
Otherwise, if the selected production rule is recursive, the evaluator
recursively evaluates all the constituent non-terminals and returns an
s-expression starting with $\tik$ followed by all the evaluated
sub-expressions.
Note that each evaluation step yields an s-expression where the first element is
a phrase tag $\tik$ that unambiguously identifies the production rule $\rik$
that was selected to produce the expression. Hence, every expression maps
uniquely to its corresponding parse tree simply by reading the phrase tags
\citep{turbak2008}. As a result, the probability of any expression under its
sampling semantics is unambiguous which is essential for the soundness
properties established in Section~\ref{subsec:pcfg-mcmc}.

Finally, we let $\Lang(G, N_i)$ denote the set of all strings that can be
yielded starting from non-terminal $N_i$ $(i = 1,\dots,m)$, according to the
sampling semantics $\Downarrow_G$. The \textit{language} $\Lang(G)$ generated by
$G$ is the set of all strings derivable from the start symbol $S$, so that
$\Lang(G) \Coloneqq \Lang(G, S)$. Conceptually, for a probabilistic
domain-specific language $\Lang$ specified by tagged PCFGs, terminal symbols $s
\in \Sigma$ are used to represent the program parameters and tag symbols
$t \in T$ represent the program structure.

\subsection{A Default Prior Semantics with the Normalization Property}
\label{subsec:pcfg-prior}

Having described the big-step sampling semantics for expressions $E \in \Lang$
generated by a tagged probabilistic context-free grammar $G$, we next describe
the ``prior'' denotational semantics of an expression, given by the semantic
function $\Prior : \Lang(G) \to (0,1]$.
To aid with the construction, we first introduce some additional notation.
%
%
%
%
Define the semantic function %
$\GenDist : \Lang(G) \to N \to [0, 1]$ which takes an expression and a
non-terminal symbol, and returns the probability that the non-terminal evaluates
to the given expression:
\begin{align*}
\gendist{(\tik \; s)}{N_i}
  &\Coloneqq P(\tik) \cdot Q(\tik,s) \\
\gendist{(\tik \; E_1 \; \cdots \; E_{\hik})} {N_i}
  & \Coloneqq P(\tik) \cdot \textstyle\prod_{z=1}^{\hik} \gendist{E_z}{\tilde{N}_{z}}\\
  &\; \mbox{where } \rik = (\tik, N_i, \tilde{N}_1, \dots, \tilde{N}_{\hik}),
\end{align*}
for $i=1,\dots,n$ and $k=1,\dots,r_i$.

\begin{lemma}
For each non-terminal $N_i$ and for all expressions $E \in \Lang(G, N_i)$,
we have
\begin{align*}
\gendist{E}{N_i} = p \mbox{ if and only if } N_i \Downarrow^p_G E.
\end{align*}
\end{lemma}

\begin{proof}
By structural induction on the s-expression $E$.
\end{proof}

Recalling that $S$ is the start symbol of $G$, we define the prior semantics
\begin{equation}
\prior{E} \Coloneqq \gendist{E}{S}.
\label{ref:define-prior-gendist}
\end{equation}

Having established the correspondence between the sampling semantics and
denotational semantics, we next provide necessary and sufficient conditions on
$G$ that are needed for $\Prior$ to be properly normalized as required by
Condition~\ref{cond:prior-normalized}. We begin with the following definition:

\begin{definition}[Consistency \citep{booth1973}]
\label{def:pcfg-consistency}
A probabilistic context-free grammar $G$ is consistent if the probabilities
assigned to all words derivable from $G$ sum to 1.
\end{definition}

The following result of \citet{booth1973} plays a key role in our construction.

\begin{theorem}[Sufficient Condition for Consistency \citep{booth1973}]
\label{thm:pcfg-consistency}
A proper probabilistic context-free grammar $G$ is consistent if the largest
eigenvalue (in modulus) of the expectation matrix of $G$ is less than 1.
\end{theorem}

The expectation matrix of $G$ can be computed explicitly using the transition
rule probabilities $P$ for non-terminal symbols \citep[Equation 2]{gecse2010}.
If $G$ is non-recursive then it is necessarily consistent since $\Lang(G)$
consists of a finite number of finite length words.
Otherwise, if $G$ is recursive,  consistency is equivalently described by having
the expected number of steps in the stochastic rewriting-process $\Downarrow_G$
be finite \citet{gecse2010}.
%
%
To ensure that the $\Prior$ semantic function of $G$ is correctly normalized, we
construct the expectation matrix and confirm that the modulus of the largest
eigenvalue is less than 1.
(Note that the probabilities $Q$ of random symbols in the tagged PCFG are
immaterial in the analysis, since they appear only in non-recursive production
rules and do not influence the production rule probabilities $P$.)
We henceforth require every tagged probabilistic context-free grammar to satisfy
the stated conditions for consistency.

\subsection{Bayesian Synthesis Using Markov Chain Monte Carlo}
\label{subsec:pcfg-mcmc}

We now derive a Bayesian synthesis algorithm specialized to
domain-specific languages which are generated by a tagged context-free grammar,
and prove that the algorithm satisfies the conditions for convergence in
Theorem~\ref{thm:mcmc-convergence}.
Recall that implementing Algorithm~\ref{alg:mcmc} requires an implementation of
three procedures:

\begin{enumerate}
\item $\textproc{generate-expression-from-prior}$, which generates
  expression $E$ with probability $\prior{E}$ (subject to
  Condition~\ref{cond:prior-normalized}).
  Letting $E \sim \gendist{\cdot}{N_i}$ mean that $E$ is sampled randomly with
  probability $\gendist{E}{N_i}$ $(N_i \in N)$, we implement this procedure by
  sampling $E \sim \gendist{\cdot}{S}$ as described in
  Section~\ref{subsec:pcfg-prior}.

\item $\textproc{evaluate-likelihood}(X, E)$, which evaluates $\lik{E}$
  (subject to Conditions~\ref{cond:likelihood-normalized} and
  \ref{cond:likelihood-bounded}).
  This procedure is DSL specific; two concrete examples are given in
  Sections~\ref{sec:dsl-time-series} and \ref{sec:dsl-crosscat}.

\item $\textproc{generate-new-expression}(X, E)$, which generates a new
  expression $E'$ from the starting expression $E$ with probability
  $\TransOperX{E \to E'}$ (subject to Conditions~\ref{cond:irreducibility},
  \ref{cond:invariance}, and \ref{cond:aperiodicity}), which we describe below.
\end{enumerate}

We construct a class of transition operators $\TransOperX{E \to E'}$ that (i)
stochastically replaces a random sub-expression in $E$; then (ii) stochastically
accepts or rejects the mutation depending on how much it increases (or
decreases) the value of $\lik{E}$.
Our construction applies to any data-modeling DSL generated by a tagged PCFG and
is shown in Algorithm~\ref{alg:cfg-mcmc}. In this section we establish that the
transition operator in Algorithm~\ref{alg:cfg-mcmc} satisfies
Conditions~\ref{cond:irreducibility}, \ref{cond:invariance}, and
\ref{cond:aperiodicity} for Bayesian synthesis with Markov chain Monte Carlo
given in Section~\ref{subsec:bayesian-syntheis-mcmc}.

\begin{algorithm}[ht]
\caption{Transition operator $\TransOper$ for a context-free data-modeling language.}
\label{alg:cfg-mcmc}
\begin{algorithmic}[1]
\Procedure{generate-new-expression}{$E, X$}
  \Comment{{\footnotesize Input expression $E$ and data set $X$}}
    \State $a \sim \mbox{SelectRandomElementUniformly}(A_E)$
      \Comment{{\footnotesize Randomly select a node in parse tree}}
    \State $(N_i, \Ehole) \gets \sever{E}{a}$
      \Comment{{\footnotesize Sever the parse tree and return the
        non-terminal symbol at the sever point}}
      \label{algline:cfg-mcmc-sever}
    \State $\Esub \sim \gendist{\cdot}{N_i}$
      \Comment{{\footnotesize Generate random $\Esub$ with probability
        $\gendist{\Esub}{N_i}$}}
      \label{algline:cfg-mcmc-expand}
    \State $E' \gets \Ehole[\Esub]$
      \Comment{{\footnotesize Fill hole in $\Ehole$ with expression $\Esub$}}
    \State $L \gets \lik{E}$
      \Comment{{\footnotesize Evaluate likelihood for expression $E$
        and data set $X$}}
      \label{algline:cfg-mcmc-likelihood}
    \State $L' \gets \lik{E'}$
      \Comment{{\footnotesize Evaluate likelihood for expression $E'$
        and data set $X$}}
    \State $\paccept \gets \min\left\{1, (|A_E| / |A_{E'}|) \cdot (L' /  L)\right\}$
      \Comment{{\footnotesize Compute the probability of accepting the mutation}}
    \State $r \sim \mbox{UniformRandomNumber}([0, 1])$
      \Comment{{\footnotesize Draw a random number from the unit interval}}
    \If{$r < \paccept$}
        \Comment{{\footnotesize If-branch has probability $\paccept$}}
        \State \textbf{return} $E'$
          \Comment{{\footnotesize Accept and return the mutated expression}}
    \Else
        \Comment{{\footnotesize Else-branch has probability $1-\paccept$}}
        \State \textbf{return} $E$
          \Comment{{\footnotesize Reject the mutated expression,
            and return the input expression}}
    \EndIf
\EndProcedure
\end{algorithmic}
\end{algorithm}

Let $G$ be a grammar from Definition~\ref{def:tagged-pcfg} with language $\Lang
\Coloneqq \Lang(G)$.
We first describe a scheme for uniquely identifying syntactic locations in the
parse tree of each expression $E \in \Lang$.
Define the set %
$A \Coloneqq \set{(a_1, a_2, \ldots, a_l) \mid
    a_i \in \set{1, 2, \ldots, h_{\max}},
    l \in \set{0, 1, 2, \ldots}}
$
to be a countably infinite set that indexes nodes in the parse tree of $E$,
where $h_{\max}$ denotes the maximum number of symbols that appear on the right
of any given production rule of the form Eq~\eqref{eq:production-rule} in the
grammar.
Each element $a \in A$ is a sequence of sub-expression positions on the path
from the root node of a parse tree to another node.
For example, if $E = (t_0 \; E_1 \; E_2)$ where $E_1 = (t_1 \; E_3 \; E_4)$ and
$E_2 = (t_2 \; E_5 \; E_6)$, then the root node of the parse tree has index
$\aroot \Coloneqq ()$; the node corresponding to $E_1$ has index $(1)$, the node
corresponding to $E_2$ has index $(2)$; the nodes corresponding to $E_3$ and
$E_4$ have indices $(1, 1)$ and $(1, 2)$ respectively; and the nodes
corresponding to $E_5$ and $E_6$ have indices $(2, 1)$ and $(2, 2)$.
For an expression $E$, let $A_E \subset A$ denote the finite subset of nodes
that exist in the parse tree of $E$.

Let $\square$ denote a hole in an expression.
For notational convenience, we extend the definition of $\GenDist$ to be
a partial function with
signature $\GenDist : \{\Sigma \cup T \cup \square \}^* \to N \to [0, 1]$ and
add the rule: $\gendist{\square}{N_i} \Coloneqq 1$ (for each $i=1,\dots,m)$.
We define an operation $\Sever$, parametrized by $a \in A$, that takes an
expression $E$ and either returns a tuple $(N_i, \Ehole)$ where $\Ehole$
is the expression with the sub-expression located at $a$ replaced with $\square$
and where $N_i$ is the non-terminal symbol from which the removed sub-expression
is produced, or returns failure ($\varnothing$) if the parse tree for $E$ has no
node $a$:
\begin{align*}
\infer{(a, (t \; E_1 \cdots E_l)) \SeverTo (N_i, \square)}
  {a = (),\; \exists{k}. t \equiv T_{ik}}
\qquad
\infer{(a, (t \; E_1 \; \cdots E_l))
  \SeverTo (N_i, (t \; E_1 \; \cdots \; E_{j-1}
    \; \Esev \; E_{j+1} \; \cdots \; E_l)
  }
  {((a_2, a_3, \ldots), E_j) \SeverTo (N_i, \Esev) \mbox{ and } a_1 = j}
\\
\sever{(t \; E_1 \; E_2 \; \ldots \; E_l)}{a} \Coloneqq
  \begin{cases}
    (N_i, \Ehole) & \mbox{ if }
      (a, (t \; E_1 \; E_2 \; \ldots \; E_l))
      \SeverTo (N_i, \Ehole)\\ \varnothing & \mbox{ otherwise }
  \end{cases}
\end{align*}
Note that for any expression $E \in \Lang$, setting $a = \aroot \equiv ()$ gives
$((), E) \SeverTo (S, \square)$ where $S \in N$ is the start symbol.
Also note that $\Ehole \not \in \Lang$ because $\Ehole$ contains
$\square$.
For expression $\Ehole$ that contains a single hole, where $(a, E) \SeverTo
(N_i, \Ehole)$ for some $a$ and $E$, let $\Ehole[\Esub] \in \Lang$
denote the expression formed by replacing $\square$ with $\Esub$, where
$\Esub \in \Lang(G, N_i)$.
We further define an operation $\Subexpr$, parametrized by $a$, that extracts
the sub-expression corresponding to node $a$ in the parse tree:
\begin{align*}
  \subexpr{(t \; E_1 \; E_2 \; \ldots \; E_l)}{a} \Coloneqq
  \begin{cases}
    \varnothing & \mbox{ if } a = () \mbox{ or } a_1 > k \\
    E_j & \mbox{ if } a = (j) \mbox{ for some } 1 \le j \le k \\
    \subexpr{E_j}{(a_2, a_3, \ldots)} & \mbox{ if } i \ne (j) \mbox{ and }  a_1 = j \mbox{ for some } 1 \le j \le k
  \end{cases}
\end{align*}

Consider the probability that $\TransOper$ takes an expression $E$ to another
expression $E'$, which by total probability is an average over the uniformly
chosen node index $a$:
\begin{align}
\TransOperX{E \to E'}
  &= \frac{1}{|A_E|} \sum_{a \in A_E}
  \TransOperXa{E \to E'}{a}
  = \frac{1}{|A_E|} \sum_{a \in A_E \cap A_{E'}} \TransOperXa{E \to E'}{a},
  \notag \\
\intertext{where}
\TransOperXa{E \to E'}{a} &\Coloneqq
  \begin{cases}
  \begin{aligned}
  &\gendist{\subexpr{E'}{a}}{N_i} \cdot \alpha(E, E')
    + \mathbb{I}[E = E'] (1 - \alpha(E, E')) \span\\
  &\phantom{0}\qquad \mbox{if}\;\,
    \sever{E}{a} = \sever{E'}{a} = (N_i, \Ehole)
    \mbox{ for some } i \mbox{ and } \Ehole , \\
  &0\qquad \mbox{otherwise},
  \end{aligned}  \\
  \end{cases}
  \notag \\
\alpha(E, E')
  &\Coloneqq \min \set{1, \frac{|A_E| \cdot \lik{E'}}{|A_{E'}| \cdot \lik{E}}}.
  \label{eq:mh-acceptance-ratio}
\end{align}
Note that we discard terms in the sum with $a \in A_E \setminus A_{E'}$
because for these terms $\sever{E'}{a} = \varnothing$ and $\TransOperXa{E \to E'}{a}
= 0$.

\begin{lemma} \label{lemma:expand-nonzero}
For $E \in \Lang$, if\, $\sever{E}{a} = (N_i, \Ehole)$ then
$\gendist{\subexpr{E}{a}}{N_i} > 0$.
\end{lemma}
\begin{proof}
If $\sever{E}{a} = (N_i, \Ehole)$ then $\subexpr{E}{a}$ is an expression with
tag $t \in T_i$.
Since the expression $E \in \Lang$, it follows that
$\gendist{\subexpr{E}{a}}{N_i} > 0$ because otherwise $\prior{E} = 0$ (a
contradiction).
\end{proof}

\begin{lemma} \label{lemma:factor}
For $E \in \Lang$, if $\sever{E}{a} = (N_i, \Ehole)$ then:
\begin{align*}
  \gendist{E}{S} = \gendist{\subexpr{E}{a}}{N_i} \cdot \gendist{\Ehole}{S}.
\end{align*}
\end{lemma}
\begin{proof}
Each factor in $\gendist{E}{S}$ corresponds to a particular node $a'$ in the
parse tree.
Each factor corresponding to $a'$ appears in $\gendist{\subexpr{E}{a}}{N_i}$ if
$a'$ is a descendant of $a$ in the parse tree, or in $\gendist{\Ehole}{S}$
otherwise.
\end{proof}

\begin{lemma} \label{lemma:cfg-mcmc-new-prior}
For any $E, E' \in \Lang$ where $\sever{E}{a} = \sever{E'}{a} = (N_i, \Ehole)$:
\begin{align*}
  \prior{E'} = \prior{E}
    \cdot \frac{\gendist{\subexpr{E'}{a}}{N_i}}{\gendist{\subexpr{E}{a}}{N_i}}.
\end{align*}
\end{lemma}
\begin{proof}
Use $\prior{E} = \gendist{E}{S}$ and $\prior{E'} = \gendist{E'}{S}$, and expand
using Lemma~\ref{lemma:factor}.
\end{proof}

\begin{lemma}\label{lem:mcmc-pcfg-invariant}
The transition operator $\TransOper$ of Algorithm~\ref{alg:cfg-mcmc} satisfies
Condition~\ref{cond:invariance} (posterior invariance).
\end{lemma}
\begin{proof}
We show detailed balance for $\TransOper$ with respect to the posterior:
\begin{align*}
  \post{E} \cdot \TransOperX{E \to E'}
    = \post{E'} \cdot \TransOperX{E' \to E} && (E, E' \in \Lang).
\end{align*}
First, if $\post{E} \cdot \TransOperX{E \to E'} = 0$ then %
$\post{E'} \cdot \TransOperX{E' \to E} = 0$ as follows.
Either $\post{E} = 0$ or $\TransOperX{E \to E'} = 0$.
\begin{enumerate}
\item If $\post{E} = 0$ we have $\lik{E} = 0$ which implies that %
  $\alpha(E', E) = 0$ and therefore $\TransOperX{E' \to E} = 0$.

\item If $\TransOperX{E \to E'} = 0$ then $\TransOperXa{E \to E'}{\aroot} = 0$
  which implies either that the acceptance probability $\alpha(E, E') = 0$ or
  that $\gendist{\subexpr{E'}{\aroot}}{S} = 0$. %
  If $\alpha(E, E') = 0$ then $\lik{E'} = 0$ and $\post{E'} = 0$. But
  $\gendist{\subexpr{\aroot}{E'}}{S} = 0$ is a contradiction since
  $\gendist{\subexpr{\aroot}{E'}}{S} = \gendist{E'}{S} =
  \prior{E'} > 0$.
\end{enumerate}

Next, consider $\post{E} \cdot \TransOperX{E \to E'} > 0$ and $\post{E'} \cdot
\TransOperX{E' \to E} > 0$.
It suffices to show that $\post{E} \cdot |A_{E'}| \cdot \TransOperXa{E \to
E'}{a} = \post{E'} \cdot |A_E| \cdot \TransOperXa{E' \to E}{a}$ for all $a \in
A_E \cap A_{E'}$.
If $E = E'$, then this is vacuously true.
Otherwise $E \ne E'$, then consider two cases:
\begin{enumerate}
\item If $\sever{E}{a} = \sever{E'}{a} = (N_i, \Ehole)$ for some $i$ and
  $\Ehole$, then it suffices to show that:
\begin{align*}
&\post{E}
  \cdot |A_{E'}|
  \cdot \gendist{\subexpr{E'}{a}}{N_i}
  \cdot \alpha(E, E') \\
&\qquad= \post{E'}
  \cdot |A_E|
  \cdot
  \gendist{\subexpr{E}{a}}{N_i}
  \cdot \alpha(E', E)
\end{align*}
Both sides are guaranteed to be non-zero because $\post{E} > 0$ implies $\lik{E}
> 0$ implies $\alpha(E', E) > 0$, and by Lemma~\ref{lemma:expand-nonzero} (and
similarly for $\post{E'} > 0$).
Therefore, it suffices to show that
\begin{align*}
\frac{\alpha(E, E')}{\alpha(E', E)}
  &= \frac{\post{E'} \cdot |A_E| \cdot \gendist{\subexpr{E}{a}}{N_i}}
          {\post{E} \cdot |A_E'| \cdot \gendist{\subexpr{E'}{a}}{N_i}}
  \\
  &= \frac{\prior{E'} \cdot \lik{E'} \cdot |A_E| \cdot \gendist{\subexpr{E}{a}}{N_i}}
          {\prior{E} \cdot \lik{E} \cdot |A_E'| \cdot \gendist{\subexpr{E'}{a}}{N_i}}
  \\
  &= \frac{\lik{E'} \cdot |A_E|}{\lik{E} \cdot |A_E'|},
\end{align*}
where for the last step we used Lemma~\ref{lemma:cfg-mcmc-new-prior} to expand
$\prior{E'}$, followed by cancellation of factors.
Note that if $\alpha(E, E') < 1$ then $\alpha(E', E) = 1$ and
\begin{align*}
    \frac{\alpha(E, E')}{\alpha(E', E)}
    = \frac{\alpha(E, E')}{1}
    = \frac{\lik{E'} \cdot |A_E|}{\lik{E} \cdot |A_E'|}.
\end{align*}
If $\alpha(E', E) < 1$ then $\alpha(E, E') = 1$ and
\begin{align*}
    \frac{\alpha(E, E')}{\alpha(E', E)}
    = \frac{1}{\alpha(E', E)}
    = \frac{\lik{E'} \cdot |A_E|}{\lik{E} \cdot |A_E'|}.
\end{align*}
If $\alpha(E, E') = 1 = \alpha(E', E)$ then %
$\alpha(E, E') / \alpha(E', E) = 1 = \alpha(E, E')$.
\item If $\sever{E}{a} \ne \sever{E'}{a}$, then $\TransOperXa{E \to E'}{a} =
  \TransOperXa{E' \to E}{a} = 0$.
\end{enumerate}
If $E \ne E'$, then $\TransOperXa{E \to E'}{a} = \prior{\subexpr{E'}{a}} \cdot
\alpha(E, E')$.
Finally, posterior invariance follows from detailed balance:
\begin{align*}
  \textstyle\sum_{E \in \Lang} \post{E} \TransOperX{E \to E'}
  = \textstyle\sum_{E \in \Lang} \post{E'} \TransOperX{E' \to E}
  = \post{E'}.
\end{align*}
\end{proof}

\begin{lemma}\label{lem:mcmc-pcfg-irreducible}
The transition operator $\TransOper$ of Algorithm~\ref{alg:cfg-mcmc} satisfies
Condition~\ref{cond:irreducibility} (posterior irreducibility).
\end{lemma}
\begin{proof}
We show that for all expressions $E$ and $E' \in \{\Lang : \post{E'} > 0\}$,
$E'$ is reachable from $E$ in one step (that is, $\TransOperX{E \to E'} > 0$).
Note for all $E, E' \in \Lang$, $\aroot \in A_E \cap A_{E'}$ and
$\sever{E}{\aroot} = \sever{E'}{\aroot} = (S, \square)$.
Since $\post{E'} > 0$, we know that $\prior{E'} = \gendist{E'}{S} > 0$, and
$\lik{E'} > 0$, which implies $\alpha(E, E') > 0$.
Finally:
\begin{align*}
  \TransOperX{E \to E'}
    \ge \frac{1}{|A_E|} \cdot \TransOperXa{E \to E'}{a}
    \ge \gendist{E'}{S} \cdot \alpha(E, E') > 0.
\end{align*}
\end{proof}

\begin{lemma}\label{lem:mcmc-pcfg-aperiodic}
The transition operator $\TransOper$ of Algorithm~\ref{alg:cfg-mcmc} satisfies
Condition~\ref{cond:aperiodicity} (aperiodicity).
\end{lemma}
\begin{proof}
For all $E \in \Lang$, we have
$\TransOperX{E \to E}
  \ge (1/|A_E|) \cdot \gendist{E}{S}
  = (1/|A_E|) \cdot \prior{E} > 0
$.
The inequality derives from choosing only $a = \aroot$ in the sum over $a \in
|A_E|$.
\end{proof}

We have thus established (by Lemmas~\ref{lem:mcmc-pcfg-invariant},
\ref{lem:mcmc-pcfg-irreducible}, and \ref{lem:mcmc-pcfg-aperiodic}) that for a
language $\Lang$ specified by a consistent probabilistic context-free grammar
with well-defined $\Prior$ and $\Likelihood$ semantics
(Conditions~\ref{cond:prior-normalized}, \ref{cond:likelihood-normalized}, and
\ref{cond:likelihood-bounded}) the Bayesian synthesis procedure in
Algorithm~\ref{alg:mcmc} with the transition operator $\TransOper$ from
Algorithm~\ref{alg:cfg-mcmc} satisfies the preconditions for sound inference
given in Theorem~\ref{thm:mcmc-convergence}.


\section{Gaussian Process DSL for Univariate Time Series Data}
\label{sec:dsl-time-series}


\begin{figure}[bp]
\scriptsize
\hrule
\bigskip

\begin{subfigure}[t]{.5\linewidth}
\textbf{\footnotesize Syntax and Production Rules}
\begin{align*}
v &\in \DomReals\; \\
H &\in \DomParameters \Coloneqq \ttpl \texttt{gamma}\, v \ttpr
  && \mathrm{[GammaParam]} \\
K &\in \DomKernel \Coloneqq
  \ttpl \texttt{const}\; H \ttpr                  && \mathrm{[Constant]} \\
  & \mid \ttpl \texttt{wn}\; H \ttpr              && \mathrm{[WhiteNoise]} \\
  & \mid \ttpl \texttt{lin}\; H \ttpr             && \mathrm{[Linear]}\\
  & \mid \ttpl \texttt{se}\; H \ttpr              && \mathrm{[SquaredExp]} \\
  & \mid \ttpl \texttt{per}\; H_1\; H_2 \ttpr     && \mathrm{[Periodic]} \\
  & \mid \ttpl \texttt{+}\; K_1\; K_2 \ttpr       && \mathrm{[Sum]} \\
  & \mid \ttpl \texttt{*}\; K_1\; K_2 \ttpr       && \mathrm{[Product]} \\
  & \mid \ttpl \texttt{cp}\; H\; K_1\; K_2 \ttpr  && \mathrm{[ChangePoint]}
\end{align*}
\end{subfigure}%
\begin{subfigure}[t]{.5\linewidth}
\textbf{\footnotesize Production Rule Probabilities}
\begin{align*}
H &:P(\texttt{gamma}) \Coloneqq 1.\\
K &:P(\texttt{const}) = P(\texttt{wn}) = P(\texttt{lin}) = P(\texttt{se}) \Coloneqq 0.14,\\
&\qquad P(\texttt{+}) = P(\texttt{*}) \Coloneqq 0.135,\, P(\texttt{cp}) \Coloneqq 0.03.
\end{align*}
\smallskip

\textbf{\footnotesize Terminal Symbol Probabilities}
\begin{align*}
\texttt{gamma}:
  Q(\texttt{gamma}, v) &\Coloneqq
    \frac{\beta^\alpha v^{\alpha-1}e^{-\beta{v}}}{\Gamma(\alpha)}
    \; (\alpha=1, \beta=1).
\end{align*}
\smallskip

\textbf{\footnotesize Prior Denotation}

Default prior for probabilistic context-free grammars
(Section~\ref{subsec:pcfg-prior}).
\end{subfigure}

\bigskip

\begin{subfigure}{\linewidth}
\textbf{\footnotesize Likelihood Denotation}
\begin{align*}
\Denotation[\Covar]{(\texttt{const}\; (\texttt{gamma}\, v))}(x)(x')
  &\Coloneqq v \\
\Denotation[\Covar]{(\texttt{wn}\; (\texttt{gamma}\, v))}(x)(x')
  &\Coloneqq v \cdot \mathbb{I}[x = x'] \\
\Denotation[\Covar]{(\texttt{lin}\; (\texttt{gamma}\, v))}(x)(x')
  &\Coloneqq (x-v)(x'-v) \\
\Denotation[\Covar]{(\texttt{se}\; (\texttt{gamma}\, v))}(x)(x')
  &\Coloneqq \exp({-(x-x')^2/v}) \\
\Denotation[\Covar]{(\texttt{per}\;
  (\texttt{gamma}\, v_1)\; (\texttt{gamma}\, v_2))}(x)(x')
  &\Coloneqq \exp(-2/v_1 \sin((2\pi/v_2)\abs{x-x'})^2) \\
\Denotation[\Covar]{(\texttt{+}\; K_1\; K_2)}(x)(x')
  &\Coloneqq \Denotation[\Covar]{K_1}(x)(x') + \Denotation[\Covar]{K_2}(x)(x') \\
\Denotation[\Covar]{(\texttt{*}\; K_1\; K_2)}(x)(x')
  &\Coloneqq \Denotation[\Covar]{K_1}(x)(x') \times \Denotation[\Covar]{K_2}(x)(x') \\
\Denotation[\Covar]{(\texttt{cp}\; (\texttt{gamma}\, v)\; K_1\; K_2)}(x)(x')
  &\Coloneqq
    \Delta(x,v)\Delta(x',v)(\Denotation[\Covar]{K_1}(x)(x'))
    + (1-\Delta(x,v))(1-\Delta(x',v))(\Denotation[\Covar]{K_1}(x)(x'))\\
  &\qquad \mbox{ where } \Delta(x,v) \Coloneqq 0.5\times(1+\tanh(10(x-v)))\\
%
\Denotation[\Likelihood]{K}((\mathbf{x}, \mathbf{y})) &\Coloneqq
  \exp\Big(-1/2
    \textstyle \sum_{i=1}^{n}y_i\left(
    \textstyle\sum_{j=1}^{n}(\lbrace
      [\Denotation[\Covar]{K}(x_i)(x_j) +
        0.01\delta(x_i,x_j)]_{i,j=1}^{n}
      \rbrace^{-1}_{ij})y_j
    \right) \\
  &\qquad
  -1/2\log\left\lvert{[\Denotation[\Covar]{K}(x_i)(x_j)
    + 0.01\delta(x_i,x_j)]_{i,j=1}^{n}}\right\rvert
  - (n/2)\log{2\pi}
  \Big)
\end{align*}
\end{subfigure}

\begin{subfigure}{\linewidth}
\textbf{\footnotesize DSL to Venture Translation}
\begin{align*}
\Denotation[\VKernel]{(\texttt{const}\; (\texttt{gamma}\, v))}
  &\Coloneqq \texttt{((x1, x2) -> \{$v$\})} \\
\Denotation[\VKernel]{(\texttt{wn}\; (\texttt{gamma}\, v))}
  &\Coloneqq \texttt{((x1, x2) -> \{if (x1==x2) \{$v$\} else \{0\}\})} \\
\Denotation[\VKernel]{(\texttt{lin}\; (\texttt{gamma}\, v))}
  &\Coloneqq \texttt{((x1, x2) -> \{(x1-$v$) * (x2-$v$)\})} \\
\Denotation[\VKernel]{(\texttt{se}\; (\texttt{gamma}\, v))}
  &\Coloneqq \texttt{((x1, x2) -> \{exp((x1-x2)**2/$v$)\})} \\
\Denotation[\VKernel]{(\texttt{per}\;
  (\texttt{gamma}\, v_1)\; (\texttt{gamma}\, v_2))}(x)(x')
  &\Coloneqq \texttt{((x1, x2) -> \{-2/$v_1$ * sin(2*pi/$v_2$ * abs(x1-x2))**2\})} \\
\Denotation[\VKernel]{(\texttt{+}\; K_1\; K_2)}
  &\Coloneqq \texttt{((x1, x2) ->
    \{$\Denotation[\VKernel]{K_1}$(x1, x2)
    + $\Denotation[\VKernel]{K_2}$(x1, x2)\})
  } \\
\Denotation[\VKernel]{(\texttt{*}\; K_1\; K_2)}
  &\Coloneqq \texttt{((x1, x2) ->
    \{$\Denotation[\VKernel]{K_1}$(x1, x2)
    * $\Denotation[\VKernel]{K_2}$(x1, x2)\})
  } \\
\Denotation[\VKernel]{((\texttt{cp}\; (\texttt{gamma}\, v)\; K_1\; K_2)}
  &\Coloneqq \texttt{((x1, x2) -> \{} \\
  &\qquad \texttt{sig1 = sigmoid(x1, $v$, .1) * sigmoid(x2, $v$, .1); } \\
  &\qquad \texttt{sig2 = (1-sigmoid(x1, $v$, .1)) * (1-sigmoid(x2, $v$, .1)); } \\
  &\qquad \texttt{sig1 * $\Denotation[\VKernel]{K_1}$(x1,x2)
    + sig2 * $\Denotation[\VKernel]{K_2}$(x1,x2)\})} \\
%
\Denotation[\VProg]{K}
  &\Coloneqq \texttt{assume gp = gaussian\_process(gp\_mean\_constant(0), $\Denotation[\VKernel]{K}$);}
\end{align*}
\end{subfigure}
\bigskip
\hrule

\captionsetup{skip=5pt}
\caption{Components of domain-specific language for automatic data
modeling of time series using Gaussian process models. This DSL belongs to
the family of probabilistic context-free grammars from Section~\ref{sec:pcfg}.}
\label{fig:dsl-gp}
\end{figure}

Section~\ref{subsec:example-gp-introduction} introduced the mathematical
formalism and domain-specific language for Gaussian process models.
The DSL is generated by a tagged probabilistic context-free grammar as described
in Definition~\ref{def:tagged-pcfg}, which allows us to reuse the general
semantics and inference algorithms in Section~\ref{sec:pcfg}.
Figure~\ref{fig:dsl-gp} shows the core components of the domain-specific
language: the syntax and production rules; the symbol and production rule
probabilities; the likelihood semantics; and the translation procedure of
DSL expressions into Venture.
By using addition, multiplication, and change point operators to build composite
Gaussian process kernels, it is possible to discover a wide range of time series
structure from observed data~\citep{lloyd2014}.
The synthesis procedure uses Algorithm~\ref{alg:mcmc} along the generic
transition operator for PCFGs in Algorithm~\ref{alg:cfg-mcmc}, which is sound by
Lemmas~\ref{lem:mcmc-pcfg-invariant}, %
\ref{lem:mcmc-pcfg-irreducible}, and \ref{lem:mcmc-pcfg-aperiodic}.
Appendix~\ref{appendix:gp-dsl-proofs} formally proves that the prior and
likelihood semantics satisfy the preconditions for Bayesian synthesis.
%
%

This Gaussian process domain-specific language (whose implementation is shown in
Figure~\ref{fig:gp-tutorial}) is suitable for both automatic model structure
discovery of the covariance kernel as well as for estimating the parameters of
each kernel.
This contrasts with standard software packages for inference in GPs (such as
python's scikit-learn~\citep{pedergosa2011}, or MATLAB's
GPML~\citep{rasmussen2010}) which do not have the ability to synthesize
different Gaussian process model structures and can only perform parameter
estimation given a fixed, user-specified model structure which is based on
either the user's domain-specific knowledge or (more often) an arbitrary choice.

\subsection{Time Complexity of Bayesian Synthesis Algorithm}
\label{subsec:dsl-gp-complexity}

In Algorithm~\ref{alg:cfg-mcmc} used for synthesis, each transition step from an
existing expression $K$ to a candidate new expression $K'$ requires the
following key computations:

\begin{enumerate}
\item Severing the parse tree at a randomly selected node using $\Sever$
    (Algorithm~\ref{alg:cfg-mcmc}, line \ref{algline:cfg-mcmc-sever}).
    The cost is linear in the number $\abs{K}$ of subexpressions in
    the overall s-expression $K$.

\item Forward-sampling the probabilistic context-free grammar using $\GenDist$
    (Algorithm~\ref{alg:cfg-mcmc}, line~\ref{algline:cfg-mcmc-expand}). The
    expected cost is linear in the average length of strings $l_G$
    generated by the PCFG, which can be determined
    from the transition probabilities \citep[Eq.~3]{gecse2010}).

\item Assessing the likelihood under the existing and proposed expressions
    using $\Denotation[\Likelihood]{K}((\mathbf{x}, \mathbf{y}))$
    (Algorithm~\ref{alg:cfg-mcmc}, line~\ref{algline:cfg-mcmc-likelihood}).
    For a Gaussian process with $n$ observed time points, the cost
    of building the covariance matrix $\Denotation[\Covar]{K}$
    is $O(\abs{K}n^2)$ and the cost of obtaining its inverse is $O(n^3)$.
\end{enumerate}

The overall time complexity of a transition from $K$ to $K'$ is thus
$O(l_G + \max(\abs{K}, \abs{K'})n^2 + n^3)$.
This term is typically dominated by the $n^3$ term from assessing $\Likelihood$.
Several existing techniques use sparse approximations of Gaussian processes to
reduce this complexity to $O(m^2n)$ where $m \ll n$ is a parameter representing
the size of a subsample of the full data to be used
\citep[Chapter~8]{rasmussen2006}.
These approximation techniques trade-off predictive accuracy with significant
increase in scalability.
\citet{quinonero2005} show that sparse Gaussian process approximation methods
typically correspond to well-defined probabilistic models that have different
likelihood functions than the standard Gaussian process, and so our synthesis
framework can naturally incorporate sparse Gaussian processes by adapting the
$\Likelihood$ semantics.

As for the quadratic factor $O(\max(\abs{K}, \abs{K'})n^2)$, this scaling
depends largely on the characteristics of the underlying time series data. For
simple time series with only a few patterns we expect the program size $\abs{K}$
to be small, whereas for time series with several complex temporal patterns then
the synthesized covariance expressions $K$ are longer to obtain strong fit to
the data.


\begin{figure}[!b]
\scriptsize
\hrule
\bigskip

\begin{subfigure}[t]{.5\linewidth}
\textbf{\footnotesize Syntax and Production Rules}
\begin{align*}
x,y,w &\in \DomReals, a \in [m], s \in [n] && \\
P \in \DomPartition
  &\Coloneqq \ttpl \texttt{partition}\, B_1\, \dots\, B_k \ttpr
  && \mathrm{[Partition]} \\
B \in \DomBlock
  &\Coloneqq \ttpl \texttt{block}\,
    (a_1\, \dots\, a_l)\, C_1\, \dots\, C_t \ttpr
  && \mathrm{[Block]} \\
C \in \DomCluster
  &\Coloneqq \ttpl \texttt{cluster}\, s\, V_1\, V_2\, \dots\, V_l \ttpr
  && \mathrm{[Cluster]} \\
V \in \DomVariable
  &\Coloneqq \ttpl \texttt{var}\, a\, D \ttpr
  && \mathrm{[Variable]} \\
D \in \DomDist
  &\Coloneqq \ttpl \texttt{normal}\, (x\, y) \ttpr
    && \mathrm{[Gaussian]} \\
  &\mid \ttpl \texttt{poisson}\, y \ttpr
    && \mathrm{[Poisson]} \\
  &\mid \ttpl \texttt{categorical}\, w_1\, \dots\, w_q \ttpr
    && \mathrm{[Categorical]}
\end{align*}

\textbf{\footnotesize Prior Denotation}
\begin{align*}
%
%
\Denotation[\Prior]{(\texttt{partition}\, B_1\, \dots\, B_k)}
  &\Coloneqq \frac{\textstyle\prod_{i=1}^{k} \Denotation[\Prior]{B_k}}{m!}\\
\Denotation[\Prior]{(\texttt{block}\, (a_1\, \dots\, a_l)\, C_1\, \dots\, C_t)}
  &\Coloneqq (l-1)! \, \frac{\textstyle\prod_{i=1}^{t}\Denotation[\Prior]{C_i}}{n!}\\
\Denotation[\Prior]{(\texttt{cluster}\, s\, V_1\, \dots\, V_l)}
  &\Coloneqq (s-1)! \, \textstyle\prod_{i=1}^{l}\Denotation[\Prior]{V_i} \\
\Denotation[\Prior]{(\texttt{var}\, a\, D)}
  &\Coloneqq \Denotation[\Prior]{D} \\
\Denotation[\Prior]{(\texttt{normal}\, v\, y)}
  &\Coloneqq
    \sqrt{\frac{\lambda}{y^22\pi}}
    \frac{\beta^{\alpha}}{\Gamma(\alpha)}\left(\frac{1}{y^2}\right)^{\alpha+1}\\
  &\;
    \exp\left(\frac{-(2\beta+\lambda(v-\eta)^2)}{2y^2} \right) \\
\Denotation[\Prior]{(\texttt{poisson}\, y)}
  &\Coloneqq \frac{\xi^\nu y^{\nu-1}e^{-\xi{y}}}{\Gamma(\nu)} \\
\Denotation[\Prior]{(\texttt{categorical}\, w_1\, \dots\, w_q)}
  &\Coloneqq \frac{\Gamma(\kappa)^q}{\Gamma(\kappa)}
    \textstyle\prod_{i=1}^{q}w^{\alpha}_{i} \\
\alpha, \beta, \lambda, \eta, \xi, \nu,
  \kappa &\Coloneqq (\textrm{statistical constants})
\end{align*}

\textbf{\footnotesize Likelihood Denotation}
\begin{align*}
%
\Denotation[\Likelihood]{(\texttt{partition}\, B_1\, \dots\, B_k)}(\mathbf{X})
  &\Coloneqq \textstyle\prod_{i=1}^{k}
    \Denotation[\Likelihood]{B_k}(\mathbf{X}) \\
\Denotation[\Likelihood]{(\texttt{block}\, (a_1\, \dots\, a_l)\, C_1\, \dots\, C_t)}(\mathbf{X})
  &\Coloneqq
    \textstyle\prod_{i=1}^{n}\sum_{j=1}^{t}
    \Denotation[\Likelihood]{C_j}(\mathbf{X}_{i})\\
\Denotation[\Likelihood]{(\texttt{cluster}\, s\, V_1\, \dots\, V_l)}(\mathbf{x})
  &\Coloneqq \frac{s}{n} \; \textstyle\prod_{i=1}^{l}\Denotation[\Likelihood]{V_i}(\mathbf{x}) \\
\Denotation[\Likelihood]{(\texttt{var}\, a\, D)}(\mathbf{x})
  &\Coloneqq \Denotation[\Likelihood]{D}(\mathbf{x}_a) \\
\Denotation[\Likelihood]{(\texttt{normal}\, v\, y)}(x)
  &\Coloneqq \frac{1}{\sqrt{2\pi{y}^2}}
    e^{-\left(\frac{x-v}{\sqrt{2}y}\right)^2}\\
\Denotation[\Likelihood]{(\texttt{poisson}\, y)}(x)
  &\Coloneqq y^xe^{-y}/x! \\
\Denotation[\Likelihood]{(\texttt{categorical}\, w_1\, \dots\, w_q)}(x)
  &\Coloneqq w_{x}
\end{align*}
\end{subfigure}\hfill%
\begin{subfigure}[t]{.45\linewidth}
\textbf{\footnotesize DSL to Venture Translation}
\medskip

\underline{\footnotesize Probabilistic Program in DSL}
\begin{lstlisting}[style=dsledits,basicstyle=\ttfamily\scriptsize]
(partition
 (block (1)
  (cluster 6 (var 1 (normal 0.6 2.1)))
  (cluster 4 (var 1 (normal 0.3 1.7))))
 (block (2 3)
  (cluster 2 (var 2 (normal 7.6 1.9)
             (var 3 (poisson 12))))
  (cluster 3 (var 2 (normal 1.1 0.5)
             (var 3 (poisson 1))))
  (cluster 5 (var 2 (normal -0.6 2.9)
             (var 3 (poisson 4))))))
\end{lstlisting}
\bigskip

\underline{\footnotesize Probabilistic Program in Venture}
\begin{lstlisting}[style=venturescript, basicstyle=\ttfamily\scriptsize]
assume block1_cluster =
  categorical(simplex([0.6, 0.4])) #block:1;

assume var1 = cond(
  (block1_cluster == 0) (normal(0.6, 2.1))
  (block1_cluster == 1) (normal(0.3, 1.7)));

assume block2_cluster =
  categorical(simplex([0.2, 0.3, 0.5])) #block:2;

assume var2 = cond(
  (block2_cluster == 0) (normal(7.6, 1.9))
  (block2_cluster == 1) (normal(1.1, 0.5))
  (block2_cluster == 2) (normal(-0.6, 2.9)));

assume var3 = cond(
  (block2_cluster == 0) (poisson(12))
  (block2_cluster == 1) (poisson(1))
  (block2_cluster == 2) (poisson(4)));
\end{lstlisting}
\bigskip

\underline{Prediction Query in Venture}
\begin{lstlisting}[style=venturescript, basicstyle=\ttfamily]
observe var3 = 8;
infer repeat(100, {gibbs(quote(block), one)});
sample [var1, var2];
\end{lstlisting}
\end{subfigure}
\bigskip
\hrule

\captionsetup{skip=5pt}
\caption{Components of domain-specific language for automatic data
modeling of multivariate tabular data using nonparametric mixture models.
This DSL is context-sensitive and contains a custom prior denotation.}
\label{fig:dsl-crosscat}
\end{figure}


\definecolor{EditDeleteColor}{RGB}{255,175,175}
\definecolor{EditAddColor}{RGB}{200,225,255}

\begin{figure}[t]
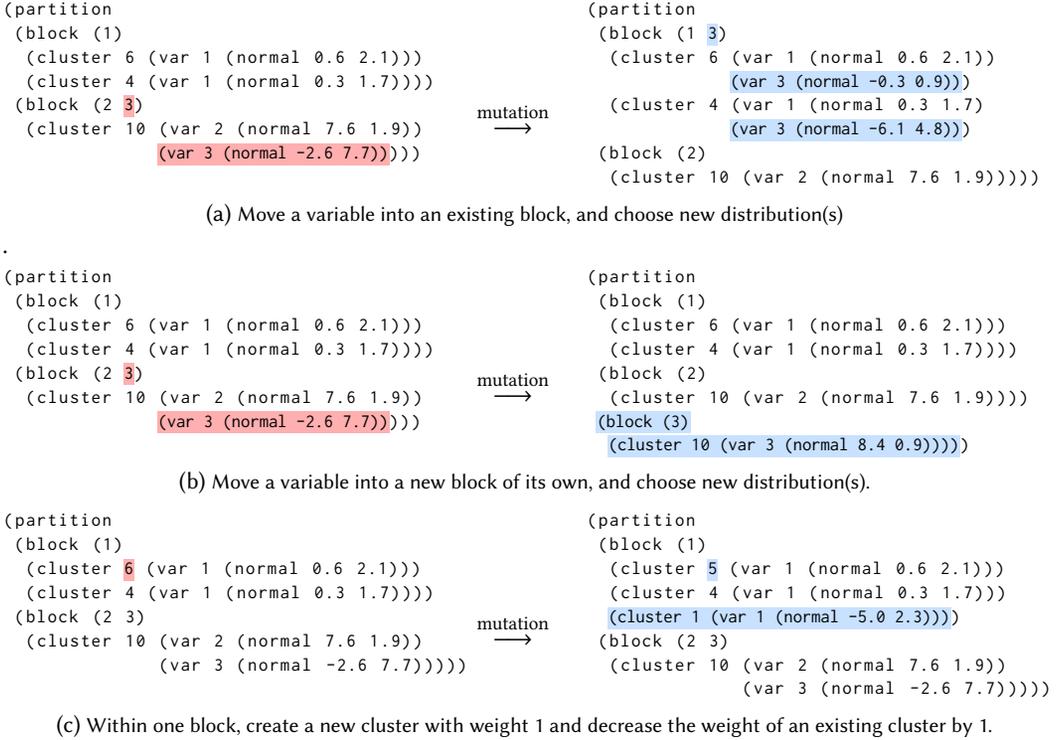

\centering


\begin{subfigure}{1.0\linewidth}
\begin{minipage}[t]{0.45\textwidth}
\sethlcolor{EditDeleteColor}
\begin{lstlisting}[style=dsledits,basicstyle=\ttfamily\scriptsize]
(partition
 (block (1)
  (cluster 6 (var 1 (normal 0.6 2.1)))
  (cluster 4 (var 1 (normal 0.3 1.7))))
 (block (2 (*@\hl{3}@*))
  (cluster 10 (var 2 (normal 7.6 1.9))
              (*@\hl{(var 3 (normal -2.6 7.7))}@*))))
\end{lstlisting}
\end{minipage}
\begin{minipage}[t]{0.1\textwidth}
\vspace{1.5cm}
$\overset{\rm mutation}{\longrightarrow}$
\end{minipage}
\begin{minipage}[t]{0.55\textwidth}
\sethlcolor{EditAddColor}
\begin{lstlisting}[style=dsledits,basicstyle=\ttfamily\scriptsize]
(partition
 (block (1 (*@\hl{3}@*))
  (cluster 6 (var 1 (normal 0.6 2.1))
             (*@\hl{(var 3 (normal -0.3 0.9))}@*))
  (cluster 4 (var 1 (normal 0.3 1.7)
             (*@\hl{(var 3 (normal -6.1 4.8))}@*))
 (block (2)
  (cluster 10 (var 2 (normal 7.6 1.9)))))
\end{lstlisting}
\end{minipage}
\captionsetup{skip=0pt}
\caption{\footnotesize Move a variable into an existing block, and choose new distribution(s)}.
\end{subfigure}


\begin{subfigure}{1.0\linewidth}
\begin{minipage}[t]{0.45\textwidth}
\sethlcolor{EditDeleteColor}
\begin{lstlisting}[style=dsledits,basicstyle=\ttfamily\scriptsize]
(partition
 (block (1)
  (cluster 6 (var 1 (normal 0.6 2.1)))
  (cluster 4 (var 1 (normal 0.3 1.7))))
 (block (2 (*@\hl{3}@*))
  (cluster 10 (var 2 (normal 7.6 1.9))
              (*@\hl{(var 3 (normal -2.6 7.7))}@*))))
\end{lstlisting}
\end{minipage}
\begin{minipage}[t]{0.1\textwidth}
\vspace{1.5cm}
$\overset{\rm mutation}{\longrightarrow}$
\end{minipage}
\begin{minipage}[t]{0.55\textwidth}
\sethlcolor{EditAddColor}
\begin{lstlisting}[style=dsledits,basicstyle=\ttfamily\scriptsize]
(partition
 (block (1)
  (cluster 6 (var 1 (normal 0.6 2.1)))
  (cluster 4 (var 1 (normal 0.3 1.7))))
 (block (2)
  (cluster 10 (var 2 (normal 7.6 1.9))))
 (*@\hl{(block (3)}@*)
  (*@\hl{(cluster 10 (var 3 (normal 8.4 0.9))))}@*))
\end{lstlisting}
\end{minipage}
\captionsetup{skip=0pt}
\caption{\footnotesize Move a variable into a new block of its own, and choose new distribution(s).}
\end{subfigure}

\begin{subfigure}{1.0\linewidth}
\begin{minipage}[t]{0.45\textwidth}
\sethlcolor{EditDeleteColor}
\begin{lstlisting}[style=dsledits,basicstyle=\ttfamily\scriptsize]
(partition
 (block (1)
  (cluster (*@\hl{6}@*) (var 1 (normal 0.6 2.1)))
  (cluster 4 (var 1 (normal 0.3 1.7))))
 (block (2 3)
  (cluster 10 (var 2 (normal 7.6 1.9))
              (var 3 (normal -2.6 7.7)))))
\end{lstlisting}
\end{minipage}
\begin{minipage}[t]{0.1\textwidth}
\vspace{1.5cm}
$\overset{\rm mutation}{\longrightarrow}$
\end{minipage}
\begin{minipage}[t]{0.45\textwidth}
\sethlcolor{EditAddColor}
\begin{lstlisting}[style=dsledits,basicstyle=\ttfamily\scriptsize]
(partition
 (block (1)
  (cluster (*@\hl{5}@*) (var 1 (normal 0.6 2.1)))
  (cluster 4 (var 1 (normal 0.3 1.7)))
  (*@\hl{(cluster 1 (var 1 (normal -5.0 2.3)))}@*))
 (block (2 3)
  (cluster 10 (var 2 (normal 7.6 1.9))
              (var 3 (normal -2.6 7.7)))))
\end{lstlisting}
\end{minipage}
\captionsetup{skip=0pt}
\caption{\footnotesize Within one block, create a new cluster with weight 1 and decrease the weight of an existing cluster by 1.}
\end{subfigure}

\captionsetup{skip=5pt}
\caption{Examples of mutation operators applied to a mixture modeling DSL program during
Bayesian synthesis.}
\vspace{-.175in}
\label{fig:crosscat-edits}
\end{figure}

\section{Nonparametric Mixture Model DSL for Multivariate Tabular Data}
\label{sec:dsl-crosscat}

We now describe a second DSL for multivariate tabular data called MultiMixture. The
setup starts with a data table containing $m$ columns and $n$ rows. Each column
$c$ represents a distinct random variable $X_c$ and each row $r$ is a joint
instantiation $\set{x_{r1},\dots,x_{rm}}$ of all the $m$ random variables. Our
DSL describes the data generating process for cells in the table based on the
nonparametric ``Cross-Categorization'' mixture model introduced by
\citet{mansinghka2014}, which we review.

Figure~\ref{fig:dsl-crosscat} shows the components required for Bayesian
synthesis in the MultiMixture DSL.
Probabilistic programs in this DSL assert that the set of columns $[m]
\coloneqq \set{1,\dots,m}$ is partitioned into $k$ non-overlapping blocks $B_i$,
which is specified by the $\texttt{partition}$ production rule. Each
$\texttt{block}$ contains a set of unique columns $\set{a_1,\dots,a_l}
\subset [m]$, where a particular column must appear in exactly one block of the
\texttt{partition}.
Each \texttt{block} has clusters $C_i$, where each cluster specifies a
component in a mixture model. Each \texttt{cluster} has relative
weight $s$, a set of \texttt{var} objects that specify a variable index
$a$, and a primitive distribution $D$. The distributions are \texttt{normal}
(with mean $x$ and variance $y$); \texttt{poisson} (with mean $y$); and
\texttt{categorical} (with category weights $w_i$).

Since the MultiMixture DSL is not context-free, we specify custom $\Prior$ semantics that
recursively describes the joint probability of all terms in an expression.
These probabilities provide a hierarchical decomposition of the distribution in
\citep[Section~2.2]{mansinghka2014} based on the structure of the DSL program.
Prior normalization (Condition~\ref{cond:prior-normalized}) holds since the five
production rules in the grammar are non-recursive.
The $\Likelihood$ semantics break down the full probability of a data table of
observations $\mathbf{X}$ into the cell-wise probabilities within each mixture
component. The likelihood is bounded (Condition~\ref{cond:likelihood-bounded})
since \texttt{poisson} and \texttt{categorical} likelihoods are bounded by
one, and $\Denotation[\Prior]{(\texttt{normal}\, v\, y)}$ is conjugate to
$\Denotation[\Likelihood]{(\texttt{normal}\, v\, y)}$ \citep{bernardo1994}.
The semantics of MultiMixture are designed to improve on the model
family from \citet{mansinghka2014} by making it (i) easier to embed the model in a
probabilistic language without enforcing complicated exchangeable coupling
(i.e.~the Venture programs for modeling and prediction in
Figure~\ref{fig:dsl-crosscat}); and (ii) possible to write a recursive,
decomposable semantics for the prior and likelihood.
%

Since the MultiMixture DSL is not generated by a context-free grammar, we employ custom
synthesis algorithms using the cycle of Gibbs kernels given in
\citet[Section~2.4]{mansinghka2016}.
Figure~\ref{fig:crosscat-edits} shows several examples of program mutations that
occur over the course of synthesis to sample from the approximate posterior
distribution over DSL expressions.
For a given expression $E \in \Lang$, a full sweep over all program mutations
shown in Figure~\ref{fig:crosscat-edits} has time complexity $O(mnkt)$, where
$m$ is the number of columns, $n$ is the number of rows, $k$ is the number of
$\texttt{block}$ subexpressions, and $t$ is the maximum number of
\texttt{cluster} subexpressions under any \texttt{block} subexpression. Full
details of this algorithmic analysis can be found in
\citet[Section~2.4]{mansinghka2016}.



\begin{figure}[!t]
\begin{subtable}{.5\linewidth}
\footnotesize
\begin{tabular}{cll}
\multirow{4}{*}{\includegraphics[width=.45\linewidth]{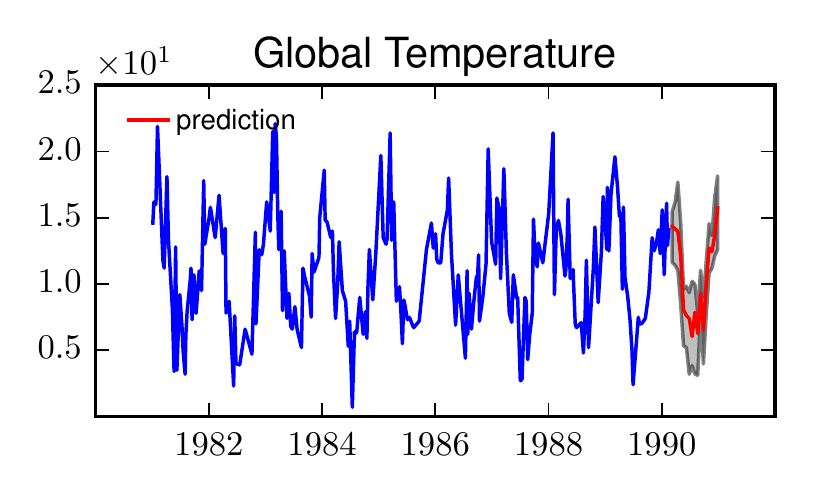}}
                & Temporal Structure & $p^{\scriptscriptstyle{\rm synth}}$ \\ \cline{2-3} \addlinespace[0.5em]
       & $\checkmark$ White Noise  & 100\%     \\
       & $\times$ Linear Trend     & 16\%      \\
       & $\checkmark$ Periodicity  & 92\%      \\
       & $\times$ Change Point     & 4\%
\end{tabular}
\captionsetup{skip=0pt}
\caption{}
\label{subfig:gp-structure-temperature}
\end{subtable}%
\begin{subtable}{.5\linewidth}
\footnotesize
\begin{tabular}{cll}
\multirow{4}{*}{\includegraphics[width=.45\linewidth]{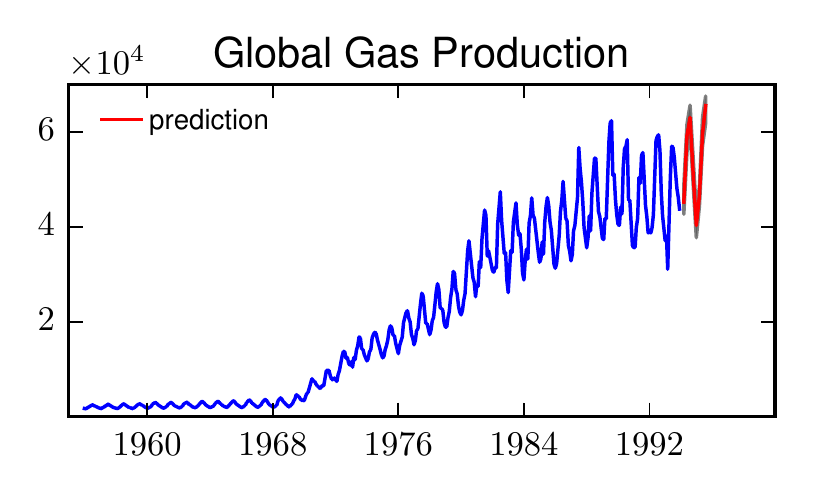}}
                & Temporal Structure & $p^{\scriptscriptstyle{\rm synth}}$ \\ \cline{2-3} \addlinespace[0.5em]
        & $\checkmark$ White Noise  & 100\%     \\
        & $\checkmark$ Linear Trend & 85\%      \\
        & $\checkmark$ Periodicity  & 76\%      \\
        & $\checkmark$ Change Point & 76\%
\end{tabular}
\captionsetup{skip=0pt}
\caption{}
\label{subfig:gp-structure-wages}
\end{subtable}
\medskip

\begin{subtable}{.5\linewidth}
\footnotesize
\begin{tabular}{cll}
\multirow{4}{*}{\includegraphics[width=.45\linewidth]{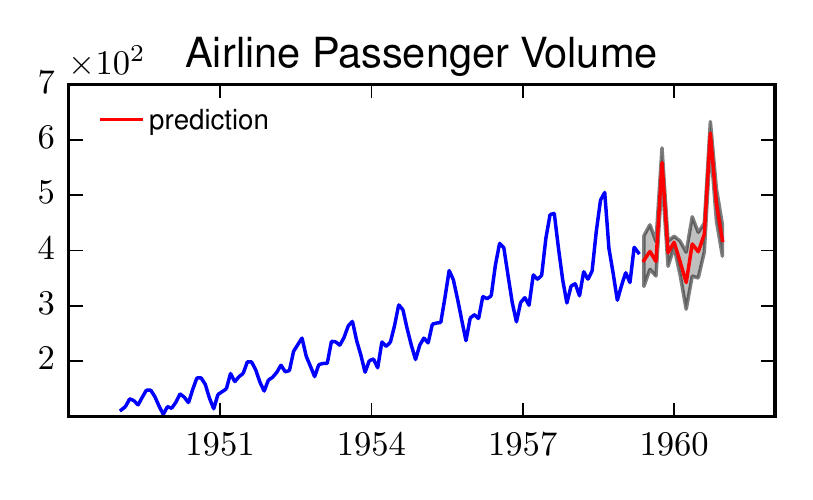}}
                & Temporal Structure & $p^{\scriptscriptstyle{\rm synth}}$ \\ \cline{2-3} \addlinespace[0.5em]
        & $\checkmark$ White Noise  & 100\%     \\
        & $\checkmark$ Linear Trend  & 95\%     \\
        & $\checkmark$ Periodicity  & 95\%      \\
        & $\times$ Change Point     & 22\%
\end{tabular}
\captionsetup{skip=0pt}
\caption{}
\label{subfig:gp-structure-airline}
\end{subtable}%
\begin{subtable}{.5\linewidth}
\footnotesize
\begin{tabular}{cll}
\multirow{4}{*}{\includegraphics[width=.45\linewidth]{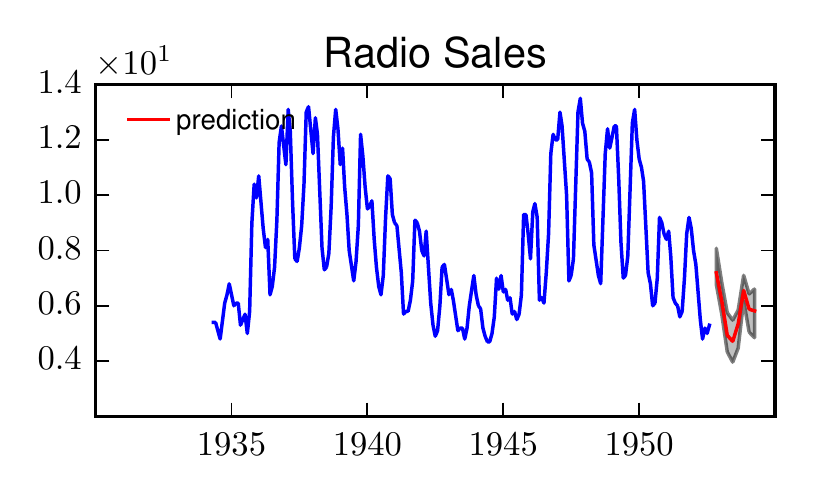}}
                & Temporal Structure & $p^{\scriptscriptstyle{\rm synth}}$ \\ \cline{2-3} \addlinespace[0.5em]
        & $\checkmark$ White Noise   & 100\%   \\
        & $\times$ Linear Trend      & 6\%     \\
        & $\checkmark$ Periodicity   & 93\%    \\
        & $\times$ Change Point      & 23\%
\end{tabular}
\captionsetup{skip=0pt}
\caption{}
\label{subfig:gp-structure-gas}
\end{subtable}
\medskip

\begin{subtable}{.5\linewidth}
\footnotesize
\begin{tabular}{cll}
\multirow{4}{*}{\includegraphics[width=.45\linewidth]{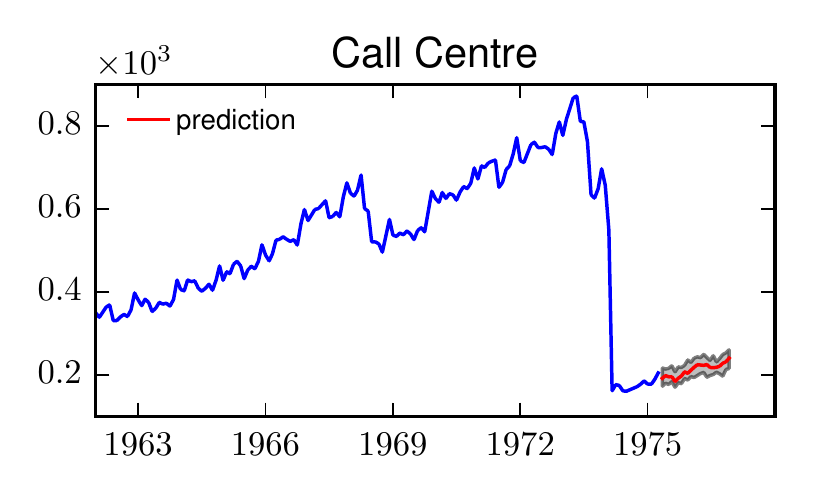}}
                & Temporal Structure & $p^{\scriptscriptstyle{\rm synth}}$ \\ \cline{2-3} \addlinespace[0.5em]
        & $\checkmark$ White Noise  & 100\%     \\
        & $\checkmark$ Linear Trend & 93\%      \\
        & $\checkmark$ Periodicity  & 97\%      \\
        & $\checkmark$ Change Point & 90\%
\end{tabular}
\captionsetup{skip=0pt}
\caption{}
\label{subfig:gp-structure-centre}
\end{subtable}%
\begin{subtable}{.5\linewidth}
\footnotesize
\begin{tabular}{cll}
\multirow{4}{*}{\includegraphics[width=.45\linewidth]{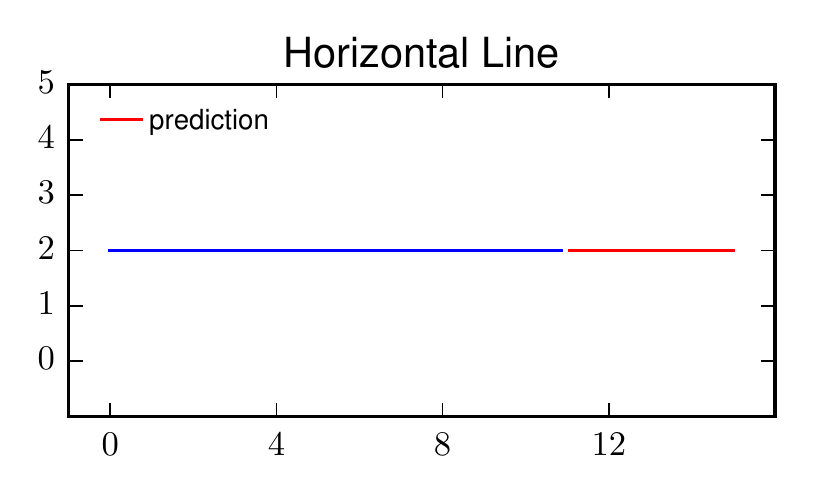}}
                & Temporal Structure & $p^{\scriptscriptstyle{\rm synth}}$ \\ \cline{2-3} \addlinespace[0.5em]
        & $\times$ White Noise  & 3\%     \\
        & $\times$ Linear Trend & 8\%      \\
        & $\times$ Periodicity  & 1\%      \\
        & $\times$ Change Point & 2\%
\end{tabular}
\captionsetup{skip=0pt}
\caption{}
\label{subfig:gp-structure-const}
\end{subtable}

\captionsetup{skip=5pt}
\caption{Detecting probable temporal structures in multiple time series with
varying characteristics. In each of the panels %
\subref{subfig:gp-structure-temperature}--\subref{subfig:gp-structure-const},
the plot shows observed time series data in blue and the table identifies which
temporal structures truly exist in the time series as well as the posterior
probability $p^{\rm synth}$ that each structure is present in a Gaussian process
program from Bayesian synthesis given the data. As described in
Eq~\eqref{eq:prob-has-property} of
Section~\ref{subsec:bayesian-syntheis-qualitative}, $p^{\rm synth}$ is estimated
by returning the fraction of programs in the ensemble that contain each
structure. Programs from Bayesian synthesis accurately reflect the probable
presence or absence of linear, periodic, and change point characteristics. The
red lines show predictions from a randomly selected synthesized program, showing
that they additionally capture compositions of temporal structures to faithfully
model the data.}
\label{fig:gp-structures}
\end{figure}

\section{Experimental Results}

We have developed a set of benchmark problems in two domains: (i) time series
data and (ii) multivariate tabular data. These benchmarks reflect a
broad range of real-world data generating processes with varying qualitative
structure.
We evaluated these probabilistic programs in two ways. First, we qualitatively
assessed the inferences about structure that were made by processing the text of
the programs. Second, we quantitatively benchmarked the predictive accuracy of
the synthesized programs. Probabilistic programs from Bayesian synthesis
often provide improved accuracy over standard methods for data modeling.


\begin{figure}[!tbp]

\begin{subfigure}{\linewidth}
\begin{subfigure}{\linewidth}
\includegraphics[width=\linewidth]{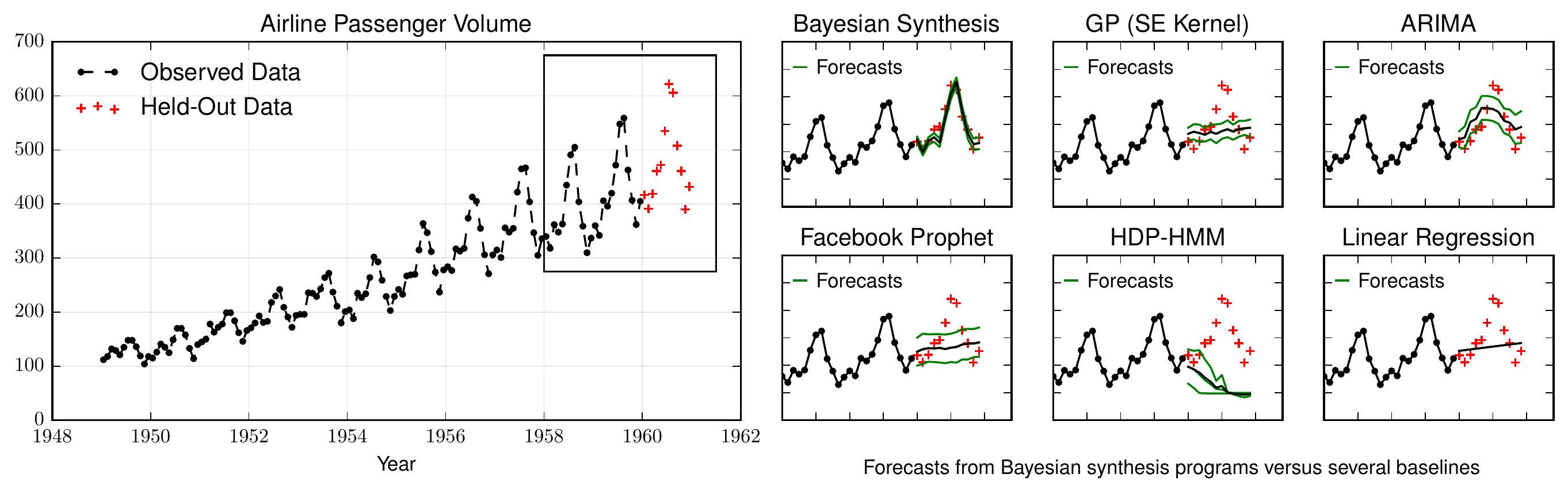}
\end{subfigure}

\begin{subtable}{\linewidth}
\footnotesize
\subcaption*{Standardized Root Mean Squared Forecasting Error (RMSE)
  on Real-World Benchmark Problems}
\begin{tabular*}{\textwidth}{@{\extracolsep{\fill}}lrrrrrrr}
\toprule
  & \texttt{temperature}
  & \texttt{airline}
  & \texttt{call}
  & \texttt{mauna}
  & \texttt{radio}
  & \texttt{solar}
  & \texttt{wheat}
\\ \cmidrule{2-8}
Bayesian Synthesis
  & \textbf{1.0}
  & \textbf{1.0}
  & \textbf{1.0}
  & \textbf{1.0}
  & \textbf{1.0}
  & 1.47
  & 1.50
\\
Gaussian Process (Squared Exponential Kernel)
  & 1.70
  & 2.01
  & 4.26
  & 1.54
  & 2.03
  & 1.63
  & 1.37
\\
Auto-Regressive Integrated Moving Average
  & 1.85
  & 1.32
  & 2.44
  & 1.09
  & 2.08
  & \textbf{1.0}
  & 1.41
\\
Facebook Prophet
  & 2.00
  & 1.83
  & 5.61
  & 1.23
  & 3.09
  & 1.73
  & 1.29
\\
Hierarchical-DP Hidden Markov Model
  & 1.77
  & 4.61
  & 2.26
  & 14.77
  & 1.19
  & 3.49
  & 1.89
\\
Linear Regression
  & 1.30
  & 1.79
  & 6.23
  & 2.19
  & 2.73
  & 1.57
  & \textbf{1.0}
\\ \bottomrule
\end{tabular*}
\end{subtable}
\end{subfigure}

\captionsetup{skip=5pt}
\caption{Quantitative evaluation of forecasting using Gaussian process
programs from Bayesian synthesis, as compared to five common baselines. The top
panels show extrapolated time series by each method on the \texttt{airline}
data. The table shows prediction errors achieved by each method on seven
real-world time series.}
\label{fig:gp-predictions}

\end{figure}

\subsection{Inferring Qualitative Structure from Time Series Data}

This set of benchmark problems consists of five real-world time series datasets
from \citet{gpss}: (i) bi-weekly global temperature measurements from 1980 to
1990, (ii) monthly global gas production from 1956 to 1995, (iii) monthly
airline passenger volume from 1948, (iv) monthly radio sales from 1935 to 1954,
and (iv) monthly call center activity from 1962 to 1976. Each of the datasets
has a different qualitative structure. For example, the global temperature
dataset is comprised of yearly periodic structure overlaid with white noise,
while the call center dataset has a linear trend until 1973 and a sharp drop
afterwards.
Figure~\ref{fig:gp-structures} shows the five datasets, along with the
inferences our technique made about the qualitative structure of the data
generating process underlying each dataset. The results show that our technique
accurately infers the presence or absence of each type of structure in each
benchmark. Specifically, if a specific qualitative structure is deemed to be
present if the posterior probability inferred by our technique is above 50\%,
and inferred to be absent otherwise, then the inferences from Bayesian synthesis
match the real world structure in every case.
To confirm that the
synthesized programs report the absence of any meaningful structure when it is
not present in the data, we also include a horizontal
line shown in panel~\subref{subfig:gp-structure-const}.

\subsection{Quantitative Prediction Accuracy for Time Series Data}

Figure~\ref{fig:gp-predictions} shows quantitative prediction accuracy for seven
real world econometric time series. The top panel visually compares predictions
from our method with predictions obtained by baseline methods on the airline
passenger volume benchmark. We chose to compare against baselines that (i) have
open source, re-usable implementations; (ii) are widely used and cited in the
statistical literature; and (iii) are based on flexible model families that,
like our Gaussian process technique, have default hyperparameter settings and do
not require significant manual tuning to fit the data. Baseline methods include
Facebook's Prophet algorithm~\citep{taylor2017}; standard econometric
techniques, such as Auto-Regressive Integrated Moving Average (ARIMA)
modeling~\citep{hyndman2008}; and advanced time series modeling techniques from
nonparametric Bayesian statistics, such as the Hierarchical Dirichlet Process
Hidden Markov Model \citep{johnson2013}. Bayesian synthesis is the only
technique that accurately captures both quantitative and qualitative structure.
The bottom panel of Figure~\ref{fig:gp-predictions} shows quantitative results
comparing predictive accuracy on these datasets. Bayesian synthesis
produces more accurate models for five of the seven benchmark problems and is
competitive with other techniques on the other two problems.

\subsection{Runtime Versus Prediction Accuracy for Time Series Data}


\begin{figure}[!tb]
\begin{subfigure}[b]{.5\linewidth}
\includegraphics[width=\linewidth]{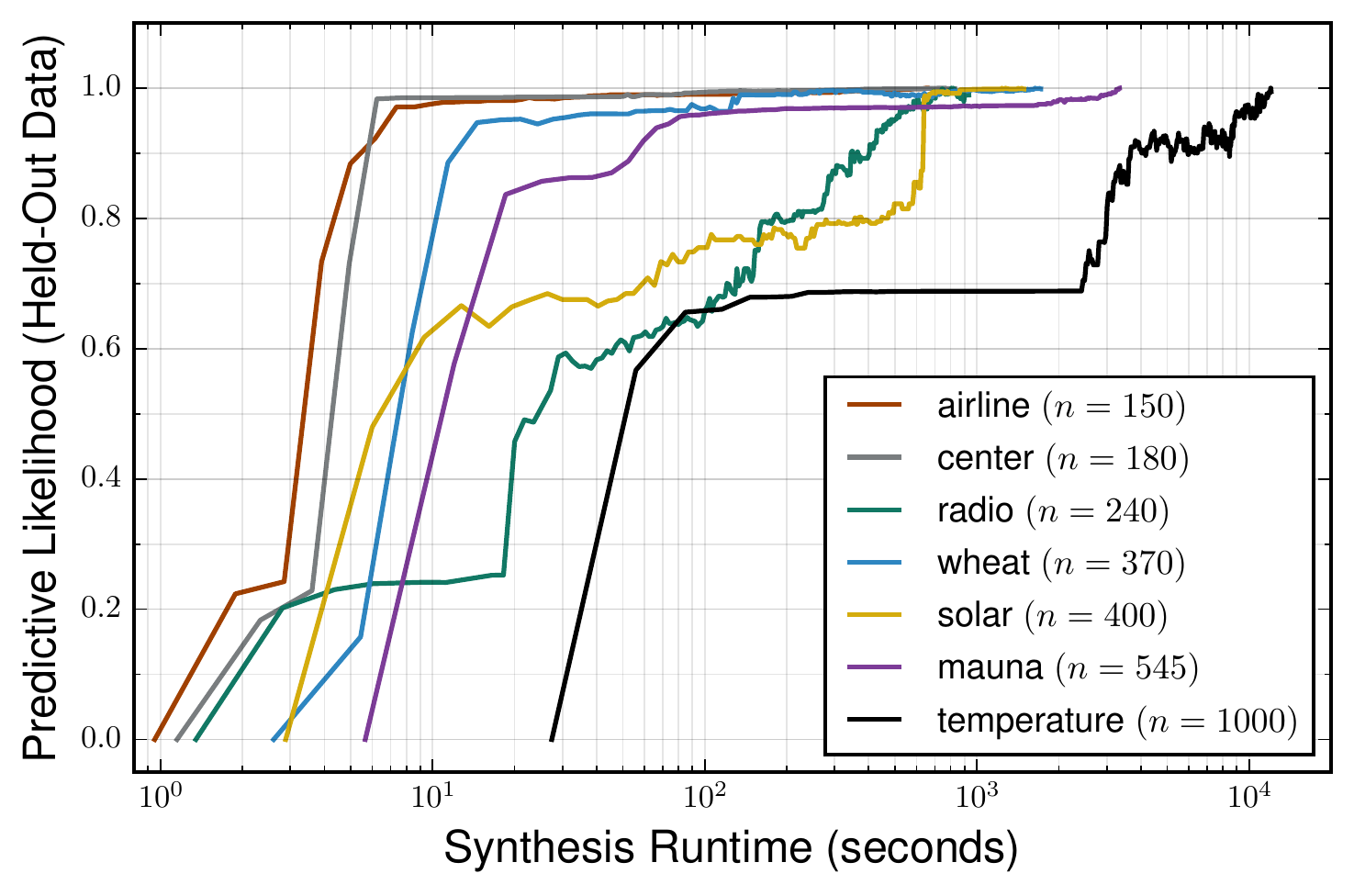}
\subcaption{\footnotesize Runtime vs. Accuracy Profiles}
\label{subfig:gp-synthesis-runtime-plot}
\end{subfigure}%
\begin{subfigure}[b]{.5\linewidth}
\includegraphics[width=\linewidth]{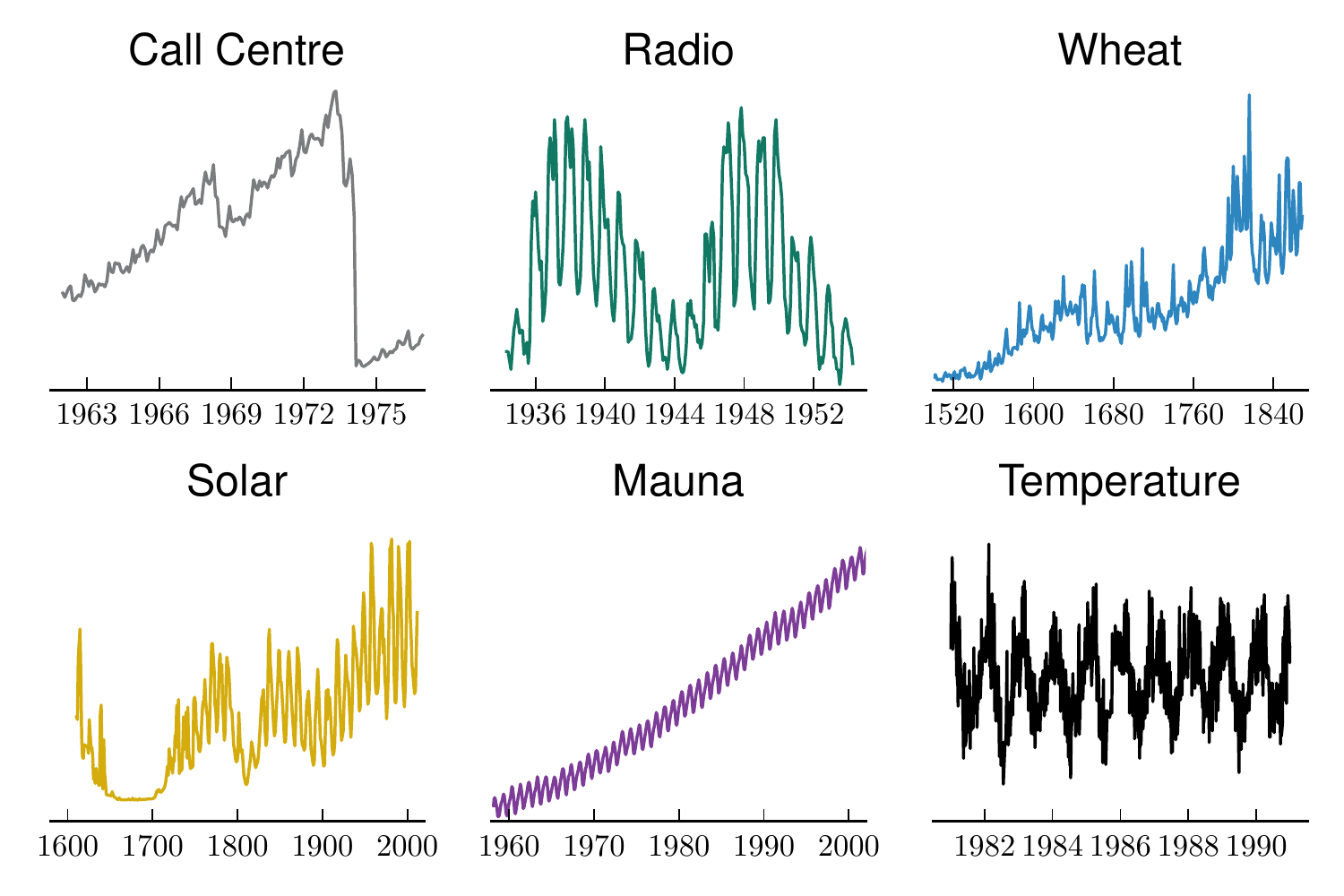}
\subcaption{\footnotesize Time Series Datasets}
\label{subfig:gp-synthesis-runtime-datasets}
\end{subfigure}%
\captionsetup{skip=5pt}
\caption{Synthesis runtime versus held-out predictive likelihood on the
time series benchmarked in Figure~\ref{fig:gp-predictions}. Panel
\subref{subfig:gp-synthesis-runtime-plot} shows a log-linear plot of the median
runtime and predictive likelihood values taken over 100 independent runs of
synthesis. Higher values of predictive likelihood on the y-axis indicate better
fit to the observed data. Panel \subref{subfig:gp-synthesis-runtime-datasets}
shows the time series datasets used to compute runtime and accuracy measurements.}
\label{fig:gp-synthesis-runtime}
\end{figure}


\begin{figure}[!p]

\scriptsize
\begin{tabular}{p{0.05cm}l|l|lc|rlrl}
\toprule
  & \multicolumn{1}{c|}{\multirow{2}{*}{Variable 1}}
  & \multicolumn{1}{c|}{\multirow{2}{*}{Variable 2}}
  & \multicolumn{2}{c|}{\multirow{2}{*}{True Predictive Structure}}
  & \multicolumn{4}{c}{Predictive Relationship Detected By}
\\
  &
  &
  &
  &
  & \multicolumn{2}{c}{Pearson Correlation}
  & \multicolumn{2}{c}{Bayesian Synthesis}
\\ \midrule
\subref{subfig:grid-0-0} & flavanoids        & color-intensity     & linear + bimodal                   & $\checkmark$ & $\times$ & $(0.03)$    & $\checkmark$ & $(0.97)$ \\
\subref{subfig:grid-0-1} & A02               & A07                 & linear + heteroskedastic           & $\checkmark$ & $\times$ & $(0.16)$   & $\checkmark$ & $(0.89)$ \\
\subref{subfig:grid-0-2} & A02               & A03                 & linear + bimodal + heteroskedastic & $\checkmark$ & $\times$ & $(0.03)$   & $\times    $ & $(0.66)$ \\
\subref{subfig:grid-0-3} & proline           & od280-of-wines      & nonlinear + missing regime         & $\checkmark$ & $\times$ & $(0.09)$   & $\checkmark$ & $(0.97)$ \\
\subref{subfig:grid-1-0} & compression-ratio & aspiration          & mean shift                         & $\checkmark$ & $\times$ & $(0.07)$   & $\checkmark$ & $(0.98)$ \\
\subref{subfig:grid-1-1} & age               & income              & different group tails              & $\checkmark$ & $\times$ & $(0.06)$   & $\checkmark$ & $(0.90)$ \\
\subref{subfig:grid-1-2} & age               & varices             & scale shift                        & $\checkmark$ & $\times$ & $(0.00)$  & $\checkmark$ & $(0.90)$ \\
\subref{subfig:grid-1-3} & capital-gain      & income              & different group tails              & $\checkmark$ & $\times$ & $(0.05)$   & $\times    $ & $(0.77)$ \\
\subref{subfig:grid-2-0} & city-mpg          & highway-mpg         & linearly increasing                & $\checkmark$ & $\checkmark$ & $(0.95)$ & $\checkmark$ & $(1.00)$ \\
\subref{subfig:grid-2-1} & horsepower        & highway-mpg         & linearly decreasing                & $\checkmark$ & $\checkmark$ & $(0.65)$ & $\checkmark$ & $(1.00)$ \\
\subref{subfig:grid-2-2} & education-years   & education-level     & different group means              & $\checkmark$ & $\checkmark$ & $(1.00)$ & $\checkmark$ & $(1.00)$ \\
\subref{subfig:grid-2-3} & compression-ratio & fuel-type           & different group means              & $\checkmark$ & $\checkmark$ & $(0.97)$ & $\checkmark$ & $(0.98)$ \\
\subref{subfig:grid-3-0} & cholesterol       & max-heart-rate      & none (+ outliers)                  & $\times$     & $\times$ & $(0.00)$     & $\times$ & $(0.08)$ \\
\subref{subfig:grid-3-1} & cholesterol       & st-depression       & none (+ outliers)                  & $\times$     & $\times$ & $(0.00)$    & $\times$ & $(0.00)$ \\
\subref{subfig:grid-3-2} & blood-pressure    & sex                 & none                               & $\times$     & $\times$ & $(0.01)$    & $\times$ & $(0.26)$ \\
\subref{subfig:grid-3-3} & st-depression     & electrocardiography & none                               & $\times$     & $\times$ & $(0.04)$    & $\times$ & $(0.00)$ \\  \bottomrule
\end{tabular}

\bigskip

\begin{subfigure}{.25\linewidth}
\includegraphics[width=\linewidth]{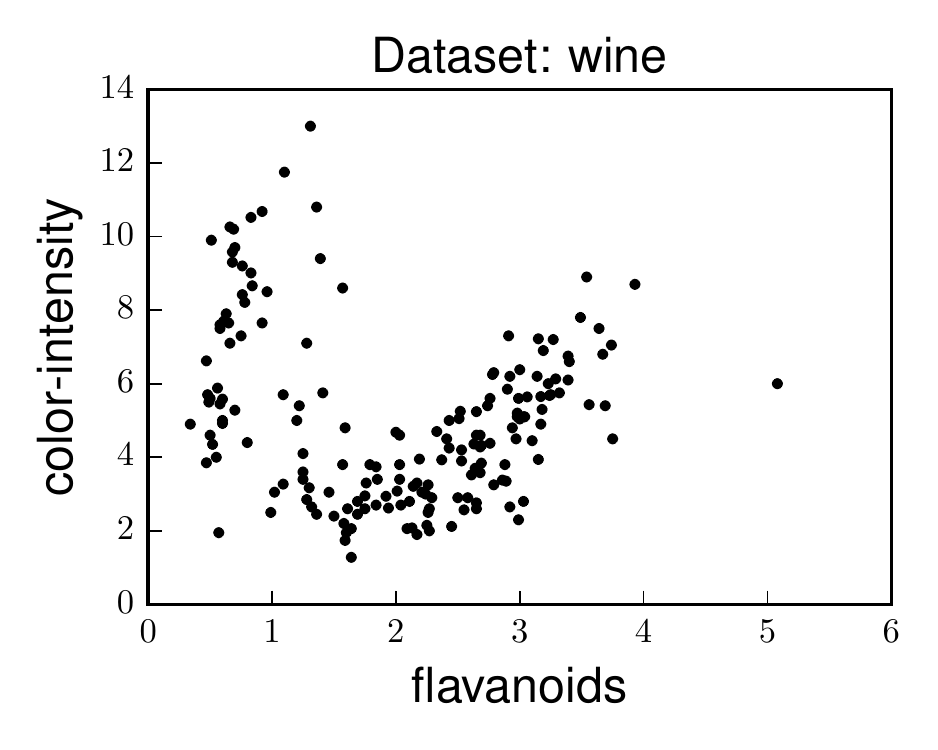}
\captionsetup{skip=-2pt}
\caption{\scriptsize Linear + Bimodal}
\label{subfig:grid-0-0}
\end{subfigure}%
\begin{subfigure}{.25\linewidth}
\includegraphics[width=\linewidth]{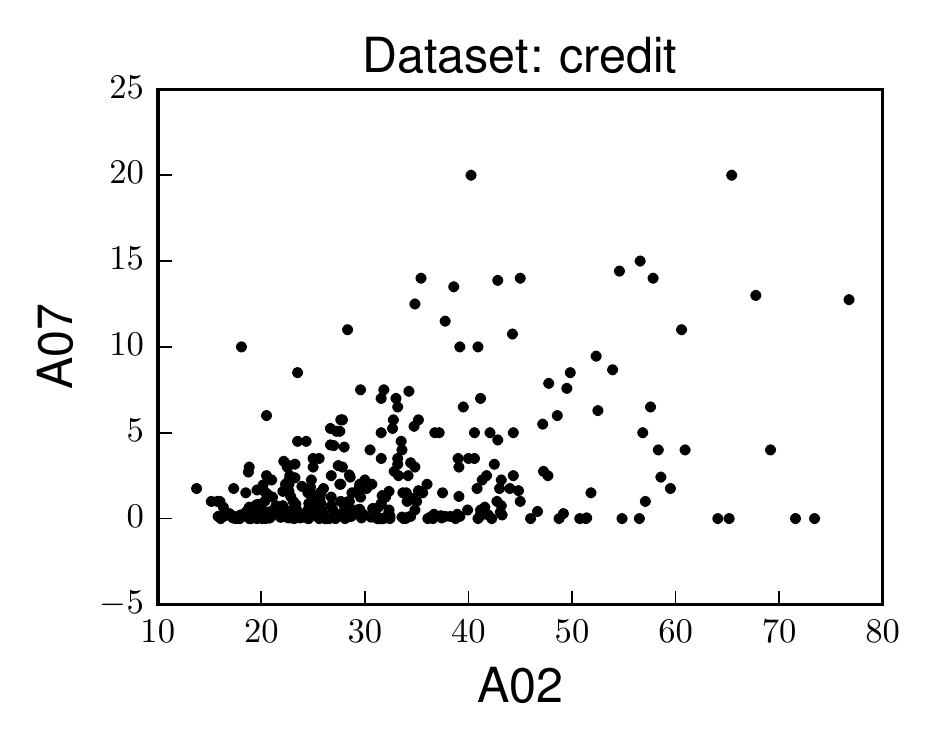}
\captionsetup{skip=-2pt}
\caption{\scriptsize Linear + Heteroskedastic}
\label{subfig:grid-0-1}
\end{subfigure}%
\begin{subfigure}{.25\linewidth}
\includegraphics[width=\linewidth]{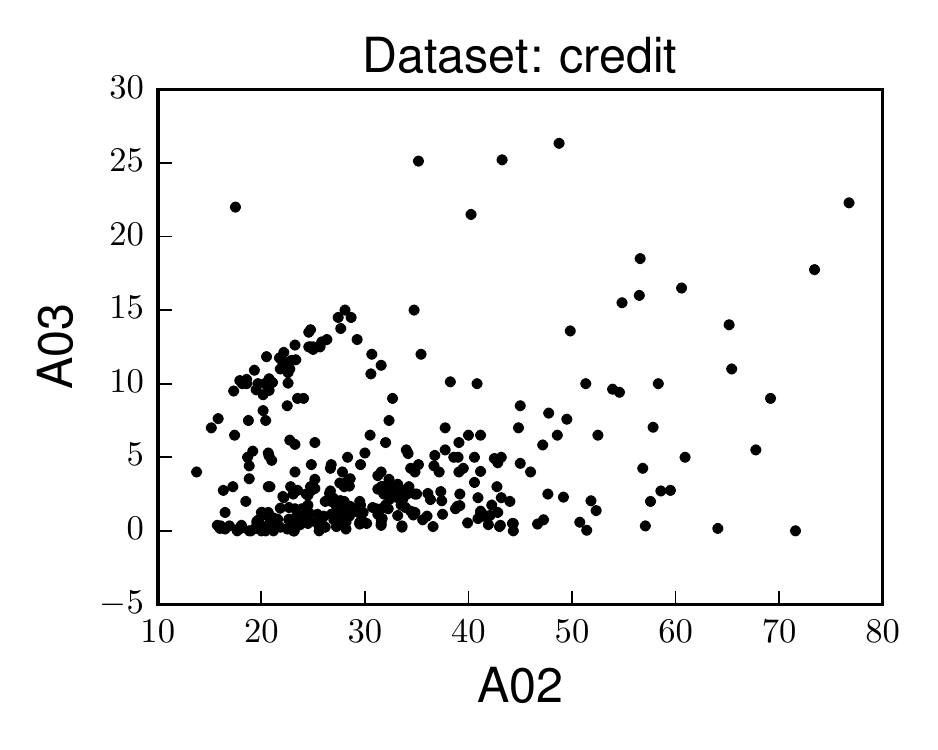}
\captionsetup{skip=-2pt}
\caption{\scriptsize Bimodal + Heteroskedastic}
\label{subfig:grid-0-2}
\end{subfigure}%
\begin{subfigure}{.25\linewidth}
\includegraphics[width=\linewidth]{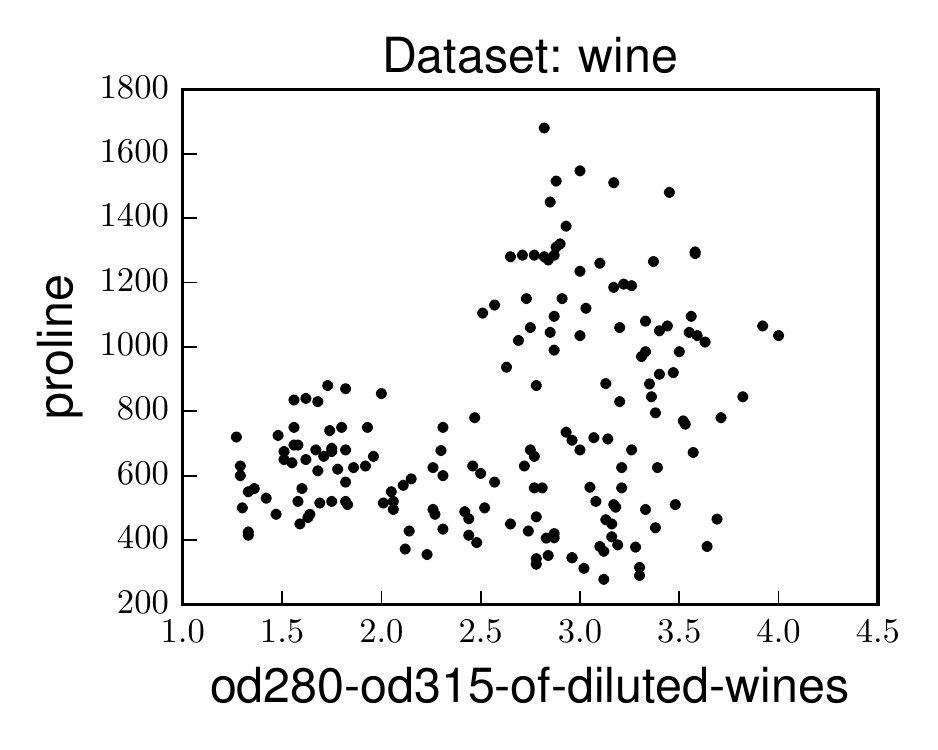}
\captionsetup{skip=-2pt}
\caption{\scriptsize Missing Regime}
\label{subfig:grid-0-3}
\end{subfigure}%

\begin{subfigure}{.25\linewidth}
\includegraphics[width=\linewidth]{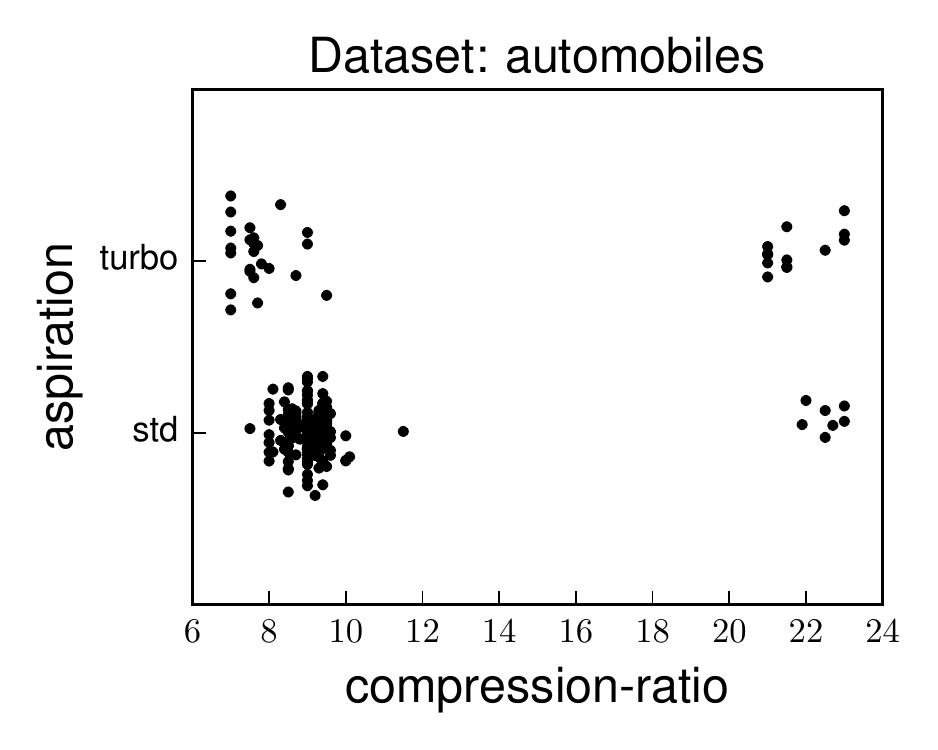}
\captionsetup{skip=-2pt}
\caption{\scriptsize Mean Shift}
\label{subfig:grid-1-0}
\end{subfigure}%
\begin{subfigure}{.25\linewidth}
\includegraphics[width=\linewidth]{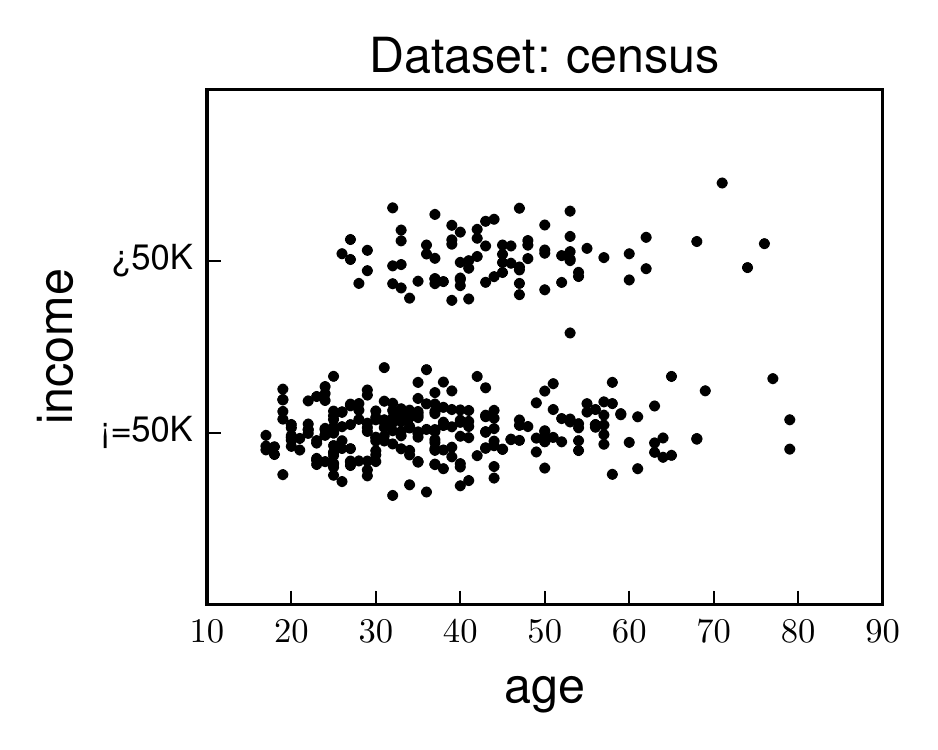}
\captionsetup{skip=-2pt}
\caption{\scriptsize Scale Shift}
\label{subfig:grid-1-1}
\end{subfigure}%
\begin{subfigure}{.25\linewidth}
\includegraphics[width=\linewidth]{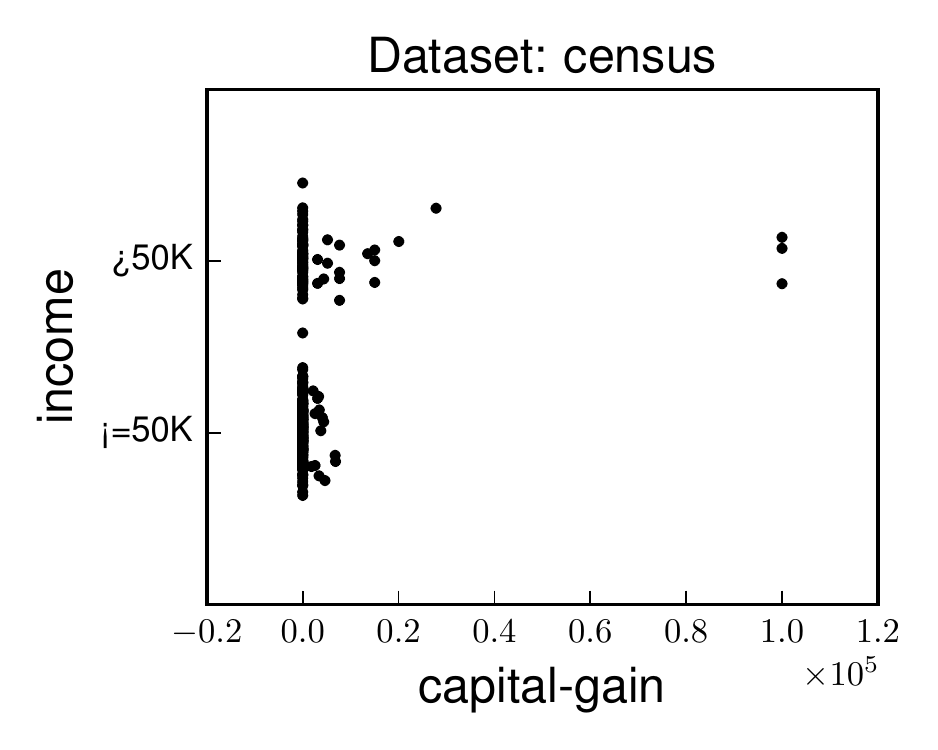}
\captionsetup{skip=-2pt}
\caption{\scriptsize Different Tails}
\label{subfig:grid-1-2}
\end{subfigure}%
\begin{subfigure}{.25\linewidth}
\includegraphics[width=\linewidth]{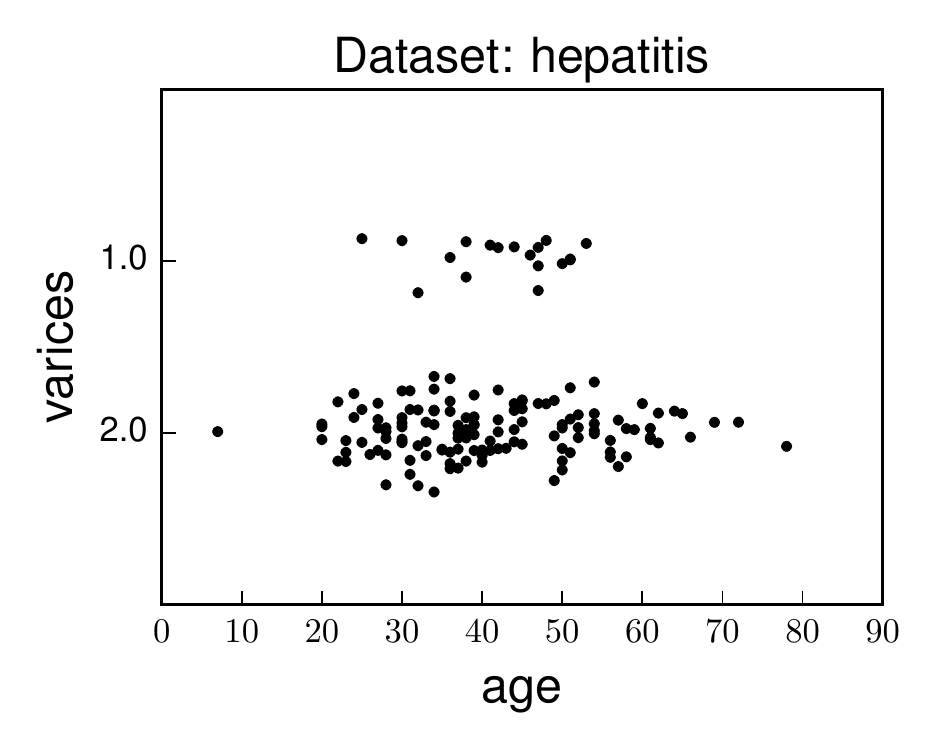}
\captionsetup{skip=-2pt}
\caption{\scriptsize Different Tails}
\label{subfig:grid-1-3}
\end{subfigure}%

\begin{subfigure}{.25\linewidth}
\includegraphics[width=\linewidth]{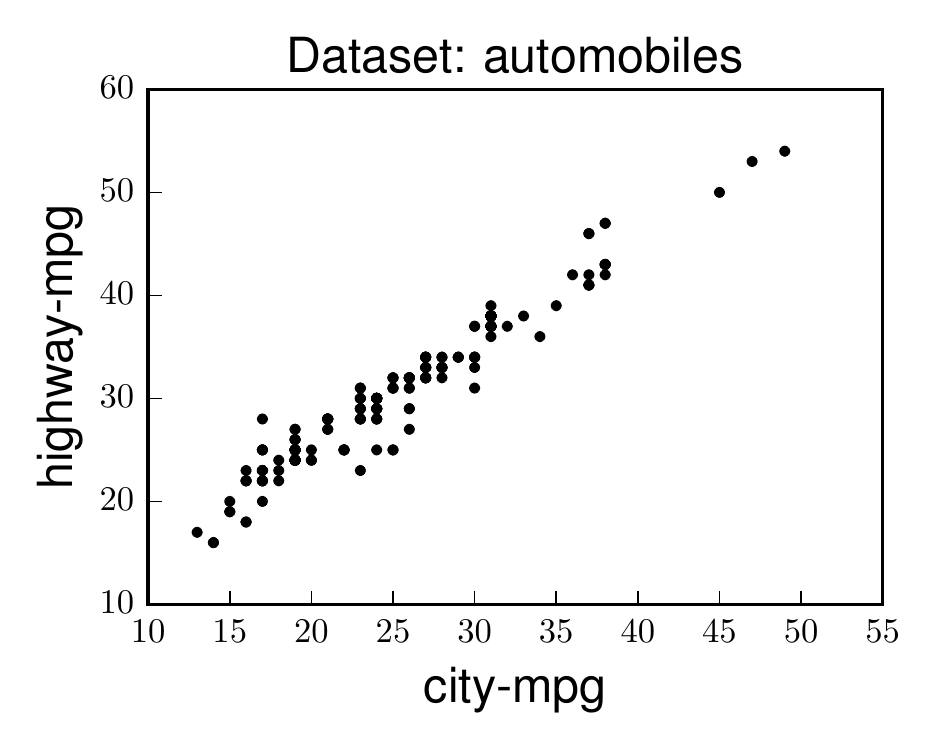}
\captionsetup{skip=-2pt}
\caption{\scriptsize Linearly Increasing}
\label{subfig:grid-2-0}
\end{subfigure}%
\begin{subfigure}{.25\linewidth}
\includegraphics[width=\linewidth]{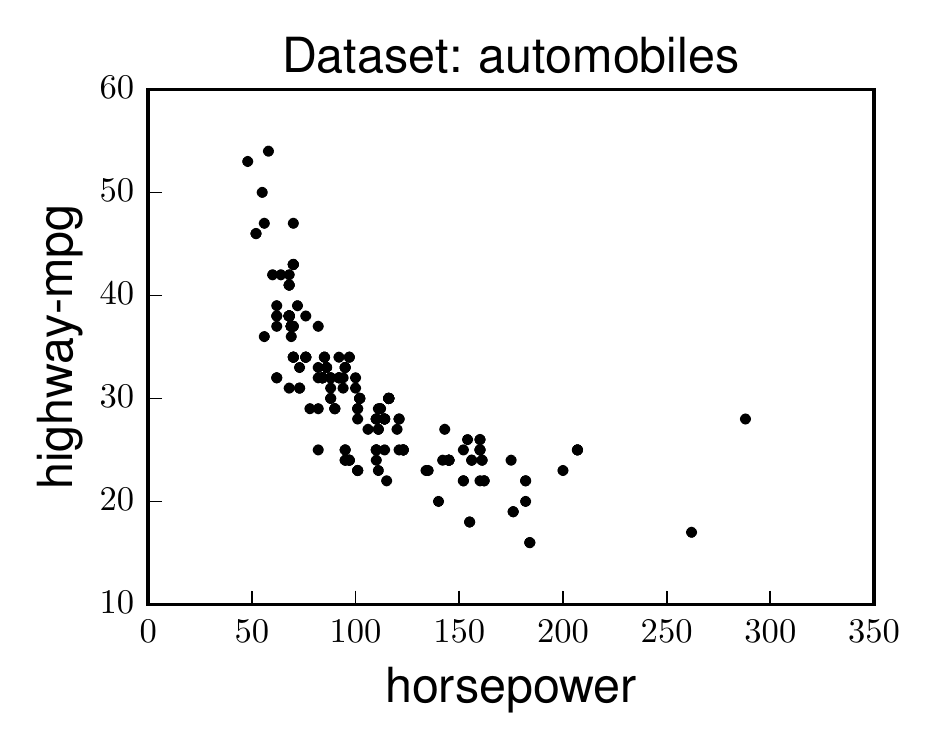}
\captionsetup{skip=-2pt}
\caption{\scriptsize Linearly Decreasing}
\label{subfig:grid-2-1}
\end{subfigure}%
\begin{subfigure}{.25\linewidth}
\includegraphics[width=\linewidth]{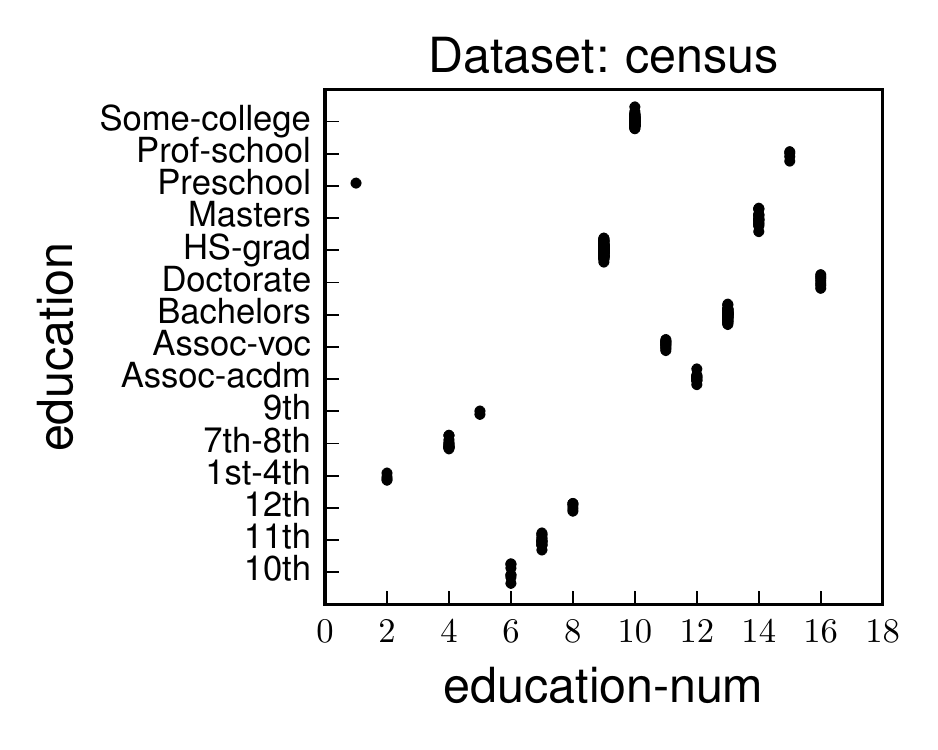}
\captionsetup{skip=-2pt}
\caption{\scriptsize Different Group Means}
\label{subfig:grid-2-2}
\end{subfigure}%
\begin{subfigure}{.25\linewidth}
\includegraphics[width=\linewidth]{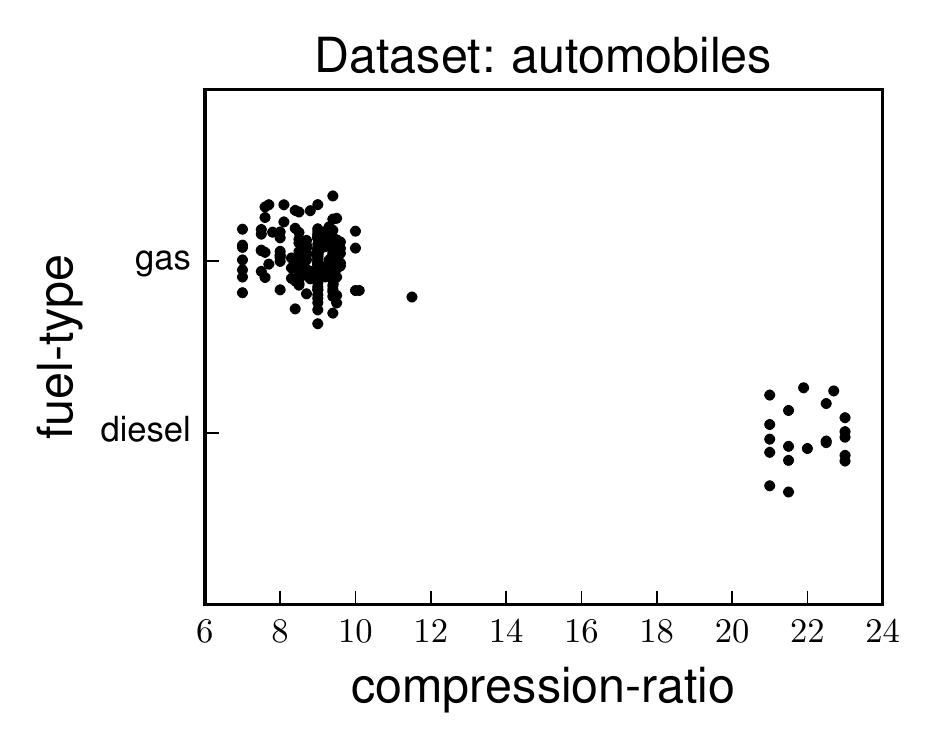}
\captionsetup{skip=-2pt}
\caption{\scriptsize Different Group Means}
\label{subfig:grid-2-3}
\end{subfigure}%

\begin{subfigure}{.25\linewidth}
\includegraphics[width=\linewidth]{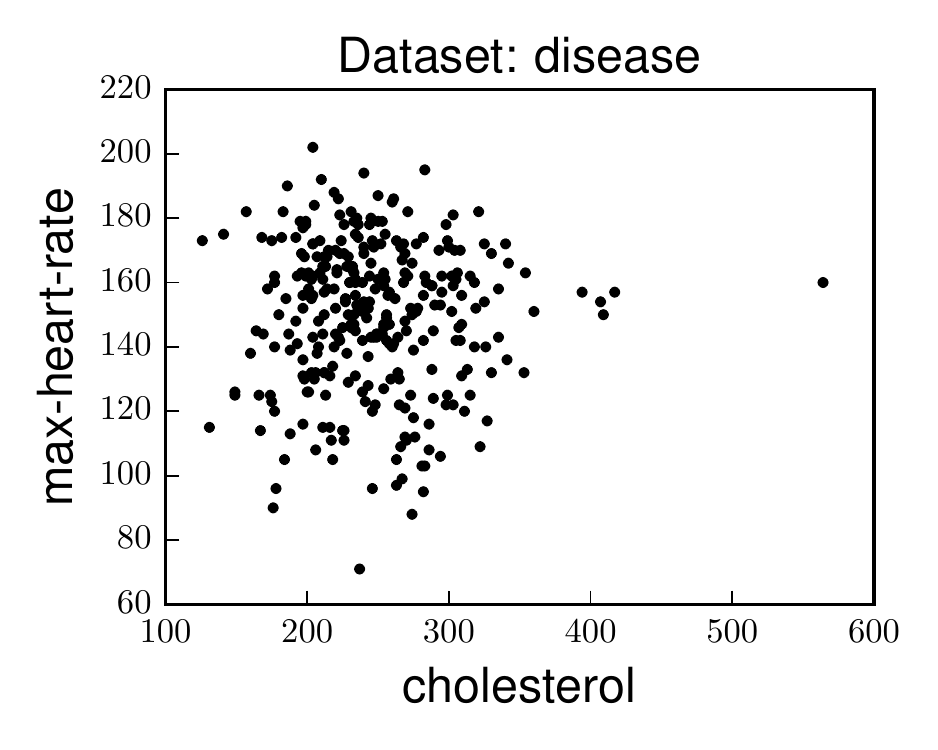}
\captionsetup{skip=-2pt}
\caption{\scriptsize No Dependence + Outliers}
\label{subfig:grid-3-0}
\end{subfigure}%
\begin{subfigure}{.25\linewidth}
\includegraphics[width=\linewidth]{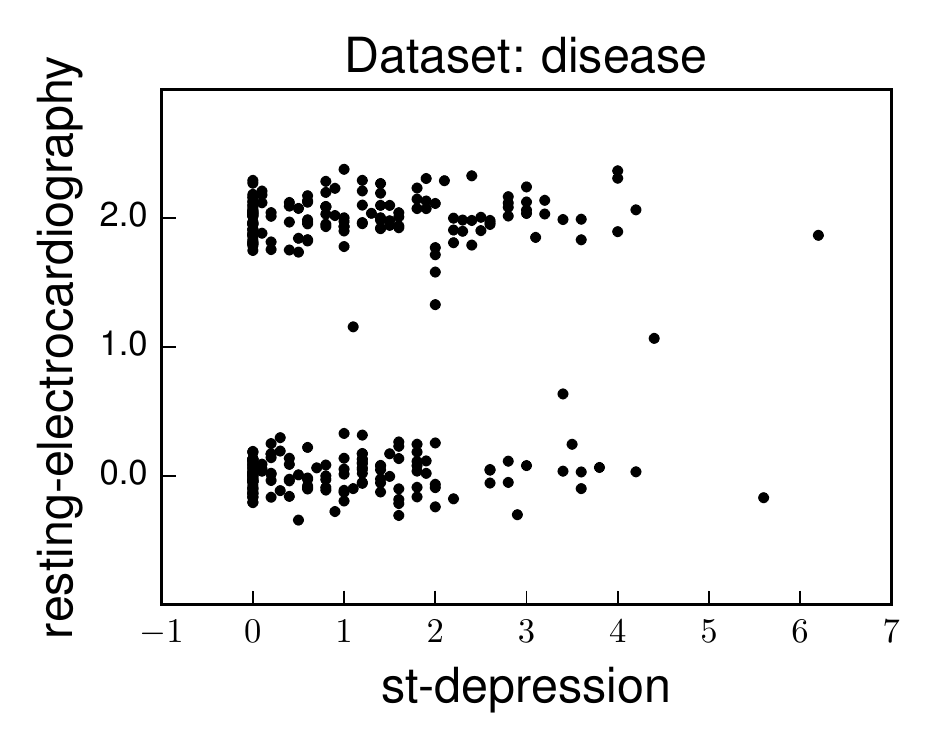}
\captionsetup{skip=-2pt}
\caption{\scriptsize No Dependence + Outliers}
\label{subfig:grid-3-1}
\end{subfigure}%
\begin{subfigure}{.25\linewidth}
\includegraphics[width=\linewidth]{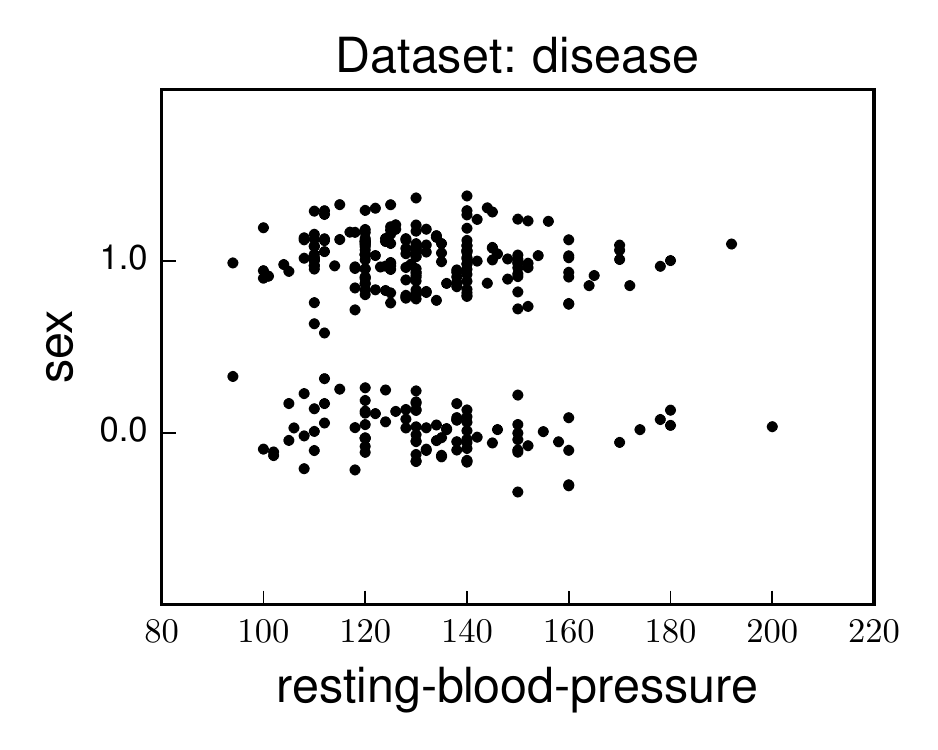}
\captionsetup{skip=-2pt}
\caption{\scriptsize No Dependence}
\label{subfig:grid-3-2}
\end{subfigure}%
\begin{subfigure}{.25\linewidth}
\includegraphics[width=\linewidth]{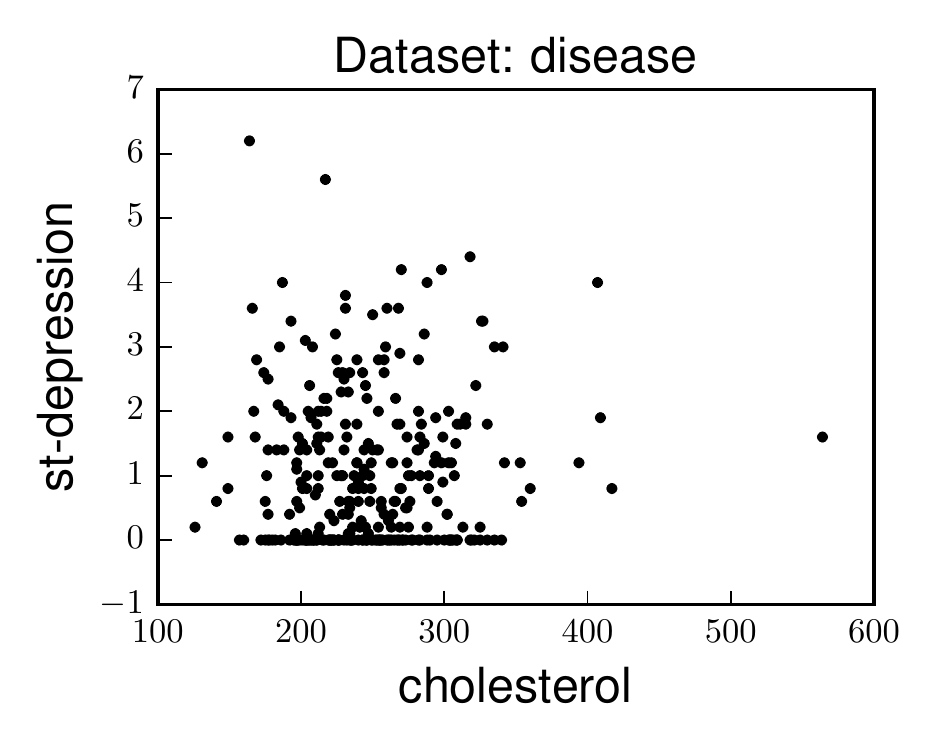}
\captionsetup{skip=-2pt}
\caption{\scriptsize No Dependence}
\label{subfig:grid-3-3}
\end{subfigure}%

\caption{Detecting probable predictive relationships between pairs of variables
for widely varying dependence structures including nonlinearity,
heteroskedasticity, mean and scale shifts, and missing regimes. The table in the
top panel shows 16 pairs of variables, the true predictive structure between
them, and indicators as to whether a predictive relationship, if any, was
detected by Pearson correlation and Bayesian synthesis.}
\label{fig:crosscat-structures}
\end{figure}

Figure~\ref{fig:gp-synthesis-runtime} shows a profile of the predictive likelihood on
held-out data versus the synthesis runtime (in seconds) for each of the seven
econometric time series.
For each dataset, we ran 2000 steps of synthesis, taking runtime and predictive
likelihood measurements at every 5 steps.
The plotted runtime and predictive likelihood values represent the median values
taken over 100 independent runs of Bayesian synthesis for each dataset.
The predictive likelihood measurements on the y-axis are scaled between 0 and 1
so that these profiles can be compared across the seven datasets.
We see that there is significant variance in the synthesis runtime versus
prediction quality. For certain time series such as \texttt{airline} and
\texttt{call} datasets (150 and 180 observations), prediction quality stabilizes
after around 10 seconds; for \texttt{radio} and \texttt{solar} datasets (240 and
400 observations) the prediction quality takes around 1000 seconds to stabilize;
and for the \texttt{temperature} dataset (1000 observations) the predictions
continue to improve even after exceeding the maximum timeout.

Recall from Section~\ref{subsec:dsl-gp-complexity} that each step of
synthesis has a time complexity of $O(l_G + |E|n^2 + n^3)$.
Scaling the synthesis to handle time series with more than few thousands data
points will require the sparse Gaussian process techniques discussed in
Section~\ref{subsec:dsl-gp-complexity}.
It is also important to emphasize that the time cost of each iteration of
synthesis is not only a function of the number of observations $n$ but also the
size $\abs{E}$ of the synthesized programs. Thus, while more observations
typically require more time, time series with more complex temporal patterns
typically result in longer synthesized programs. This trade-off can be seen by
comparing the profiles of the mauna data (545 observations) with the radio data
(240 observations). While mauna has more observations $n$, the radio data
contains much more complex periodic patterns and which requires comparatively
longer programs and thus longer synthesis times.

\subsection{Inferring Qualitative Structure from Multivariate Tabular Data}

A central goal of data analysis for multivariate tabular data is to identify
predictive relationships between pairs of variables
\citep{ezekiel1941,draper1966}.
Recall from Section~\ref{sec:dsl-crosscat} that synthesized programs from the
multivariate tabular data DSL place related variables in the same variable block
expression. Using Eq~\eqref{eq:prob-has-property} from
Section~\ref{subsec:bayesian-syntheis-qualitative}, we report a relationship
between a pair of variables if 80\% of the synthesized programs place these
variables in the same block.
We compare our method to the widely-used Pearson correlation baseline
\citep{abott2016} by checking whether the correlation
value exceeds 0.20 (at the $5\%$ significance level).

We evaluate the ability of our method to
correctly detect related variables in six real-world datasets from the UCI
machine learning repository~\citep{dua2017}: automobile data (26 variables),
wine data (14 variables), hepatitis data (20 variables), heart disease data (15
variables), and census data (15 variables).
Figure~\ref{fig:crosscat-structures} shows results for sixteen selected pairs of
variables in these datasets (scatter plots shown in
\ref{subfig:grid-0-0}--\ref{subfig:grid-3-3}). Together, these pairs exhibit a
broad class of relationships, including linear, nonlinear, heteroskedastic, and
multi-modal patterns. Four of the benchmarks, shown in the final row, have no
predictive relationship.


\begin{figure}[!t]

\begin{subfigure}{.5\linewidth}
\includegraphics[width=\linewidth]{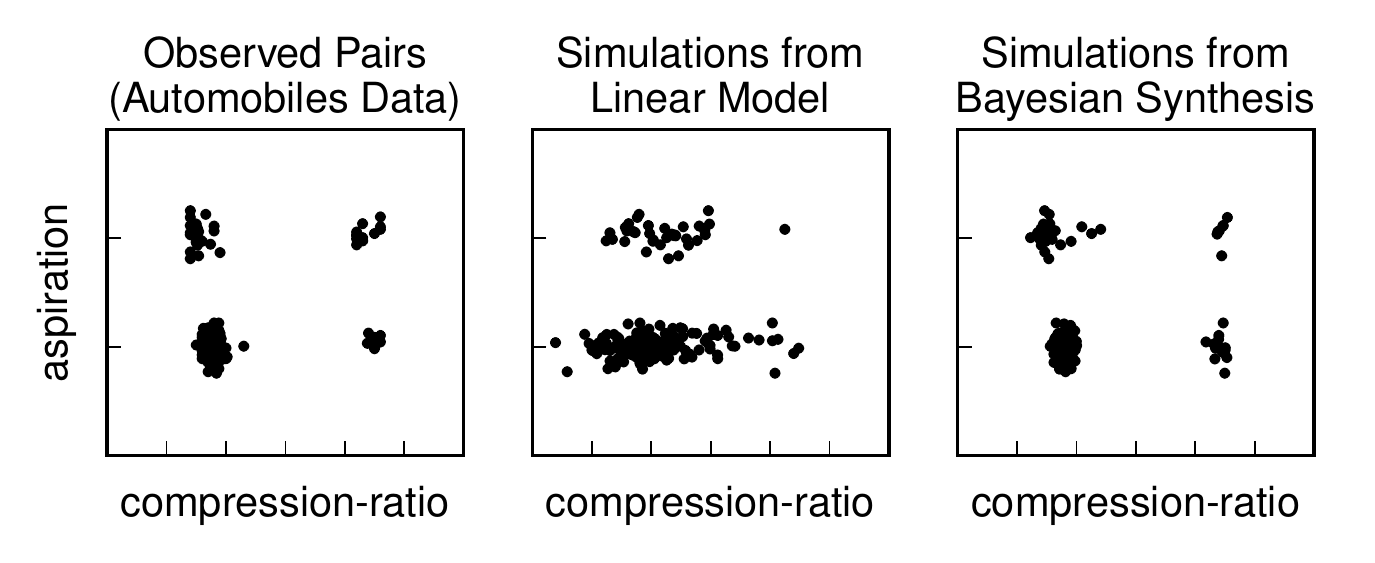}
\captionsetup{skip=0pt}
\caption{\footnotesize Bimodal Dependence}
\label{subfig:crosscat-linear-bimodal-nominal}
\end{subfigure}%
\begin{subfigure}{.5\linewidth}
\includegraphics[width=\linewidth]{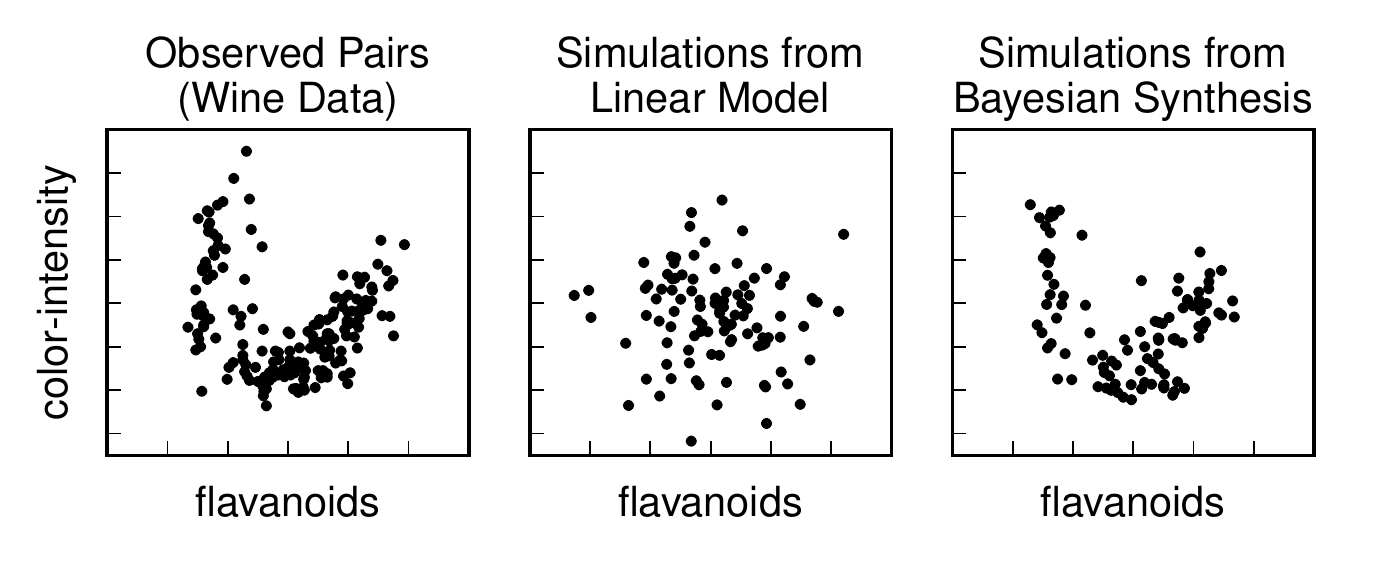}
\captionsetup{skip=0pt}
\caption{\footnotesize Bimodal Increasing Linear Dependence}
\label{subfig:crosscat-linear-bimodal-linear}
\end{subfigure}
\hrule

\begin{subfigure}{.5\linewidth}
\includegraphics[width=\linewidth]{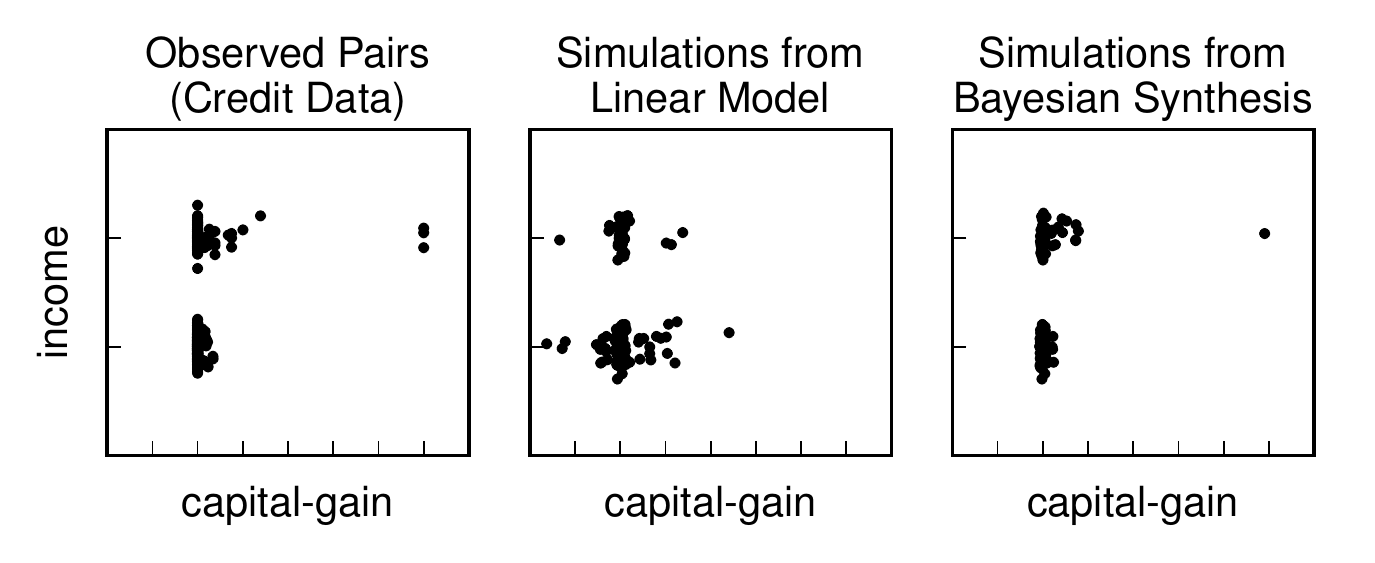}
\captionsetup{skip=0pt}
\caption{\footnotesize Heavy-Tailed Nonlinear Dependence}
\label{subfig:crosscat-linear-heavy-tailed-nominal}
\end{subfigure}%
\begin{subfigure}{.5\linewidth}
\includegraphics[width=\linewidth]{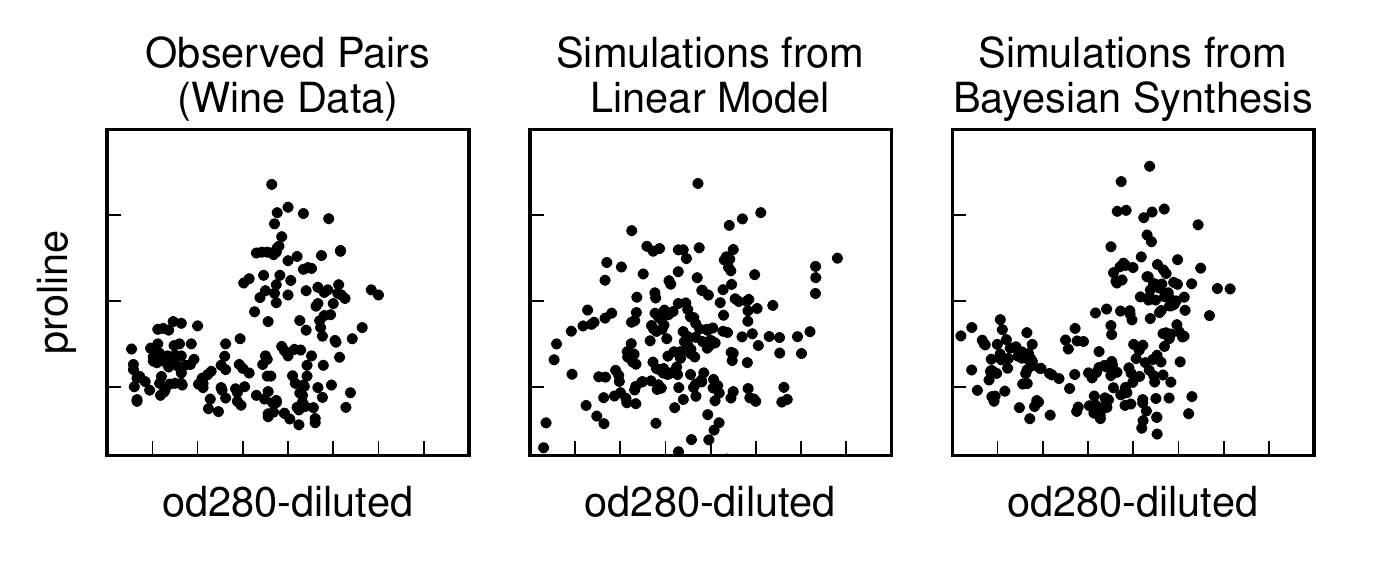}
\captionsetup{skip=0pt}
\caption{\footnotesize Missing Regime Nonlinear Dependence}
\label{subfig:crosscat-linear-missing-regime}
\end{subfigure}%

\captionsetup{skip=5pt}
\caption{Comparing the quality of predictive models discovered by Bayesian
synthesis to commonly-used generalized linear models, for several
real-world datasets and dependence patterns. In each of the panels
\subref{subfig:crosscat-linear-bimodal-nominal}--%
\subref{subfig:crosscat-linear-missing-regime}, the first column shows a scatter
plot of two variables in the observed dataset; the second column shows data
simulated from a linear model trained on the observed data; and the third column
shows data simulated from probabilistic programs obtained using Bayesian
synthesis given the observed data. The synthesized programs are able to detect
underlying patterns and emulate the true distribution, including nonlinear,
multi-modal, and heteroskedastic relationships. Simulations from linear models
have a poor fit.}
\label{fig:crosscat-linear}
\vspace{-.5cm}
\end{figure}

The table in Figure~\ref{fig:crosscat-structures} shows that our method
correctly detects the presence or absence of a predictive relationship in 14 out
of 16 benchmarks, where the final column shows the posterior probability of a
relationship in each case.
In contrast, Pearson correlation, shown in the second-to-last column, only
obtains the correct answer in 8 of 16 benchmarks. This baseline yields incorrect
answers for all pairs of variables with nonlinear, bimodal, and/or
heteroskedastic relationships.
These results show that Bayesian synthesis provides a practical method for
detecting complex predictive relationships from real-world data that are missed
by standard baselines.

Moreover, Figure~\ref{fig:crosscat-linear} shows that probabilistic programs
from Bayesian synthesis are able to generate entirely new datasets that
reflect the pattern of the predictive relationship in the underlying data
generating process more accurately than generalized linear statistical models.



\begin{figure}[!b]
\vspace{-.6cm}
\begin{subtable}{.6\linewidth}
\footnotesize
\subcaption*{Median Error in Predictive Log-Likelihood of Held-Out Data}
\begin{tabular*}{\linewidth}{l@{\extracolsep{\fill}}cc}
\toprule
Benchmark                   & Kernel Density Estimation & Bayesian Synthesis \\
\midrule
\texttt{biasedtugwar}       & \num{-1.58e-01} & \num{-3.13e-02}  \\
\texttt{burglary}           & \num{-4.25e-01} & \num{-1.45e-03}  \\
\texttt{csi}                & \num{-1.00e-01} & \num{-4.35e-04}  \\
\texttt{easytugwar}         & \num{-2.96e-01} & \num{-9.96e-02}  \\
\texttt{eyecolor}           & \num{-4.69e-02} & \num{-5.31e-03}  \\
\texttt{grass}              & \num{-3.91e-01} & \num{-4.49e-02}  \\
\texttt{healthiness}        & \num{-1.35e-01} & \num{-3.00e-03}  \\
\texttt{hurricane}          & \num{-1.30e-01} & \num{-2.11e-04}  \\
\texttt{icecream}           & \num{-1.51e-01} & \num{-7.07e-02}  \\
\texttt{mixedCondition}     & \num{-1.06e-01} & \num{-1.43e-02}  \\
\texttt{multipleBranches}   & \num{-4.12e-02} & \num{-1.22e-02}  \\
\texttt{students}           & \num{-1.74e-01} & \num{-5.47e-02}  \\
\texttt{tugwarAddition}     & \num{-2.60e-01} & \num{-1.38e-01}  \\
\texttt{uniform}            & \num{-2.72e-01} & \num{-1.26e-01} \\ \bottomrule
\end{tabular*}
\end{subtable}\hfill%
\begin{subfigure}{.375\linewidth}
\includegraphics[width=\linewidth]{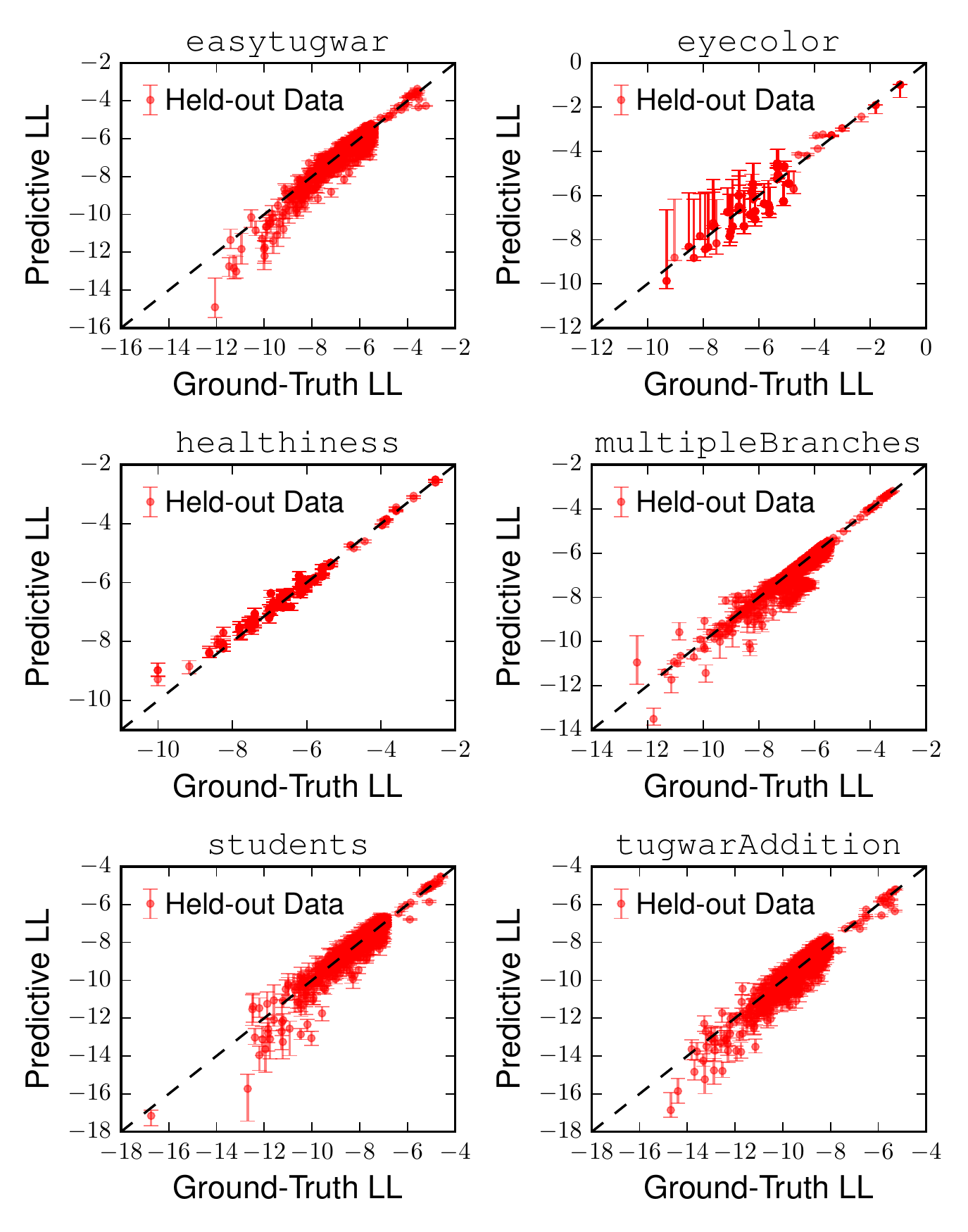}
\end{subfigure}%

\captionsetup{skip=5pt}
\caption{Comparing the error in the predictive probabilities of held-out
data according to Bayesian synthesis and KDE. For each of the 13 benchmark
problems, a training set of 10,000 data records was used to synthesize
probabilistic programs in our tabular data DSL and to fit KDE models. 10,000 new
data records were then used to assess held-out predictive probabilities, which
measures how well the synthesized probabilistic programs are able to model the
true distribution. The scatter plots on the right show the full distribution of
the true (log) probabilities (according to the ground-truth program) and the
predictive probabilities (according to the synthesized programs) for six of the
benchmarks (each red dot is a held-out data record). Estimates are most accurate
and certain in the bulk of the distribution and are less accurate and less
certain in the tails.}
\label{fig:crosscat-blog}
\end{figure}

\subsection{Quantitative Prediction Accuracy for Multivariate Tabular Data}

Another central goal of data analysis for multivariate tabular data is to learn
probabilistic models that can accurately predict the probability of new data
records, a task known as density estimation \citep{silverman1986,scott2009}.
Accurately predicting the probability of new data enables several important
tasks such as data cleaning, anomaly detection, and imputing missing data.
We obtain the probability of new data records according to the synthesized
probabilistic programs by first translating the programs into Venture (right
column of Figure~\ref{fig:dsl-crosscat}) and then executing the translated
programs in the Venture inference environment.
Given a new data record, these programs immediately return the probability of
the new data record according to the probabilistic model specified by the
Venture probabilistic program.
We compare our method to the widely-used multivariate kernel density estimation
(KDE) baseline with mixed data types from \citet{racine2004}.

We evaluate our method's ability to accurately perform density estimation
using 13 benchmark datasets adapted from~\citet{chasins2017}.
Each dataset was generated by a ``ground-truth'' probabilistic program written
in BLOG \cite{milch2007}. Because the ground-truth program is available, we can
compare the actual probability of each new data record (obtained from the
ground-truth BLOG program) with its predicted probability (obtained from the
synthesized Venture programs and from KDE).
The results in the table in Figure~\ref{fig:crosscat-blog} show that our method
is more accurate than KDE, often by several orders of magnitude.
A key advantage of our approach is that we synthesize an ensemble of
probabilistic programs for each dataset. This ensemble enables us to approximate
the full posterior distribution and then use this distribution to compute error
bars for predicted probability values. The scatter plots in
Figure~\ref{fig:crosscat-blog} show that the error bars are well-calibrated,
i.e.~error bars are wider for data records where our method gives less
accurate estimates.
In contrast, KDE only gives point estimates with no error bars and no
measure of uncertainty.


\section{Related Work}
\label{sec:related}

We discuss related work in five related fields: Bayesian synthesis of
probabilistic programs (where this work is the first), non-Bayesian synthesis of
probabilistic programs, probabilistic synthesis of non-probabilistic programs,
non-probabilistic synthesis of non-probabilistic programs, and model discovery
in probabilistic machine learning.

{\bf Bayesian synthesis of probabilistic programs.} This paper presents
the first Bayesian synthesis of probabilistic programs. It presents the first
formalization of Bayesian synthesis for probabilistic programs and the first
soundness proofs of Bayesian synthesis for probabilistic programs generated by
probabilistic context-free grammars. It also presents the first empirical
demonstration of accurate modeling of multiple real-world domains and tasks via
this approach.

This paper is also the first to: (i) identify sufficient conditions to obtain a
well-defined posterior distribution; (ii) identify sufficient conditions to
obtain a  sound Bayesian synthesis algorithm; (iii) define a general family of
domain-specific languages with associated semantics that ensure that the
required prior and posterior distributions are well-defined; (iv) present a
sound synthesis algorithm that applies to any language in this class of
domain-specific languages; and (v) present a specific domain-specific language
(the Gaussian process language for modeling time series data) that satisfies
these sufficient conditions.

\citet{nori15} introduce a system (PSKETCH) designed to complete partial
sketches of probabilistic programs for modeling tabular data. The technique is
based on applying sequences of program mutations to a base program written in a
probabilistic sketching language.  As we detail further below, the paper uses
the vocabulary of Bayesian synthesis to describe the technique but contains
multiple fundamental technical errors that effectively nullify its key claims.
Specifically, \citet{nori15} proposes to use Metropolis-Hastings sampling to
approximate the \mbox{\textit{maximum a posteriori}} solution to the sketching
problem. However, the paper does not establish that the prior or posterior
distributions on programs are well-defined. The paper attempts to use a uniform
prior distribution over an unbounded space of programs. However, there is no
valid uniform probability measure over this space. Because the posterior is
defined using the prior and the marginal likelihood of all datasets is not shown
to be finite, the posterior is also not well-defined. The paper presents no
evidence or argument that the proposed prior or posterior distribution is
well-defined.
The paper also asserts that the synthesis algorithm converges because the MH
algorithm always converges, but it does not establish that the presented
framework satisfies the properties required for the MH algorithm to converge.
And because the posterior is not well-defined, there is no probability
distribution to which the MH algorithm can converge. We further note that while
the paper claims to use the MH algorithm to determine whether a proposed program
mutation is accepted or rejected, there is in fact no tractable algorithm that
can be used to compute the reversal probability of going back from a proposed to
an existing program, which means the MH algorithm cannot possibly be used with
the program mutation proposals described in the paper.

The presented system, PSKETCH, is based on applying sequences of program
mutations to a base program written in a probabilistic sketching language with
constructs (such as if-then-else, for loops, variable assignment, and arbitrary
arithmetic operations) drawn from general-purpose programming languages. We
believe that, consistent with the experimental results presented in the paper,
working only with these constructs is counterproductive in that it produces a
search space that is far too unconstrained to yield practical solutions to
real-world data modeling problems in the absence of other sources of information
that can more effectively narrow the search.
We therefore predict that the field will turn to focus on probabilistic DSLs
(such as the ones presented in this paper) as a way to more effectively target
the problems that arise in automatic data modeling. As an example, while the
sketching language in \citet{nori15} contains general-purpose programming
constructs, it does not have domain-specific constructs that concisely represent
Gaussian processes, covariance functions, or rich nonparametric mixture models
used in this paper.

\citet{hwang2011} use beam search over arbitrary program text in a subset of the
Church~\citep{goodman08} language to obtain a generative model over tree-like
expressions. The resulting search space is so unconstrained that, as the authors
note in the conclusion, this technique does not apply to any real-world problem
whatsoever. This drawback highlights the benefit of controlling the search space
via an appropriate domain-specific language. In addition we note that although
\citet{hwang2011} also use the vocabulary of Bayesian synthesis, the
proposed program-length prior over an unbounded program space is not shown to be
probabilistically well-formed, nor is it shown to lead to a valid Bayesian
posterior distribution over programs, nor is the stochastic approximation of
the program likelihood shown to result in a sound synthesis algorithm.

{\bf Non-Bayesian synthesis of probabilistic programs.}
\citet{ellis2015} introduce a method for synthesizing probabilistic programs by
using SMT solvers to optimize the likelihood of the observed data with a
regularization term that penalizes the program length. The research works with a
domain-specific language for morphological rule learning.
\citet{perov2014} propose to synthesize code for simple random number
generators using approximate Bayesian computation.
These two techniques are fundamentally different from ours. They focus on
different problems, specifically morphological rule learning, visual concept
learning, and random number generators, as opposed to data modeling. Neither is
based on Bayesian learning and neither presents soundness proofs (or even states
a soundness property). Moreover, both attempt to find a single highest-scoring
program as opposed to synthesizing ensembles of programs that approximate a
Bayesian posterior distribution over the sampled probabilistic programs. As this
previous research demonstrates, obtaining soundness proofs or characterizing the
uncertainty of the synthesized programs against the data generating process is
particularly challenging for non-Bayesian approaches because of the ad-hoc
nature of the synthesis formulation.

\citet{tong2016} describe a method which uses an off-the-shelf model discovery
system \citep[ABCD]{gpss} to learn Gaussian process models and the code-generate
the models to Stan \citep{carpenter2015} programs. However, \citet{tong2016}
does not formalize the program synthesis problem, nor does it present any formal
semantics, nor does it show how to extract the probability that qualitative
properties hold in the data, nor does it apply to multiple model families in a
single problem domain let alone multiple problem domains.
\citet{chasins2017} present a technique for synthesizing probabilistic programs
in tabular datasets using if-else statements. Unlike our method, their approach
requires the user to manually specify a causal ordering between the variables.
This information may be difficult or impossible to obtain. Second, the proposed
technique is based on using linear correlation, which we have shown in our
experiments (Figures~\ref{fig:crosscat-linear} and \ref{fig:crosscat-structures})
fails to adequately capture complex probabilistic relationships. Finally, the
paper is based on simulated annealing, not Bayesian synthesis, has no
probabilistic interpretation, and does not claim to provide a notion of
soundness.

{\bf Probabilistic synthesis of non-probabilistic programs.}
\citet{schkufza2013} describe a technique for synthesizing high-performance X86
binaries by using MCMC to stochastically search through a space of programs and
to deliver a single X86 program which satisfies the hard constraint of
correctness and the soft constraint of performance improvement.
While both the Bayesian synthesis framework in this paper and the
superoptimization technique from \citet{schkufza2013} use MCMC algorithms to
implement the synthesis, they use MCMC in a fundamentally different way.
In \citet{schkufza2013}, MCMC is used as a strategy to search through a space of
deterministic programs that seek to minimize a cost function that itself has no
probabilistic semantics. In contrast, Bayesian synthesis uses MCMC as a strategy
to approximately sample from a space of probabilistic programs whose semantics
specify a well-defined posterior distribution programs.

Bayesian priors over structured program representations which seek to sample
from a posterior distribution over programs have also been investigated.
\citet{liang2010} use adapter grammars \citep{johnson2007} to build a
hierarchical nonparametric Bayesian prior over programs specified in combinatory
logic \citep{schonfinkel1924}.
\citet{ellis2016} describe a method to sample from a bounded program space with
a uniform prior where the posterior probability of a program is equal to
(i) zero if it does not satisfy an observed input-output constraint, or
(ii) geometrically decreasing in its length otherwise. These methods are used to
synthesize arithmetic operations, text editing, and list manipulation programs.
Both \citet{liang2010} and \citet{ellis2015} specify Bayesian priors over
programs similar to the prior in this paper. However,
these two techniques are founded on the assumption that the
synthesized programs have deterministic input-output behavior and cannot be
easily extended to synthesize programs that have probabilistic
input-output behavior. In contrast, the Bayesian synthesis framework presented in this
paper can synthesize programs with deterministic input-output behavior by adding
hard constraints to the $\Likelihood$ semantic function, although alternative
synthesis techniques may be required to make the synthesis more effective.

\citet{lee2018} present a technique for speeding up program synthesis of
non-probabilistic programs by using A* search to enumerate
programs that satisfy a set of input-output constraints in order of
decreasing prior probability. This prior distribution over programs is itself
learned using a probabilistic higher-order grammar with
transfer learning over a large corpus of existing synthesis problems and
solutions. The technique is used to synthesize programs in domain-specific
languages for bit-vector, circuit, and string manipulation tasks.
Similar to the Bayesian synthesis framework in this paper, \citet{lee2018} use
PCFG priors for specifying domain-specific languages.
However the fundamental differences are that the synthesized programs in
\citet{lee2018} are non-probabilistic and the objective is to enumerate valid
programs sorted by their prior probability, while in this paper the synthesized
programs are probabilistic so enumeration is impossible and the objective is
instead to sample programs according to their posterior probabilities.

{\bf Non-probabilistic synthesis of non-probabilistic programs.}
Over the last decade program synthesis has become a highly active area of
research in programming languages. Key techniques include
deductive logic with program transformations \citep{burstall1977,manna1979,manna1980},
genetic programming \citep{koza1992,koza1997},
solver and constraint-based methods \citep{lezama2006,jha2010,gulwani2011a,gulwani2011b,feser2015},
syntax-guided synthesis \citep{alur2013},
and neural networks \citep{graves2014,defreitas2016,balog2017,nampi2018}.
In general these approaches have tended to focus on areas where uncertainty is
not essential or even relevant to the problem being solved. In particular,
synthesis tasks in this field apply to programs that exhibit deterministic
input-output behavior in discrete problem domains with combinatorial solutions.
There typically exists an ``optimal'' solution and the synthesis goal is
to obtain a good approximation to this solution.

In contrast, in Bayesian synthesis the problem domain is fundamentally about
automatically learning models of non-deterministic data generating processes
given a set of noisy observations. Given this key characteristic of the problem
domain, we deliver solutions that capture uncertainty at two levels. First, our
synthesis produces probabilistic programs that exhibit noisy, non-deterministic
input-output behavior. Second, our technique captures inferential uncertainty
over the structure and parameters of the synthesized programs themselves by
producing an ensemble of probabilistic programs whose varying
properties reflect this uncertainty.

{\bf Model discovery in probabilistic machine learning.} Researchers
have developed several probabilistic techniques for discovering
statistical model structures from observed data. Prominent examples include
Bayesian network structures \citep{mansinghka2006},
matrix-composition models \citep{grosse2012},
univariate time series \citep{duvenaud2013},
multivariate time series \citep{saad-aistats-2018}, and
multivariate data tables \citep{mansinghka2016}.

Our paper introduces a general formalism that explains and extends these
modeling techniques.
First, it establishes general conditions for Bayesian synthesis to be
well-defined and introduces sound Bayesian synthesis algorithms that apply to a
broad class of domain-specific languages.

Second, we show how to estimate the probability that qualitative structures are
present or absent in the data. This capability rests on a distinctive aspect of
our formalism, namely that soundness is defined in terms of sampling programs
from a distribution that converges to the Bayesian posterior, not just finding a
single ``highest scoring'' program. As a result, performing Bayesian synthesis
multiple times yields a collection of programs that can be collectively
analyzed to form Monte Carlo estimates of posterior and predictive
probabilities.

Third, our framework shows one way to leverage probabilistic programming
languages to simplify the implementation of model discovery techniques.
Specifically, it shows how to translate programs in a domain-specific language
into a general-purpose language such as Venture, so that the built-in language
constructs for predictive inference can be applied (rather than requiring new
custom prediction code for each model family). A key advantage of this approach
is that representing probabilistic models as synthesized probabilistic programs
enable a wide set of capabilities for data analysis tasks \citep{saad2016,saad2017}.

Fourth, we implemented two examples of domain-specific languages and empirically
demonstrate accuracy improvements over multiple baselines.


\section{Conclusion}

We present a technique for Bayesian synthesis of probabilistic programs for
automatic data modeling. This technique enables, for the first time, users to
solve important data analysis problems without manually writing statistical
programs or probabilistic programs. Our technique (i) produces an ensemble of
probabilistic programs that soundly approximate the Bayesian posterior, allowing
an accurate reflection of uncertainty in predictions made by the programs; (ii)
is based on domain-specific modeling languages that are empirically shown to
capture the qualitative structure of a broad class of data generating processes;
(iii) processes synthesized programs to extract qualitative structures from the
surface syntax of the learned programs; (iv) translates synthesized programs
into probabilistic programs that provide the inference machinery needed to
obtain accurate predictions; and (v) is empirically shown to outperform baseline
statistical techniques for multiple qualitative and quantitative data analysis
and prediction problems.

Designing good probabilistic domain-specific languages and priors is a key
characteristic of our approach to automatic data modeling. In this paper we
focus on broad non-parametric Bayesian priors that are flexible enough to
capture a wide range of patterns and are also tractable enough for Bayesian
synthesis. In effect, the choice of prior is embedded in the language design
pattern, is done once when the language is designed, and can be reused across
applications.
The high-level nature of the DSL and carefully designed priors allow us to
restrict the space of probabilistic programs to those that are well-suited
for a particular data modeling task (such as time series and multivariate data
tables). Once an appropriate DSL is developed, the synthesis algorithms can then
be optimized on a per-DSL basis.
Future directions of research include developing optimized synthesis algorithms
to improve scalability; adding language support for user-specified qualitative
constraints that must hold in the synthesized programs; and information-flow
analyses on the synthesized programs to quantify relationships between variables
in the underlying data generating process.

\begin{acks}
This research was supported by the DARPA SD2 program (contract
\mbox{FA8750-17-C-0239}); grants from the MIT Media Lab, the Harvard Berkman
Klein Center Ethics and Governance of AI Fund, and the MIT CSAIL Systems That
Learn Initiative; as well as an anonymous philanthropic gift.
\end{acks}

\bibliography{paper}
\clearpage

\appendix
\newcommand{\CKx}{C(K,\mathbf{x})}

\section{Proofs for Gaussian Process Domain-Specific Language}
\label{appendix:gp-dsl-proofs}

Section~\ref{sec:dsl-time-series} of the main text states that the prior and
likelihood semantics of the Gaussian process domain-specific language satisfy
the theoretical preconditions required for Bayesian synthesis. In particular, we
prove the following lemma:

\begin{lemma}
The $\Prior$ and $\Likelihood$ semantic functions in Figure~\ref{fig:dsl-gp}
satisfy Conditions~\ref{cond:prior-normalized},
\ref{cond:likelihood-normalized}, and \ref{cond:likelihood-bounded}.
\end{lemma}

\begin{proof}[Proof for Condition~\ref{cond:prior-normalized}]
Letting upper-case symbols denote non-terminals $N$ and lower-case words in
teletype denote terminals, the equivalent context-free grammar in Chomsky
normal form is given by:
\begin{align*}
K &\to CB_{[0.14]}
  \mid EB_{[0.14]}
  \mid GB_{[0.14]}
  \mid IB_{[0.14]}
  \mid PB_{[0.14]}
  \mid TM_{[0.135]}
  \mid VM_{[0.135]}
  \mid WX_{[0.03]}\\
C &\to \texttt{const} \\
E &\to \texttt{wn} \\
G &\to \texttt{lin} \\
I &\to \texttt{se} \\
P &\to \texttt{per} \\
B &\to \texttt{gamma}\\
T &\to \texttt{+}\\
V &\to \texttt{*}\\
W &\to \texttt{cp}\\
M &\to KK\\
X &\to BM,
\end{align*}
where subscripts indicate the production probabilities for non-terminals with
more than one production rule . Following the notation from \citet[Equation
2]{gecse2010}, define:
\begin{align*}
m_{ij}  & = \textstyle\sum_{k=1}^{r_i} p_{ik}n_{ikj} && (i,j=1,\dots,n) \\
\mbox{where } \\
r_i     & : \mbox{ number of productions with } N_i \in N \mbox{ premise},\\
p_{ik}  & : \mbox{ probability assigned to production } k \mbox { of } N_i, \\
n_{ikj} & : \mbox{ number of occurrence of } N_j \in N \mbox{ in production }
  k \mbox{ of } N_i.
\end{align*}
The expectation matrix $M \Coloneqq [m_{ij}]$ is given by:
\begin{align*}\setcounter{MaxMatrixCols}{20}
\kbordermatrix{%
 & K & C & E & G & I & P & B & T & V & W & M & X \\
K & 0 & 0.14 & 0.14 & 0.14 & 0.14 & 0.14 & 0.70 & 0.135 & 0.135 & 0.03 & 0.27 & 0.03 \\
C & 0 & 0 & 0 & 0 & 0 & 0 & 0 & 0 & 0 & 0 & 0 & 0 \\
E & 0 & 0 & 0 & 0 & 0 & 0 & 0 & 0 & 0 & 0 & 0 & 0 \\
G & 0 & 0 & 0 & 0 & 0 & 0 & 0 & 0 & 0 & 0 & 0 & 0 \\
I & 0 & 0 & 0 & 0 & 0 & 0 & 0 & 0 & 0 & 0 & 0 & 0 \\
P & 0 & 0 & 0 & 0 & 0 & 0 & 0 & 0 & 0 & 0 & 0 & 0 \\
B & 0 & 0 & 0 & 0 & 0 & 0 & 0 & 0 & 0 & 0 & 0 & 0 \\
T & 0 & 0 & 0 & 0 & 0 & 0 & 0 & 0 & 0 & 0 & 0 & 0 \\
V & 0 & 0 & 0 & 0 & 0 & 0 & 0 & 0 & 0 & 0 & 0 & 0 \\
W & 0 & 0 & 0 & 0 & 0 & 0 & 0 & 0 & 0 & 0 & 0 & 0 \\
M & 2 & 0 & 0 & 0 & 0 & 0 & 0 & 0 & 0 & 0 & 0 & 0 \\
X & 0 & 0 & 0 & 0 & 0 & 0 & 1 & 0 & 0 & 0 & 1 & 0
}.
\end{align*}
The eigenvalues of $M$ are $0.78512481$, $-0.67128139$, and $-0.11384342$ and
they are all less than 1 in absolute value. The conclusion follows from
Theorem~\ref{thm:pcfg-consistency} in the main text.
\end{proof}

\begin{proof}[Proof for Condition~\ref{cond:likelihood-normalized}]
The semantic function $\Denotation[\Likelihood]{K}((\mathbf{x},\mathbf{y}))$ is
normalized since the expression represents a standard multivariate normal
distribution in $\mathbf{x}$, with covariance matrix $\CKx \Coloneqq
[c_{ij}]_{i,j=1}^{n}$, where $c_{ij} \Coloneqq \Denotation[\Covar]{K}(x_i)(x_j)$.
\end{proof}

\begin{proof}[Proof for Condition~\ref{cond:likelihood-bounded}]

We first establish the following Lemma:

\begin{lemma}\label{lem:minkowski-determinant}
If $A$ and $B$ be two $n\times{n}$ symmetric non-negative, positive
semi-definite matrices, then
\begin{align}
\abs{A + B} \ge \abs{A} + \abs{B},
\end{align}
where $\abs{\cdot}$ denotes the matrix determinant.
\end{lemma}

\begin{proof}
Minkowski's determinant inequality \citep[Theorem~4.1.8, p.115]{marcus1992}
states that for any non-negative $n\times{n}$ Hermitian matrices, we have:
\begin{align*}
(\abs{A + B})^{1/n} \ge \abs{A}^{1/n} + \abs{B}^{1/n}.
\end{align*}
Applying the binomial theorem gives
\begin{align*}
\abs{A + B}
  &= \sum_{k=0}^{n}\binom{n}{k} (\abs{A}^{1/n})^{n-k}(\abs{B}^{1/n})^{k} \\
  &= (\abs{A}^{1/n})^n + \dots + (\abs{B}^{1/n})^{n} \\
  &\ge (\abs{A}^{1/n})^n + (\abs{B}^{1/n})^{n} \\
  &= \abs{A} + \abs{B},
\end{align*}
where the inequality follows from the fact that $A$ and $B$ are positive
semi-definite, so that $\abs{A}$ and $\abs{B}$ are non-negative.
\end{proof}

Let $K \in \DomKernel$ be an expression in the Gaussian process domain-specific
modeling language, and define
\begin{align}
\begin{aligned}
\Denotation[\Likelihood]{K}((\mathbf{x}, \mathbf{y})) &=
  \exp\Big(-1/2
    \textstyle \sum_{i=1}^{n}y_i\left(
    \textstyle\sum_{j=1}^{n}(\lbrace
      [\Denotation[\Covar]{K}(x_i)(x_j) +
        0.01\delta(x_i,x_j)]_{i,j=1}^{n}
      \rbrace^{-1}_{ij})y_j
    \right) \\
  &\qquad
  -1/2\log\left\lvert{[\Denotation[\Covar]{K}(x_i)(x_j)
    + 0.01\delta(x_i,x_j)]_{i,j=1}^{n}}\right\rvert
  - (n/2)\log{2\pi}
  \Big)
\label{eq:lik-denotation}
\end{aligned}
\end{align}
to be the likelihood semantic function. We will show that there exists a finite
real number $c > 0$ such that
\begin{align*}
\forall\, \mathbf{x}, \mathbf{y} \in \mathbb{R}^{n\times 2}.\;
  0 < \sup
    \set{\Denotation[\Likelihood]{K}((\mathbf{x}, \mathbf{y}))
      \mid K \in \DomKernel}
    < c.
\end{align*}

Let $\CKx \Coloneqq [\Denotation[\Covar]{K}(x_i)(x_j)]_{i,j=1}^{n}$ be the
$n\times{n}$ covariance matrix defined by the DSL expression $K$ and input data
$\mathbf{x}$, $\mathbf{I}$ be the $n\times{n}$ identity matrix, $\sigma
\Coloneqq 0.01$ the minimum variance level, and $\abs{M}$ the determinant
of matrix $M$.

Then Eq~\eqref{eq:lik-denotation} can be written in matrix notation as:
\begin{align*}
\Denotation[\Likelihood]{K}((\mathbf{x}, \mathbf{y}))
&= \exp\left(-\frac{1}{2}
  \left(\mathbf{y}^{\intercal}
    (\CKx + \sigma\mathbf{I})^{-1} \mathbf{y}\right)
    -\frac{1}{2}\log\left(\lvert \CKx + \sigma\mathbf{I}\rvert\right)
    - (n/2)\log{2\pi}
\right)\\
&= ((2\pi)^n \lvert \CKx + \sigma\mathbf{I} \rvert)^{-1/2}
\exp\left(-\frac{1}{2}
  \left(\mathbf{y}^{\intercal}
    (\CKx + \sigma\mathbf{I})^{-1} \mathbf{y}\right) \right) \\
&\le ((2\pi)^n \abs{\CKx + \sigma\mathbf{I}}^{-1/2} \\
&\le \abs{\CKx + \sigma\mathbf{I}}^{-1/2} \\
&\le (\abs{\CKx} + \abs{\sigma\mathbf{I}})^{-1/2} \\
&\le (\abs{\sigma\mathbf{I}})^{-1/2} \\
&= (\sigma)^{-1/2}
\end{align*}
where the first inequality follows from the positive semi-definiteness of the
covariance matrix $\CKx + \sigma \mathbf{I}$ and the identity
$e^{-z} < 1$ $(z > 0)$; the second inequality from $(2\pi)^{-n/2} < 1$; the
third inequality from Lemma~\ref{lem:minkowski-determinant};
and the final inequality from the positive semi-definiteness of $\CKx$.

Since $K \in \DomKernel$ was arbitrary and the upper bound $\sigma^{-1/2}$ is
independent of $K$, the conclusion follows.
\end{proof}


\end{document}